\documentclass[a4paper]{article}

\usepackage{amsthm,amsmath,amssymb,amsfonts}
\usepackage[american]{babel}
\usepackage{graphicx}
\usepackage{microtype} 
\usepackage{hyperref}
\usepackage{cleveref}
\usepackage[url=false,giveninits=true,maxbibnames=99,style=alphabetic]{biblatex}
\usepackage{mathtools}
\usepackage{enumitem}
\usepackage{csquotes}
\usepackage[noend]{algpseudocode}
\usepackage{algorithm}
\usepackage{xparse}
\usepackage{color}
\usepackage{caption}
\usepackage{subcaption}
\usepackage{fullpage}
\usepackage{placeins}
\usepackage{xcolor}
\usepackage{makecell}
\usepackage{qcircuit}
\usepackage{thmtools}
\usepackage{thm-restate}

\mathchardef\mhyphen="2D 

\newcommand\newmathabbrev[2]{\newcommand{#1}{\ensuremath{#2}}}

\newcommand\cfont\mathrm
\newmathabbrev\p{\cfont{P}}
\newmathabbrev{\N}{\mathbb N}
\newmathabbrev\NP{\cfont{NP}}
\newmathabbrev\DTIME{\cfont{DTIME}}
\newmathabbrev\tSAT{3\cfont{\mhyphen{}SAT}}
\newmathabbrev\MA{\cfont{MA}}
\newmathabbrev\AM{\cfont{AM}}
\newmathabbrev\NPDAG{\cfont{NP\mhyphen{}DAG}}
\newmathabbrev\QMADAG{\cfont{QMA\mhyphen{}DAG}}
\newmathabbrev\yes{\mathrm{yes}}
\newmathabbrev\no{\mathrm{no}}
\newmathabbrev\US{\cfont{US}}
\newmathabbrev\FP{\cfont{FP}}
\newmathabbrev\PP{\cfont{PP}}
\newmathabbrev\CeP{\cfont{C_=P}}
\newmathabbrev\coCeP{\cfont{coC_=P}}
\newmathabbrev\PH{\cfont{PH}}
\newmathabbrev\SAT{\cfont{SAT}}
\newmathabbrev\SPP{\cfont{SPP}}
\newmathabbrev\GapP{\cfont{GapP}}
\newmathabbrev\BQP{\cfont{BQP}}
\newmathabbrev\QP{\cfont{QP}}
\newmathabbrev\StoqMA{\cfont{StoqMA}}
\newmathabbrev\coNP{\cfont{coNP}}
\newmathabbrev\AzPP{\cfont{A_0PP}}
\newmathabbrev\QMA{\cfont{QMA}}
\newmathabbrev\coQMA{\cfont{coQMA}}
\newmathabbrev\BPP{\cfont{BPP}}
\newmathabbrev\QCMA{\cfont{QCMA}}
\newmathabbrev\pNPlog{\p^{\NP[\log]}}
\newmathabbrev\pNP{\p^{\NP}}
\newmathabbrev\pNPtwo{\p^{\NP[2]}}
\newmathabbrev\pNPone{\p^{\NP[1]}}
\newmathabbrev\pParSAT{\p^{||\SAT}}
\newmathabbrev\pQMApar{\p^{||\QMA}}
\newmathabbrev\pCpar{\p^{||\C}}
\newmathabbrev\pStoqMApar{\p^{||\StoqMA}}
\newmathabbrev\pQMAlog{\p^{\QMA[\log]}}
\newmathabbrev\pClog{\p^{\textup{C}[\log]}}
\newmathabbrev\pC{\p^{\textup{C}}}
\newmathabbrev\QMASPACE{\cfont{QMASPACE}}
\newmathabbrev\pQMAtlog{\p^{\QMA(2)[\log]}}
\newmathabbrev\pStoqMAlog{\p^{\StoqMA[\log]}}
\newmathabbrev\pQMApt{\p^{\Vert\QMA(2)}}
\newmathabbrev\pQMA{\p^{\QMA}}
\newmathabbrev\SharpP{\cfont{\#P}}
\newmathabbrev\pSharP{\p^{\SharpP[1]}}
\newmathabbrev\PromisePP{\cfont{PromisePP}}
\newmathabbrev\lett{\le_\mathrm{tt}}
\newmathabbrev\YES{\mathsf{YES}}
\newmathabbrev\NO{\mathsf{NO}}
\newmathabbrev\PSPACE{\cfont{PSPACE}}
\newmathabbrev\IP{\cfont{IP}}
\newmathabbrev\POLY{\cfont{POLY}}
\newmathabbrev\DAG{\cfont{DAG}}
\newmathabbrev\StoqMADAG{\StoqMA\mhyphen\cfont{DAG}}
\newmathabbrev\CDAG{C\mhyphen\cfont{DAG}}
\newmathabbrev\CDAGf{C\mhyphen\cfont{DAG}_f}
\newmathabbrev\CDAGs{C\mhyphen\cfont{DAG}_s}
\newmathabbrev\CDAGd{C\mhyphen\cfont{DAG}_{d}}
\newmathabbrev\CDAGo{C\mhyphen\cfont{DAG}_1}
\newmathabbrev\LOGS{\cfont{LOGS}}
\newmathabbrev\TAUT{\cfont{TAUTOLOGY}}
\newmathabbrev\SBQP{\cfont{SBQP}}
\newmathabbrev\Fc{F_\coNP}
\newmathabbrev\Fa{F_\AzPP}
\newmathabbrev\GSCON{\cfont{GSCON}}
\newmathabbrev\GSCONexp{\GSCON_\cfont{exp}}
\newmathabbrev\QMAexp{\QMA_\cfont{exp}}
\newmathabbrev\UQMA{\cfont{UQMA}}
\newmathabbrev\R{\mathbb R}
\newmathabbrev\Trees{\cfont{TREES}}
\newmathabbrev\apxsim{\cfont{APX\mhyphen{}SIM}}
\newmathabbrev\AWPP{\cfont{AWPP}}
\newmathabbrev\X{\mathcal{X}}
\newmathabbrev\Y{\mathcal{Y}}

\newmathabbrev\Z{\mathcal{Z}}
\newmathabbrev\ZZ{\mathbb{Z}}
\newmathabbrev\Hprop{H_\mathrm{prop}}
\newmathabbrev\Hin{H_\mathrm{in}}
\newmathabbrev\Hout{H_\mathrm{out}}
\newmathabbrev\Hstab{H_\mathrm{stab}}
\newmathabbrev\Lext{\L_\mathrm{ext}}
\newmathabbrev\BTWNP{\cfont{BTW}(\NP)}
\newmathabbrev\BSN{\cfont{BSN}}
\newmathabbrev\SN{\cfont{SN}}
\newmathabbrev\BD{\cfont{BD}}
\newmathabbrev\HYPERTREE{\cfont{NP\mhyphen{}HYPERTREE}}
\newmathabbrev\Hext{H_\mathrm{ext}}
\newmathabbrev\Hpropt{\tilde{H}_\mathrm{prop}}
\newmathabbrev\Hint{\tilde{H}_\mathrm{in}}
\newmathabbrev\Houtt{\tilde H_\mathrm{out}}
\newmathabbrev\EXP{\cfont{EXP}}
\newmathabbrev\A{\mathcal{A}}
\newmathabbrev\U{\mathcal{U}}

\newcommand\kLH{k\cfont{\mhyphen{}LH}}
\renewcommand\L{\mathcal{L}}

\newcommand\B{\mathcal B}
\newcommand\Piacc{\Pi_\mathrm{acc}}

\newmathabbrev\DAGSSAT{\DAGS(\SAT)}
\newmathabbrev\DAGS{\mathrm{DAGS}}
\newmathabbrev\DAGSNP{\DAGS(\NP)}

\newcommand\psiin{\psi_\mathrm{in}}
\newmathabbrev\AND{\cfont{AND}}

\newmathabbrev\STCONN{{S,T}\cfont{\mhyphen{}CONN}}
\newmathabbrev\CNF{\cfont{CNF}}
\newmathabbrev\NEXP{\cfont{NEXP}}
\newmathabbrev\NPSPACE{\cfont{NPSPACE}}
\newmathabbrev\QCMASPACE{\cfont{QCMASPACE}}
\newmathabbrev\BQPSPACE{\cfont{BQPSPACE}}
\newmathabbrev{\PCP}{\cfont{PCP}}
\newmathabbrev\BQUPSPACE{\cfont{BQ_UPSPACE}}

\protected\def\verythinspace{%
  \ifmmode
    \mskip0.5\thinmuskip
  \else
    \ifhmode
      \kern0.08334em
    \fi
  \fi
}

\newcommand{\bin}{\{0,1\}}
\newcommand{\Piy}{\Pi_\mathrm{yes}}
\newcommand{\Pin}{\Pi_\mathrm{no}}
\newcommand{\Pii}{\Pi_\mathrm{inv}}
\newcommand{\Gammay}{\Gamma_\mathrm{yes}}
\newcommand{\Gamman}{\Gamma_\mathrm{no}}

\newcommand{\C}{\mathbb C}

\newcommand{\desc}{\mathrm{Desc}}
\newcommand{\anc}{\mathrm{Anc}}
\newcommand{\be}{\begin{equation}}
\newcommand{\ee}{\end{equation}}
\newcommand{\QV}{\mathcal{QV}}
\newcommand{\QVp}{\mathcal{QV}^+}

\renewcommand{\epsilon}{\varepsilon}
\DeclareMathOperator*{\argmax}{arg\,max}

\DeclareMathOperator{\Herm}{Herm}

\DeclareMathOperator{\Tr}{Tr}

\DeclareMathOperator{\indeg}{indeg}
\DeclareMathOperator{\outdeg}{outdeg}
\DeclareMathOperator{\depth}{depth}
\DeclareMathOperator{\level}{level}

\DeclareMathOperator{\Null}{Null}

\DeclareMathOperator{\suc}{children}
\DeclareMathOperator{\pred}{parents}

\DeclareMathOperator\outputc{out-wire}
\DeclareMathOperator\width{width}
\DeclareMathOperator\tw{tw}
\DeclareMathOperator\s{s}

\DeclareMathOperator\negl{negl}
\DeclareMathOperator\inv{preimage}

\newcommand\hermp[1]{\Herm\left(#1\right)}
\newcommand\lmin{\lambda_{\mathrm{min}}}

\newcommand\lminp[1]{\lmin\left(#1\right)}

\newcommand{\set}[1]{{\left\{#1\right\}}}    

\newcommand{\pout}{p_{\textup{out}}}

\newcommand{\fix}{\textup{fix}}

\DeclareMathOperator{\herm}{Herm}

\DeclareMathOperator{\poly}{poly}
\DeclareMathOperator{\polylog}{polylog}

\DeclarePairedDelimiter\bra{\langle}{\rvert}
\DeclarePairedDelimiter\ket{\lvert}{\rangle}

\DeclarePairedDelimiter\abs{\lvert}{\rvert}
\DeclarePairedDelimiter\norm{\lVert}{\rVert}
\DeclarePairedDelimiter\enorm{\lVert}{\rVert_2}

\DeclarePairedDelimiter\trnorm{\lVert}{\rVert_{\mathrm{tr}}}
\DeclarePairedDelimiterX\braket[2]{\langle}{\rangle}{#1 \delimsize\vert #2}
\DeclarePairedDelimiterX\ketbra[2]{\lvert}{\rvert}{#1 \delimsize\rangle\delimsize\langle #2}

\newcommand{\braketb}[2]{\bra{#1}#2\ket{#1}}

\setlist[itemize]{noitemsep, topsep=0pt}
\setlist[enumerate]{noitemsep, topsep=0pt}

\declaretheorem[numberwithin=section]{theorem}
\declaretheorem[sibling=theorem]{corollary}
\declaretheorem[sibling=theorem]{lemma}

\declaretheorem[sibling=theorem]{theme}

\declaretheorem[sibling=theorem,style=definition]{definition}
\declaretheorem[sibling=theorem,style=definition]{remark}
\declaretheorem[sibling=theorem,name=Rule,style=definition]{rul}

\newcommand{\spre}{S_{\textup{pre}}}
\newcommand{\spost}{S_{\textup{post}}}
\newcommand{\f}{f} 
\newcommand{\g}{g} 
\newcommand{\psiout}{\psi_{\textup{out}}}
\newcommand{\phiin}{\phi_{\textup{in}}}

\newcommand{\wm}{w_{1\cdots m}}

\newcommand{\pgood}{p_{\textup{good}}}
\newcommand{\hin}{H_{\textup{in}}}
\newcommand{\hprop}{H_{\textup{prop}}}
\newcommand{\hstab}{H_{\textup{stab}}}

\newcommand\Hw{H_\mathrm{w}}

\newcommand{\spa}[1]{\mathcal{#1}}

\newcommand{\hout}{H_{\textup{out}}}

\newcommand{\pbad}{p_{\textup{bad}}}

\newcommand{\nn}{N}

\newcommand{\psihist}{\psi_{\textup{hist}}}

\newcommand{\qflag}{q_{\textup{flag}}}

\newcommand{\unitary}[1]{\textup{U}\left(#1\right)}

\newcommand{\COs}[1]{\Call{ComputeOutput*}{#1}}

\makeatletter
\newcommand{\subalign}[1]{%
  \vcenter{%
    \Let@ \restore@math@cr \default@tag
    \baselineskip\fontdimen10 \scriptfont\tw@
    \advance\baselineskip\fontdimen12 \scriptfont\tw@
    \lineskip\thr@@\fontdimen8 \scriptfont\thr@@
    \lineskiplimit\lineskip
    \ialign{\hfil$\m@th\scriptstyle##$&$\m@th\scriptstyle{}##$\hfil\crcr
      #1\crcr
    }%
  }%
}
\NewDocumentCommand{\LeftComment}{s m}{%
  \Statex \IfBooleanF{#1}{\hspace*{\ALG@thistlm}}\(\triangleright\) #2}

\def\moverlay{\mathpalette\mov@rlay}
\def\mov@rlay#1#2{\leavevmode\vtop{%
   \baselineskip\z@skip \lineskiplimit-\maxdimen
   \ialign{\hfil$\m@th#1##$\hfil\cr#2\crcr}}}
\newcommand{\charfusion}[3][\mathord]{
    #1{\ifx#1\mathop\vphantom{#2}\fi
        \mathpalette\mov@rlay{#2\cr#3}
      }
    \ifx#1\mathop\expandafter\displaylimits\fi}

\makeatother

\newcommand{\cupdot}{\charfusion[\mathbin]{\cup}{\cdot}}
\newcommand{\bigcupdot}{\charfusion[\mathop]{\bigcup}{\cdot}}

\algnewcommand{\LineComment}[1]{\State \(\triangleright\) #1}

\algblockdefx[ON]{Blk}{EndBlk}[1]
  {#1}
  {}

\makeatletter
\ifthenelse{\equal{\ALG@noend}{t}}%
  {\algtext*{EndBlk}}
  {}%
\makeatother

\AtEveryBibitem{%
  \clearlist{language}%
}

\usepackage[onehalfspacing]{setspace}

\title{On polynomially many queries to NP or QMA oracles}
\author{Sevag Gharibian\footnote{Department of Computer Science and Institute for Photonic Quantum Systems (PhoQS), Paderborn University, Germany. Email: \{sevag.gharibian, dorian.rudolph\}@upb.de.} ~and Dorian Rudolph\footnotemark[1] }

\begin{document}

\maketitle

\begin{abstract}
    We study the complexity of problems solvable in deterministic polynomial time with access to an \NP{} or Quantum Merlin-Arthur (\QMA)-oracle, such as $\pNP$ and $\pQMA$, respectively.
    The former allows one to classify problems more finely than the Polynomial-Time Hierarchy (\PH), whereas the latter characterizes physically motivated problems such as Approximate Simulation (\apxsim{}) [Ambainis, CCC 2014].
    In this area, a central role has been played by the classes $\pNPlog$ and $\pQMAlog$, defined identically to $\pNP$ and $\pQMA$, except that only \emph{logarithmically} many oracle queries are allowed. Here, [Gottlob, FOCS 1993] showed that if the adaptive queries made by a $\pNP$ machine have a ``query graph'' which is a tree, then this computation can be simulated in $\pNPlog$.

    In this work, we first show that for any verification class $C\in\{\NP, \MA, \QCMA, \QMA, \QMA(2),\allowbreak \NEXP, \QMAexp\}$, any $\pC$ machine with a query graph of ``separator number'' $s$ can be simulated using deterministic time $\exp(s\log n)$ and $s\log n$ queries to a $C$-oracle.
    When $s\in O(1)$ (which includes the case of $O(1)$-treewidth, and thus also of trees), this gives an upper bound of $\pClog$, and when $s\in O(\log^k(n))$, this yields bound $\QP^{C[\log^{k+1}]}$ ($\QP$ meaning quasi-polynomial time).
    We next show how to combine Gottlob's ``admissible-weighting function'' framework with the ``flag-qubit'' framework of [Watson, Bausch, Gharibian, 2020], obtaining a unified approach for embedding $\pC$ computations directly into \apxsim{} instances in a black-box fashion.
    Finally, we formalize a simple no-go statement about polynomials (c.f. [Krentel, STOC 1986]): Given a multi-linear polynomial $p$ specified via an arithmetic circuit, if one can ``weakly compress'' $p$ so that its optimal value requires $m$ bits to represent, then $\pNP$ can be decided with only $m$ queries to an NP-oracle.
\end{abstract}

\section{Introduction}\label{scn:intro}

The celebrated Cook-Levin Theorem~\cite{Cook1971,Levin1973} and Karp's 21 NP-complete problems~\cite{Karp1972} laid the groundwork for the theory of NP-completeness to become the \emph{de facto} ``standard'' for characterizing ``hard'' problems. Indeed, in the decades since, hundreds of decision problems have been identified as NP-complete (see, e.g.,~\cite{GJ79}). Yet, despite the success of this theory, it soon became apparent that finer characterizations were needed to capture the complexity of certain hard problems.

In this direction, Stockmeyer \cite{Stockmeyer1976} defined the \emph{Polynomial Hierarchy} ($\PH$), of which the second level will interest us here.
Specifically, one may consider $\Sigma_2^\p=\NP^{\NP}$ (i.e. an NP-machine with access to an NP-oracle) or $\Delta^\p_2=\p^\NP$ (i.e. a P machine with access to an NP-oracle).
Here, our focus is on the latter, defined as the set of decision problems solvable by a deterministic polynomial-time Turing machine making polynomially many queries to an oracle for (say) SAT.
Like NP, $\p^\NP$ has natural complete problems, such as that shown by Krentel \cite{Krentel1992}:
Given Boolean formula $\phi:\set{0,1}^n\mapsto\set{0,1}$, does the lexicographically largest satisfying assignment $x_1\cdots x_n$ of $\phi$ have $x_n=1$?

\paragraph{Restricting the number of NP queries.} In 1982, in pursuit of yet finer characterizations, Papadimitriou and Zachos~\cite{Papadimitriou1982} asked: What happens if one considers problems ``slightly harder'' than \NP, i.e. solvable by a P machine making only \emph{logarithmically} many queries to an NP-oracle? This class, denoted $\pNPlog{}$, contains both NP and co-NP (since the P machine can \emph{postprocess} the answer of the NP-oracle by negating said answer), and is thus believed strictly harder than NP. The following decade saw a flurry of activity on this topic (see \Cref{sscn:relatedwork}); for example, Wagner \cite{Wagner1987,Wagner1988} showed that deciding if the optimal solution to a MAX-$k$-SAT instance has even Hamming weight is $\pNPlog{}$-complete.

This led to the natural question: Is $\pNPlog=\pNP$?
If one restricts the $\pNP$ machine to make all NP queries in \emph{parallel} (i.e. non-adaptively), denoted $\p^{\Vert\NP}$, then
Hemachandra \cite{Hemachandra1989} and Buss and Hay \cite{Buss1991} have shown $\p^{\Vert\NP}=\pNPlog{}$.
Thus, adaptivity appears crucial; so, Gottlob~\cite{Gottlob1995} next allowed dependence between queries as follows:
One may view $\p^\NP$ as a directed acyclic graph (DAG), whose nodes represent NP queries, and directed edge $(u,v)$ indicates that query $v$ depends on the answer of query $u$.
Denote this as the ``query graph'' of the $\pNP$ computation (\Cref{def:c-dag}).
In 1995, Gottlob showed that any $\pNP$ computation whose query graph is a tree can be simulated in $\pNPlog{}$.
To the best of our knowledge, this is the current state of the art regarding $\pNP$ versus $\pNPlog$.

\paragraph{Developments on the quantum side.} A few years later, the complexity theoretic study of ``quantum constraint satisfaction problems'' began in 1999 with Kitaev \enquote{quantum Cook-Levin theorem}~\cite{Kitaev2002}, which states that the problem of estimating the ``ground state energy'' of a local Hamiltonian ($k$-LH) is complete for Quantum Merlin Arthur ($\QMA$, a quantum generalization of NP). Particularly appealing is the fact that $k$-LH is physically motivated: It encodes the problem of estimating the energy of a quantum system when cooled to its lowest energy configuration.

More formally, $k$-LH generalizes the problem MAX-$k$-SAT, and is specified as follows.
As input, we are given a (succinct) description of a Hermitian matrix $H = \sum_i H_i\in\C^{2^n\times 2^n}$, where each Hermitian $H_i$ is a local ``quantum clause'' acting non-trivially on at most $k$ qubits (out of the full $n$-qubit system).
The \emph{ground state} (i.e. optimal assignment) is then the eigenvector of $H$ with the smallest eigenvalue, which we call the \emph{ground state energy} (i.e. optimal assignment's value).
Thus, understanding the low temperature properties of a many-body system is ``simply'' an eigenvalue problem for some succinctly described exponentially large matrix $H$.
Since Kitaev's work, a multitude of other physical problems have been shown to be $\QMA$-complete (see, e.g., surveys \cite{Osborne2012a, Bookatz2014,Gharibian2015b}).\\

\noindent \emph{The formalisation of $\pQMAlog$.} In 2014, Ambainis tied the study of QMA and $\pNPlog$ together by discovering the first $\pQMAlog$-complete problem ($\pQMAlog$ is defined as $\pNP$, but with the NP-oracle replaced with a QMA-oracle): \emph{Approximate Simulation ($\apxsim{}$)}.
To define APX SIM, suppose we wish to simulate the experiment of cooling down a quantum many-body system, and then performing a local measurement so as to extract information about the ground state's properties.
Formalized (roughly) as a decision problem, we must decide, given Hamiltonian $H$ describing the system, observable $A$ describing a local measurement, and inverse polynomially gapped thresholds $\alpha$ and $\beta$, whether there exists a ground state $\ket\psi$ of $H$ with expected value $\braketb\psi A$ below $\alpha$.

For context, APX-SIM can be viewed as a quantum variation of Wagner's $\pNPlog{}$-complete problem above~\cite{Wagner1987,Wagner1988} (does the optimal solution to a MAX-SAT instance have even Hamming weight?), since both problems ask about properties of optimal solutions to quantum and classical constraint satisfaction problems, respectively.
However, in the quantum setting, APX-SIM has the additional perk of being strongly physically motivated.
This is because often in practice, one is not interested in the ground state energy, but in properties of the \emph{ground state itself} (e.g. does it exhibit certain quantum phenomena? When does it undergo a phase transition?)~\cite{Gharibian2015b}. 
APX-SIM models the ``simplest''  experiment for computing such ground state properties, making no assumptions about additional information the experimenter might \emph{a priori} have.
(For example, in APX-SIM, although the goal is to probe the ground state of $H$, one is \emph{not} given the corresponding ground state energy as input. This is crucial, both complexity theoretically\footnote{If the definition of APX-SIM were to be modified so that the ground state energy of $H$ was given as part of the input, then APX-SIM would be QMA-complete instead of $\pQMA$-complete. This is because once one knows the ground state energy, a single QMA query and no postprocessing suffices to answer APX-SIM.} and physically, since in practice an experimenter does not \emph{a priori} know the ground state energy, as it is QMA-complete to compute to begin with!)

\paragraph{$\pQMAlog$ versus $\pQMA$ and this paper.} This sets up the question inspiring the current work --- is $\pQMAlog=\pQMA$?
In 2020, Gharibian, Piddock, and Yirka~\cite{Gharibian2019b} showed that $\pQMAlog{} = \p^{\Vert\QMA{}}$, for  $\p^{\Vert\QMA{}}$ defined as $\p^{\Vert\NP{}}$ but with an NP-oracle.
This gave a quantum analogue of $\p^{\Vert\NP}=\pNPlog{}$~\cite{Hemachandra1989,Buss1991}, although it required completely different proof techniques\footnote{The roadblock quantumly is that unlike NP, QMA is a class of \emph{promise problems}. Thus, one must account for the possibility that a (say) $\pQMAlog$ machine may make ``invalid'' queries, i.e. those violating the promise of the QMA-oracle. A general survey covering such issues regarding promise problems is ~\cite{G06}.}.
In this paper, we thus set our sights on the next step: Gottlob's work on $\pNP$ computations with trees as query graphs~\cite{Gottlob1995}.
What we are able to achieve is not just a quantum analogue of~\cite{Gottlob1995}, but a significant strengthening in multiple directions for both NP and QMA: Our main result considers query graphs of \emph{bounded separator number} (which includes bounded treewidth, and hence trees), applies to a host of verification classes including NP and QMA, and gives non-trivial (quasi-polynomial) upper bounds even beyond the bounded separator number case.
Along the way, we show how to combine the techniques used with the existing work on APX-SIM and $\pQMAlog$, yielding a unified framework for mapping $\pQMA$-type problems directly to APX-SIM instances.

\subsection{Our results}\label{sscn:results}

To state our results, define (formal definitions in \Cref{sec:preliminaries})
\begin{eqnarray}
    \QV &:=& \{\NP, \MA, \QCMA, \QMA, \QMA(2), \NEXP, \QMAexp\},\\
    \QVp &:=& \QV \cup \{\StoqMA\}.
\end{eqnarray}
This is the set of classical and quantum verification classes for which our results will be stated.
However, our framework applies in principle to verification classes $C$ beyond these sets; the main properties we require are for $C$ to allow promise gap amplification\footnote{Amplification here means that $C$ with constant promise gap (difference between completeness and soundness parameters) is equal to $C$ with $1/\poly$ gap.} and classical preprocessing before verification.

Recall now that an NP query graph is a DAG encoding an arbitrary $\pNP$ computation, where nodes correspond to NP queries; denote this an NP-DAG.
Replacing NP with any $C\in\QVp$, we arrive at the notion of a C-DAG (\Cref{def:c-dag}).
As expected, deciding whether a given C-DAG corresponds to an accepting $\pC$ computation is itself a $\pC$-complete problem (\Cref{l:pC}).
To thus obtain new upper bounds on $\pC$ computations, in this work, we parameterize a given C-DAG via its \emph{separator number}, $s$.

Briefly, a graph $G=(V,E)$ on $n$ vertices has a \emph{separator} of size $s(n)$ if there exists a set of at most $s(n)$ vertices whose removal splits the graph into at least two (non-empty) connected components (\Cref{def:separator-number}). $G$ has \emph{separator number}~\cite{Gruber2012} $s(n)$ if, (1) {for all} subsets $Q\subseteq V$, the vertex-induced graph on $Q$ has a separator of size at most $s(n)$, and (2) $s(n)$ is the smallest number for which this holds. Denote by $\CDAGs$ a C-DAG of separator number $s$, where we write $\CDAG_1$ for the case of $s\in O(1)$.
Note that treewidth upper bounds separator number~\cite{Gruber2012}.

\paragraph{1. Deciding C-DAGs.} Our main result is the following. For clarity, by ``deciding'' a C-DAG, we mean deciding whether it encodes an accepting or rejecting $\pC$ computation.
\begin{restatable}{theorem}{theoremBsnS}\label{thm:bsn-s}
    Fix any $C\in\QV$ and efficiently computable function $s:\N\to\N$. Then,
    \begin{equation}\label{eqn:upper}
        \CDAGs\in  \DTIME\left(2^{O(s(n)\log n)}\right)^{C[s(n)\log n]},
    \end{equation}
    for $n$ the number of nodes in $G$.
\end{restatable}
\noindent In words, any $\pC$ computation with a query graph of separator number $s$ can be simulated by a classical deterministic Turing machine running in time $2^{O(s(n)\log n)}$ and making $s(n)\log n$ queries to a $C$-oracle.
With \Cref{thm:bsn-s} in hand, we are able to obtain the following sequence of results.

First, by setting $s=O(1)$, we significantly strengthen Gottlob's \cite{Gottlob1995} $\Trees(\NP) = \pNPlog{}$ result to the constant separator number case and broad range of verification classes $C$:
\begin{restatable}{theorem}{theoremBsn}\label{thm:bsn}
    For any $C\in\QV$, $\CDAG_1$ is $\p^{C[\log]}$-complete.
\end{restatable}
\noindent In words, any $\pC$ computation with a query graph of constant separator number is decidable in $\pClog$.

Second, an advantage of \Cref{thm:bsn-s} is that it scales with \emph{arbitrary} $s(n)$. Thus, to our knowledge, we obtain the first upper bounds for $\pC$ involving \emph{quasi}-polynomial resources:
\begin{restatable}{corollary}{corollaryQP}\label{cor:qp}
    For all integers $k\geq 1$ and $C\in\QV$, $\CDAG_{\log^k}\in \QP^{C[\log^{k+1}(n)]}$, where QP denotes quasi-polynomial time (\Cref{def:qp}).
\end{restatable}
\noindent In words, any $\pC$ computation with a query graph of polylogarithmic separator number is decidable in quasi-poly-time with polylog $C$-queries.
In general, $s(n)$ may scale as $O(n)$, in which case \Cref{thm:bsn-s} does not yield a non-trivial bound. Whether this can be improved is left as an open question (\Cref{sscn:open}).

Third, an example of a verification class which is \emph{not} known to satisfy promise gap amplification is \StoqMA{} (see, e.g.,~\cite{Aharonov2020}). Here, we also obtain non-trivial bounds, albeit weaker ones:
\begin{restatable}{theorem}{theoremBsnSStoqma}\label{thm:bsn-s-stoqma}
    Fix $C=\StoqMA$ and any efficiently computable function $s:\N\to\N$. Then,
    \begin{equation}\label{eqn:upper-stoqma}
        \CDAGs\in  \DTIME\left(2^{O(s(n)\log^2 n)}\right)^{C[s(n)\log^2 n]}.
    \end{equation}
\end{restatable}
\noindent 
Note the extra log factor in the exponents --- this prevents \Cref{thm:bsn-s-stoqma} from recovering result $\p^{\Vert \StoqMA} = \p^{\StoqMA[\log]}$~\cite{Gharibian2019b} ($\p^{\Vert \StoqMA}$ corresponds to a $\StoqMADAG$ with $s(n)=1$). Nevertheless, we \emph{do} recover and improve on~\cite{Gharibian2019b} when we instead consider the case of bounded \emph{depth} query graphs next.

Finally, Gottlob \cite{Gottlob1995} also studied query graphs of bounded \emph{depth}.
The next theorem is an extension of his result.
We define $\CDAGd$ as $\CDAGs$, except now we consider query DAGs of \emph{depth} (\Cref{def:level}) at most $d$ (as opposed to separator number $s$).
\begin{restatable}{theorem}{theoremConstantDepth}\label{thm:constant-depth}
    Let $d:\N\to\N$ be an efficiently computable function. For $C\in\{\NP,\NEXP, \QMAexp\}$,
    $
        \CDAGd \subseteq \p^{C[d(n)\log(n)]},
    $
    and for $C\in\QVp$,
    \[
        \CDAGd \subseteq \DTIME\left(2^{O(d(n)\log(n))}\right)^{C[d(n)\log(n)]}.
    \]
\end{restatable}
\noindent Using this, we obtain that deciding a $\pC$ computation with a query graph of constant depth is $\pClog$-complete (\Cref{cor:constant-depth}). This modestly improves upon $\p^{\Vert \StoqMA} = \p^{\StoqMA[\log]}$~\cite{Gharibian2019b}, which is the case of $d=1$ (versus our $d\in O(1)$ in \Cref{thm:constant-depth}).

\paragraph{2. A unified framework for embedding $\pC$ problems into APX-SIM.} To date, there are two known approaches for embedding QMA-oracle queries (and thus $\pQMAlog$ problems) into APX-SIM: The ``query gadget'' construction of Ambainis~\cite{Ambainis2014}, and the ``flag-qubit'' framework\footnote{This is a significantly generalized version of the ``sifter'' construction of Gharibian and Yirka~\cite{Gharibian2019}.} of Watson, Bausch, and Gharibian~\cite{WBG20} . Each of these frameworks has complementary pros and cons: The former handles \emph{adaptive} oracle queries, but is difficult to use when strong \emph{geometric} constraints for APX-SIM are desired (e.g. the physically motivated settings of 1D and/or translationally invariant Hamiltonians), whereas the latter requires non-adaptive queries, but is essentially agnostic to the circuit-to-Hamiltonian\footnote{Here, a ``circuit-to-Hamiltonian mapping'' is a quantum analogue of the Cook-Levin construction, i.e. a map from quantum verification circuits to local Hamiltonians.} mapping used (and thus easily handles geometric constraints).

Here, we utilize the construction behind our main result, \Cref{thm:bsn-s}, to unify these approaches into a single framework for embedding arbitrary $\pC$ computations into APX-SIM. The crux of the reduction is the following ``generalized lifting lemma'', whose full technical statement (\Cref{lem:lift}) is beyond the scope of this introduction (below, we state a significantly simplified version\footnote{For example, \Cref{lem:lift} also takes a separator tree as part of its input; for pedagogical purposes, the informal version presented here ignores this, as a separator tree is computed in $\poly(N)$ time in all our applications of \Cref{lem:lift} anyway.}).

\begin{lemma}[(Informal) Generalized Lifting Lemma (c.f. Lifting Lemma of~\cite{WBG20})]\label{lem:liftinformal}
Fix $C\in\QVp$ and any local circuit-to-Hamiltonian mapping $\Hw$ (\Cref{def:local-mapping}). Define $N_d:=2^{O(d(n)\log n)}$, and $N_s:=2^{O(s(n)\log n)}$ if $C\in\QV$ or $N_s:=2^{O(s(n)\log^2 n)}$ if $C=\StoqMA$. Define $N:=\min(N_s,N_d)$, and let $G$ be a $\CDAG$ instance $n$ vertices of separator number $s(n)$ (as in \Cref{thm:bsn-s}) and depth $d(n)$ (as in \Cref{thm:constant-depth}). Then, there exists a $\poly(N)$-time many-one reduction from $G$ to an instance $(H,A)$ of APX-SIM, such that $H$ has size $\poly(N)$ and satisfies all geometric properties of $\Hw$ (e.g. locality of clauses, 1D nearest-neighbor interactions, etc).
\end{lemma}
\noindent  In words, one can embed any $\pC$ computation directly into an APX-SIM instance $H$ in $\poly(N)$ time, irrespective of the choice of $C$ or $\Hw$ (i.e. the mapping is essentially black-box). For clarity, a lifting lemma for APX-SIM was first given in~\cite{WBG20}, which our \Cref{lem:liftinformal} generalizes as follows: (1) \cite{WBG20} requires parallel queries to $C$, whereas \Cref{lem:liftinformal} allows arbitrary $\pC$ computations (parameterized by separator number $s$), and (2) \cite{WBG20} requires promise gap amplification for $C$, which is not known to hold for StoqMA, whereas \Cref{lem:liftinformal} allows $C=\StoqMA$.

Next, by applying our lifting lemma for $C=\QMA$ and $s\in O(1)$, we obtain $\pQMAlog$-hardness of APX-SIM (\Cref{thm:liftedhardness}). This is not surprising, since our \Cref{thm:bsn} shows $\CDAG\in\pQMAlog$, and APX-SIM is $\pQMAlog$-hard~\cite{Ambainis2014,Gharibian2019}. What \emph{is} interesting, however, is:
\begin{enumerate}
     \item The map from $\pC$ to APX-SIM of \Cref{lem:liftinformal} is ``direct'', meaning we embed all the query dependencies of the input C-DAG directly into the flag qubit construction.
        \item A poly-time reduction from $\pQMA$ to APX-SIM for all $1\leq s\leq n$ would imply $\pQMA=\pQMAlog$ and is therefore unlikely, if one believes $\pQMA\ne\pQMAlog$. However, \Cref{lem:lift} shows $\pQMA$ \emph{can} be embedded into APX-SIM, at the expense of blowing up the APX-SIM instance's size to $N=2^{O(s(n)\log n)}$.
        \item Finally and most interestingly, the construction of~\cite{WBG20} is somewhat mysterious, in that it ``compresses'' multiple QMA query answers into a \emph{single} flag qubit, which \emph{a priori} appears at odds with Holevo's theorem\footnote{Roughly, Holevo's theorem says that $n$ qubits cannot reliably transmit more than $n$ bits of information.}. In the present paper, we reveal \emph{why} this works --- our construction utilizes the ``admissible weighting function'' framework of~\cite{Gottlob1995}, which Gottlob used to reduce $\pNP$ computations to maximization of a real-valued function, $f$. But as we discuss in \Cref{sscn:techniques}, this is precisely what the flag qubit framework allows one to simulate (in both~\cite{WBG20} and here)!
\end{enumerate}
In fact, we observe that~\cite{Ambainis2014} implicitly rediscovers\footnote{Like~\cite{Gottlob1995}, \cite{Ambainis2014} uses an exponentially growing weighting function to ensure soundness when simulating adaptive oracle queries, although the term ``admissible weighting function'' is not used in the latter.} a version of Gottlob's weighting function approach. Thus, underlying all three works of~\cite{Gottlob1995,Ambainis2014,WBG20}, as well as the current one, is a central unifying theme worth stressing:
\begin{theme}[Unifying theme]\label{theme:opt}
    The reduction of $\pC$ to maximizing a single real-valued function.
\end{theme}

Finally, for $C=\StoqMA$ and $s\in O(1)$, application of our lifting lemma is still possible (i.e. utilizing the $N_s$ term), but the Hamiltonian obtained is now quasi-polynomial in size, since $N:=2^{O(s(n)\log^2 n)}$ (\Cref{thm:liftedhardnessstoqma}). Luckily, we can instead utilize the $N_d$ term (i.e. bounded-depth setup) of the lifting lemma, which yields the desired $\poly(n)$-size output Hamiltonian when $d\in O(1)$. This means we recover the $\p^{\StoqMA[\log]}$-hardness result of \cite{Gharibian2019b} via the flag qubit framework (details in \Cref{sscn:apply},) resolving an open question of \cite{WBG20}. For clarity, \cite{Gharibian2019b}'s proof of this result is via perturbation theory, which we do not require here.

\paragraph{3. No-go statement for ``weak compression'' of polynomials.} To further drive home the point of \Cref{theme:opt}, we close with a simple no-go statement purely about polynomials. Roughly, given a real-valued polynomial $f$  (specified\footnote{Strictly speaking, we do not require arithmetic circuits to specify $f$. However, the multi-linear polynomials produced for our statement can have exponentially many terms if expanded fully in a monomial basis. To specify this succinctly, it suffices not to expand brackets in our polynomial descriptions (i.e. not replace $(x+y)(a+b)=xa+xb+ya+yb$); arithmetic cricuits are a natural avenue for formalising this.} via an arithmetic circuit (\Cref{def:arithmetic})), we define \emph{weak compression} as efficiently mapping $f$ to an efficiently computable real-valued function $g$, such that there exists an optimal point $y^*$ at which $g$ is maximized, from which (1) we may efficiently recover an optimal point $x^*$ maximizing $f$, and (2) $g(y^*)$ requires fewer bits than $f(x^*)$ to represent (i.e. has been ``compressed'').

\begin{restatable}{lemma}{lemmaPolycompress}\label{l:polycompress}
    Fix any function $h:R^+\to R^+$. Suppose that given any multi-linear polynomial $p$ (represented as an arithmetic circuit) requiring $B$ bits for some optimal solution (in the sense of \Cref{def:compress}), $p$ is weakly compressible to $h(B)$ bits. Then $\pNP\subseteq {\p^{\NP[h(B)]}}$.
\end{restatable}
\noindent Let us be clear that this statement is not at all surprising for the reader familiar with Krentel's work
~\cite{K88} on OptP (see \Cref{sscn:relatedwork}). Nevertheless, we believe it is worth formalizing, as it uses complexity theory to give a no-go statement about a purely mathematical concept (non-compressibility of polynomials). From \Cref{l:polycompress}, one obtains:
\begin{restatable}{corollary}{corollaryPolycompressA}\label{cor:polycompress1}
    If any multi-linear polynomial $p$ (represented as an arithmetic circuit) can be weakly compressed with $h(B)=O(\log B)$, then $\pNP\subseteq {\p^{\NP[\log]}}$.
\end{restatable}
\begin{restatable}{corollary}{corollaryPolycompressB}\label{cor:polycompress2}
    If any multi-linear polynomial $p$ requiring $B\in O(1)$ bits for some optimal solution can be weakly compressed with $h(B)=1$, then the Polynomial-Time Hierarchy (PH) collapses to its third level (more accurately, to $\textup{P}^{\Sigma_2^p}$).
\end{restatable}

\subsection{Techniques}\label{sscn:techniques}
\paragraph{1. Techniques for deciding C-DAGs.} At a high-level, our approach follows Gottlob's strategy for $\pNP$~\cite{Gottlob1995}: Given\footnote{Gottlob's modeling of query graphs is slightly different, in that nodes of the DAG encode propositional formulae, whereas here it is more convenient to put verification circuits at nodes.} a C-DAG $G$, we (1) ``compress'' $G$ to an equivalent query $G'$, (2) define an ``admissible weighting function'' on $G'$, (3) define an appropriate verifier $V$, on which binary search via C-oracle queries suffices to extract the original C-query answers in $G$, and thus to decide $G$ itself.
The key steps at which we deviate significantly from~\cite{Gottlob1995} are (1) and (3), as we now elaborate.

In more detail, in order to decide $G$, the goal is to compute a \emph{correct query string} $x$ for $G$, i.e. a string of answers to the $C$-oracle queries asked by $G$.
(Note $x$ is not necessarily unique when $C$ is a \emph{promise} class such as QMA.)
For this, fix any topological order $T$ on the nodes of $G$.
The clever insight of~\cite{Gottlob1995} (rediscovered in~\cite{Ambainis2014}), is that by ``weighting'' queries early in $T$ exponentially larger than queries later in $T$, one can force all queries, in order, to be answered correctly.
Roughly speaking, such an exponential weighting scheme $\omega$ is called ``admissible'' (\Cref{def:weighting-function}).
The core premise is then to map $(G,\omega)$ to a real-valued function $f$, so the maximum value of $f$ encodes the query string $x$.
Hence, by conducting binary search on $f$ via the C-oracle, one can identify $f$'s optimal value, thus recovering $x$. The  challenge is that for \emph{arbitrary} $G$, the maximum value of $f$ can scale exponentially in $n$, the number of nodes in $G$.
Thus, one requires $\poly(n)$ queries to extract $x$, obtaining no improvement over the $\pC$ computation $G$ we started with!\\

\begin{figure}[t]
    \begin{center}
        \includegraphics{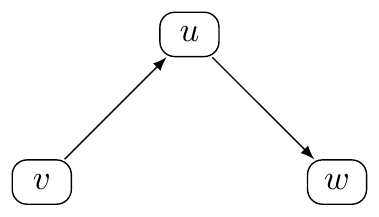}
        \hspace{10mm}
        \includegraphics{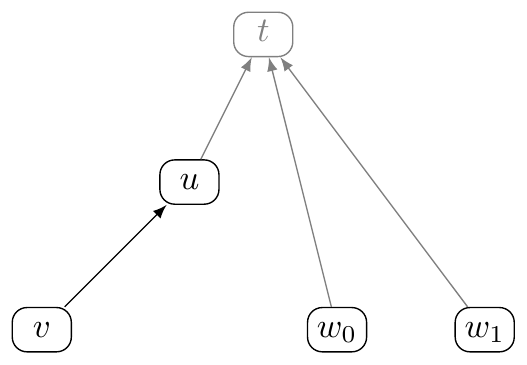}
       \caption{Simple example of a graph transformation, where the outputs of $u$ are decoupled by creating copies $w_0,w_1$ with hardcoded inputs. $t$ selects the copy of $w$ depending on the output of $u$.}
       \label{fig:simpletransform}
    \end{center}
\end{figure}

\noindent\emph{Compressing $G$.} To overcome this in our setting of bounded separator number (and beyond), we first recursively compute separators of $G$, obtaining a ``separator tree'' (\Cref{ssscn:separatortree}) structure overlaying $G$.
With this separator tree in hand, we show our main technical lemma, the Compression Lemma (\Cref{l:compress}).
Roughly, the idea behind the Compression Lemma is to ``decouple'' dependencies in $G$ by creating multiple copies of a node.
To illustrate, an oversimplified example is given in \Cref{fig:simpletransform}, where the output node $w$ depends on $u$, which depends on $v$.
(Each node encodes, say, an NP query.)
To remove the dependency of $w$ on $u$, we create two copies $w_0$ and $w_1$, where the input from $u$ is hardcoded as $0$ or $1$, respectively.
Then an output node $t$ is added to select the correct copy of $w$ depending on the output of $v$.

For clarity, this basic decoupling principle is reminiscent of that employed in~\cite{Gottlob1995}. However, whereas the latter maps $G$ to $G'$ via iterative local transformations (similar to \Cref{fig:simpletransform}, but without the $t$ node), here we are unable to make such an approach work.
Indeed, due to the much stronger coupling between nodes in our setting, we appear to acquire a carefully orchestrated, \emph{global} transformation of $G$ to $G'$.
Roughly, we must carefully exploit the separator tree as a guide to recursively create node copies and reroute wires, at the end of which we introduce a ``conductor\footnote{Meant in the sense of an ``orchestra conductor''.}'' node $t$ to orchestrate the madness.
For the reader interested in a brief peek at details (\Cref{sscn:graph-transformation}), \Cref{fig:transform} runs through an example graphically depicting the global compression, and \Cref{alg:compute-inputs} is used (e.g.) in $t$ to recursively orchestrate and compute the final output of the new C-DAG, $G'$.
The upshot of this global transformation is that, when $s\in O(1)$, $G'$ is ``compressed'' in such a way that (roughly) we can define an admissible weighting function of at most $\poly(n)$ weight on $G'$, as we do next.\\

\noindent\emph{Designing the verifier $V$.} The second main step (\Cref{sscn:solve-c-dag}) is to use an admissible weighting function on $G'$ to ``reduce'' $G'$ to maximization of a real-valued function, $t$ (\Cref{theme:opt}); we use (\Cref{eqn:t})
\begin{equation}\label{eqn:tintro}
  t(x,\psi_1,\dots,\psi_\nn) := \sum_{i=1}^\nn f(v_i) \bigl(\strut x_i\Pr[Q_i(z_i(x),\psi_i)=1] + (1-x_i)\gamma\bigr),
\end{equation}
where intuitively, $f(v_i)$ is the weight at node $i$, and $\Pr[Q_i(z_i(x),\psi_i)=1]$ is the probability that $C$-verifier $Q_i$ at node $v_i$ accepts, given incoming wires $z_i(x)$ from its parents and proof $\ket{\psi_i}$.
Function $t$ is carefully designed so that (1) any ``approximately maximum'' value of $t$ encodes a correct query string $x$ (\Cref{lem:correct-query-string}), and (2) we can design a C-verifier $V$ with acceptance probability precisely $t(x,\psi_1,\dots,\psi_\nn)$ (up to renormalization) (\Cref{lem:approximating-t}). Thus, binary search via $V$ allows us to extract $x$ from $t$. Crucially, by the compression of the previous step, when $s\in O(1)$, the maximum value of $t$ is at most $\poly(n)$, meaning $O(\log n)$ $C$-queries suffice in the binary search. Moreover, our $V$ is simple --- it simulates a random $Q_i$ (according to the distribution induced by weights $f(v_i)$) on $(x,\ket{\psi_i})$. We exploit this by defining $t$ over a \emph{cross product} of proofs $\ket{\psi_i}$ (rather than a tensor product, as is usual); this sleight of hand avoids complications regarding entanglement across proofs from previous works (e.g. \cite{WBG20}).

\paragraph{2. Techniques for a unified APX-SIM framework.} Roughly, \cite{WBG20} embeds a (say) $\pQMAlog$ computation $\Pi$ into APX-SIM as follows: (1) Build a ``superverifier'' circuit $V$, which verifies each of the queries of $\Pi$ in parallel, and conditioned on the output of each subverifier, performs a small rotation on a shared ``flag qubit'', $q$. The superverifier $V$ is then pushed through an abstract circuit-to-Hamiltonian mapping $\Hw$, and the encoding of $q$ in the resulting Hamiltonian $\Hw(V)$ is carefully penalized to force low energy states to correctly answer all queries. The advantage of this setup is that it is oblivious to the choice of $\Hw$; the disadvantage is that it requires a somewhat involved exchange argument to ensure soundness against entanglement across parallel proofs.

Recall now that our main construction rolls up an entire arbitrary C-DAG into a single C-verifier, $V$ (\Cref{lem:approximating-t}). What we next show is that our $V$ can rather simply be substituted for the superverifier $V$ of~\cite{WBG20} in the flag qubit construction. The key reason this works is again \Cref{theme:opt} --- since, as mentioned above, the acceptance probability of our $V$ literally encodes the value of $t$, we can treat the output wire of our $V$ as the ``new flag qubit'' $q$ (thus eliminating the multiple rounds of small rotations in~\cite{WBG20}). As in~\cite{WBG20}, by then mapping $V$ to $\Hw(V)$, we can now penalize $q$ on the Hamiltonian side to force all low energy states of $\Hw(V)$ to implicitly maximize $t$! Finally, we remark that since our $V$ is naturally robust against entanglement across proofs, our proof of correctness is significantly simpler than~\cite{WBG20}.

\paragraph{3. Techniques for ``weak compression'' of polynomials.} This result follows easily by combining \Cref{sec:approximating-t} with standard techniques, so we keep the discussion brief. Roughly, given an NP-DAG, we (1) apply the Cook-Levin theorem to map each NP verifier to a SAT formula, (2) arithmetize each of these SAT formula and combine them to simulate \Cref{eqn:tintro} on the Boolean hypercube, and (3) linearize the resulting multi-variate polynomial; denote the output as $p$. Since $p$ is multilinear, it is maximized on our domain of interest on a vertex of the hypercube; thus, by design, from the maximum value of $p$, we can recover the maximum value of $t$, from which we can extract the correct query string for the input NP-DAG. The argument is concluded by observing that to identify the maximum $p^*$ of $p$, a binary search via NP-oracle requires $O(\log (p^*))$ queries. As an aside, the collapse of PH in \Cref{cor:polycompress2} leverages Hartmanis' result that if $\pNPtwo=\pNPone$, then $\PH=\textup{P}^{\Sigma_2^p}$~\cite{Hartmanis1993}.

\subsection{Related Work}\label{sscn:relatedwork}

\paragraph{The classes $\pNP$ and $\pNPlog{}$.}
As mentioned above, $\NP\cup\coNP\subseteq\pNPlog\subseteq \Sigma_2^p$, and $\pNPlog\subseteq\PP$~\cite{Beigel1989}.
It holds that $\pNPlog{} = \p^{\Vert\NP}$ \cite{Hemachandra1989,Buss1991}.
Gottlob~\cite{Gottlob1995} showed that $\p^\NP$ with a tree for a query graph equals $\pNPlog{}$ (this also follows from our \Cref{thm:bsn}).
It is believed that for any $k\in O(1)$, the classes $\textup{P}^{\textup{NP}[k]}$, $\textup{P}^{\textup{NP}[\log^k n]}$, and $\textup{P}^{\textup{NP}}$ are distinct.
For example, $\textup{P}^{\textup{NP}[1]}=\textup{P}^{\textup{NP}[2]}$ implies both $\textup{P}^{\textup{NP}[1]}=\pNPlog$ and a collapse of PH to $\Delta_3^p=\textup{P}^{\Sigma_2^p}$~\cite{Hartmanis1993}.
However, it is known that $\textup{P}^{\textup{NP}[\log^k(n)]}=\textup{P}^{\Vert\textup{NP}[\log^{k+1}(n)]}$ for all $k\geq 1$~\cite{Castro1992}.
Complete problems for $\pNPlog{}$ include determining a winner in Lewis Carroll's 1876 voting system~\cite{Hemaspaandra1997}, and a $\textup{P}^{\textup{NP}[\log^2 n]}$-complete problem is model checking for certain branching-time temporal logics~\cite{S03}.

Closely related to one of the central themes of this work, \Cref{theme:opt}, is Krentel's~\cite{Krentel1988} work on OptP.
Roughly, OptP[$z(n)$] is the class of \emph{functions} (i.e. not decision problems) computable via maximization of a real-valued function, where the function is restricted to $z(n)$ bits of output precision.
Krentel shows the classes OptP[$z(n)$] and $\textup{FP}^{\textup{NP}[z(n)]}$ are equivalent (FP the set of \emph{functions} computable in poly-time).
Through this,~\cite{Krentel1988} obtains (e.g.) that determining whether the length of the shortest traveling salesperson tour in a graph $G$ is divisible by $k$ is $\pNP$-complete, but that determining if the size of the max clique in $G$ is divisible by $k$ is only $\pNPlog$-complete.
Before this, Papadimitriou had shown~\cite{P82} that deciding if $G$ has a \emph{unique} optimum traveling salesperson tour is $\pNP$-complete.

\paragraph{$\QMA$, $\pQMAlog$ and related classes.}
Kitaev's ``quantum Cook-Levin/circuit-to-Hamiltonian'' construction showing QMA-completeness for the local Hamiltonian problem has since been greatly extended to many settings (e.g. \cite{Kempe2003,Kempe2006,AGIK09,GI09}).
For $\QMA(2)$, Chailloux and Sattath~\cite{Chailloux2012} showed the \emph{separable sparse Hamiltonian problem} is $\QMA(2)$-complete.
Fefferman and Lin~\cite{Fefferman2016} prove that the local Hamiltonian problem with \emph{exponentially} small promise gap is $\PSPACE{}$-complete.
See (e.g.)~\cite{Osborne2012a,Gharibian2015b} for surveys and further results.

Ambainis~\cite{Ambainis2014} initiated the study of $\pQMAlog$, and showed APX-SIM is $\pQMAlog$-complete and SPECTRAL GAP (deciding if the spectral gap of a local Hamiltonian is large or small) is $\p^{\textup{UQMA}[\log]}$-hard.
These results were obtained for log-local observables (APX-SIM) and Hamiltonians (APX-SIM and SPECTRAL GAP).
Gharibian and Yirka~\cite{Gharibian2019} improve both results to $O(1)$-local, and show $\pQMAlog\subseteq\PP$.
In contrast to $\pNPlog$, $\pQMAlog$ is \emph{not} believed to be in PH, since even BQP is believed outside of PH~\cite{Aa10,RT18}.
Gharibian, Piddock, and Yirka~\cite{Gharibian2019b} next obtain a complexity classification of $\pQMAlog$ (along the lines of Cubitt and Montanaro~\cite{CM16}) depending on the class of Hamiltonians employed; this includes, for example, $\pStoqMAlog$-completeness for APX-SIM on stoquastic Hamiltonians.
They also introduce the ``sifter'' framework to show the first $\pQMAlog$-hardness result for 1D Hamiltonians on the line.
Watson, Bausch, and Gharibian~\cite{WBG20} significantly extend the sifter framework to develop the flag-qubit framework (also used in \Cref{scn:APXSIM}), showing (among other results) that APX-SIM on 1D {translation-invariant} systems is $\p^{\QMAexp}$-complete.

Most recently, Watson and Bausch~\cite{WB21} show a $\p^{\QMAexp}$-completeness result for approximating a critical boundary in the phase diagram of a translationally-invariant Hamiltonian.
Aharonov and Irani~\cite{AI21} and Watson and Cubitt~\cite{WC21} simultaneously and independently study variants of the problem of computing digits of the ground state energy of a translationally invariant Hamiltonian in the thermodynamic limit.
The former shows that the function version of this problem lies between $\textup{FP}^{\NEXP}$ and $\textup{FP}^{\QMAexp}$, while the latter shows that a decision version of the energy density problem is between $\p^{\textup{NEEXP}}$ and $\textup{EXP}^{\QMAexp}$ (for quantum Hamiltonians).

\subsection{Open questions}\label{sscn:open}

First, can our main result (\Cref{thm:bsn-s}) be extended to further classes of graphs, perhaps by considering different parameterizations, such as graphs with logarithmic pathwidth (which includes the case of constant separator number)?
Second, \Cref{thm:bsn-s} gives non-trivial bounds when (say) $s\in O(1)$ or $s\in O(\polylog(n))$.
For $s\in \Theta(n)$, however, the DTIME base therein scales as $2^n$, thus yielding a trivial upper bound on $\CDAGs$.
Can our bound be improved from $\DTIME\left(2^{O(s(n)\log n)}\right)^{C[s(n)\log n]}$ to $\DTIME\left(2^{O(s(n))}\right)^{C[s(n)\log n]}$ (i.e. shave off the extra log factor in the base)?
If so, one would immediately recover the $\p^{\Vert\StoqMA}\subseteq\p^{\StoqMA[\log]}$ result of~\cite{Gharibian2019b} (currently, we rely on \Cref{thm:constant-depth} to recover this here), and more generally, our framework would not take a hit when applied to classes $C$ without promise gap amplification.
However, what is \emph{unlikely} is to show a bound of $\DTIME\left(2^{O(s(n))}\right)^{C[s(n)]}$ --- since $\p^{\Vert \NP}$ has $s\in O(1)$, this would imply $\p^{\Vert \NP}=\pNPlog\in\p^{\NP[k]}$ for $k\in O(1)$.
Third, do our theorems also hold for complexity classes such as UniqueQMA ($\UQMA{}$) or $\QMA_1$ (QMA with perfect completeness)?
Here, the main difficulty seems to be \emph{invalid} queries (queries violating the promise), as then the verifier from \Cref{lem:approximating-t} does not necessarily have a unique proof or perfect completeness.
One could also consider \AM-like complexity classes instead of the \MA-like classes we used.

\section{Preliminaries}\label{sec:preliminaries}

\paragraph{Notation.} $S=\bigcupdot_i S_i$ denotes a partition of set $S$ into subsets $S_i$. $:=$ denotes a definition.

\paragraph{Promise problems.} Due to the inherently probabilistic nature of quantum computation, the quantum complexity classes we are interested in are defined in terms of \emph{promise problems}.
A promise problem $\Pi$ is defined by a tuple $\Pi = (\Piy, \Pin, \Pii)$ with $\Piy\cupdot\Pin\cupdot\Pii = \bin^*$.
We call $x\in\Piy$ a \emph{yes-instance}, $x\in\Pin$ a \emph{no-instance}, and $x\in\Pii$ an \emph{invalid instance}.

\begin{definition}[\QP{} (quasi-polynomial time)]\label{def:qp}
    $\QP=\bigcup_k\DTIME(n^{\log^k n})$ is the set of problems accepted by a deterministic Turing machine in quasi-polynomial time.
\end{definition}

\subsection{Quantum Complexity Classes}\label{sscn:complexity classes}

The circuits used by quantum complexity classes belong to \emph{polynomial-time uniform quantum circuit families} $\{ Q_n\}$.
That means, there exists a Turing machine that on input $n$ outputs a classical description of a quantum circuit $Q_n$ in time $\poly(n)$.
Qubits are represented by the Hilbert space $\B := \C^2$.

The arguably most natural quantum analogue of \NP{} (or \MA) is $\QMA$, where a $\BQP$-verifier is given an additional quantum proof.

\begin{definition}[$\QMA$]
    Fix polynomials $p(n)$ and $q(n)$.
    A promise problem $\Pi$ is in $\QMA$ (Quantum Merlin Arthur) if there exists a polynomial-time uniform quantum circuit family $\{Q_n\}$ such that the following holds:
    \begin{itemize}
        \item For all $n$, $Q_n \in \U\left(\B_A^{\otimes n} \otimes \B_B^{\otimes p(n)}\otimes \B_C^{\otimes q(n)}\right)$. The register $A$ is used for the input, $B$ contains the proof, and $C$ the ancillae initialized to $\ket0$.
        \item $\forall x\in\Pi_\text{yes} \;\exists \ket\psi \in \B^{\otimes p(|x|)}: \Pr[Q_{|x|}\text{ accepts } \ket{x}\ket\psi] \ge 2/3$
        \item $\forall x\in\Pi_\text{no} \;\forall \ket\psi \in \B^{\otimes p(|x|)}: \Pr[Q_{|x|}\text{ accepts } \ket{x}\ket\psi] \le 1/3$
    \end{itemize}
\end{definition}

Here, we say a quantum circuit $Q_n$ accepts an input $\ket x\ket\psi$ if measuring the first qubit of the ancilla register $C$ in the standard basis results in outcome $\ket1$.
The acceptance probability is then given by
\begin{equation}
    \Pr[Q\text{ accepts }\ket x\ket \psi] = \enorm*{
    \left(I_A\otimes\ketbra11_{C_1}\otimes I \right)  U\ket x_A\ket\psi_B\ket0^{\otimes q(n)}_C}^2.
\end{equation}

Note that the thresholds $c=2/3$ and and $s=1/3$ may be replaced with $c=1-\epsilon$ and $s=\epsilon$ such that $\epsilon \ge 2^{-\poly(n)}$ \cite{Kitaev2002}.
We refer to $c$ as \emph{completeness}, $s$ as \emph{soundness}, and $c-s$ as the \emph{promise gap}.

We also consider special cases of $\QMA$.
In $\QCMA$, the proof is classical, i.e. $\ket\psi\in\bin^{p(n)}$.
In $QMA(k)$, the verifier receives $k$ unentangled proofs, i.e. $\ket{\psi} = \bigotimes_{j=1}^k\ket{\psi_j}$).
It holds that $\QMA(2) = \QMA(\poly(n))$ as shown by Harrow and Montanaro~\cite{Harrow2013}.
Therefore, probability amplification is possible.
In $\QMAexp$, $p(n)$ and $q(n)$ are allowed to be exponential (i.e. $2^{\poly(n)}$) and $\{Q_n\}$ is an exponential-time uniform quantum circuit family.
$\QMAexp$ can be considered the quantum analogue of $\NEXP$.

The classical complexity classes $\NP$ and $\MA$ (Merlin-Arthur) may also be considered special cases of $\QMA$.
Restricting $\QMA$ to classical proofs and classical randomized verifiers results in $\MA$.
Additionally requiring perfect completeness and soundness yields $\NP$.
Note that $\NP$ and $\MA$ are usually equivalently defined as the problems accepted by nondeterministic (randomized) Turing machines.

Next, we define the $k$-local Hamiltonian problem, which was shown to be $\QMA$-complete in a \enquote{quantum Cook-Levin theorem} by Kitaev~\cite{Kitaev2002}.

\begin{definition}[$k$-local Hamiltonian]\label{def:local-hamiltonian}
    A Hermitian operator $H\in\hermp{\B^{\otimes n}}$ acting on $n$ qubits is a $k$-local Hamiltonian if it can be written as
    \begin{equation} H = \sum_{\subalign{S&\subseteq [n]\\ |S|&\le k}} H_S \otimes I_{[n]\setminus S}.\end{equation}
    Additionally, $0\preccurlyeq H_S \preccurlyeq I$ holds without loss of generality.

    We refer to the minimum eigenvalue $\lminp{H}$ as the \emph{ground state energy} of $H$ and the corresponding eigenvectors as \emph{ground states}.
\end{definition}


\begin{definition}[$\kLH(H,k,a,b)$]
    Given a $k$-local Hamiltonian $H=\sum_i H_i $ acting on $N$ qubits and real numbers $a,b$ such that $b-a\geq N^{-c}$, for $c>0$ constant, decide:
    \begin{itemize}
        \item[YES.] If $\lmin(H) \le a$ (i.e. the ground state energy of $H$ is at most $a$).
        \item[NO.] If $\lmin(H) \ge b$.
    \end{itemize}
\end{definition}

Next, we give formally define the oracle based complexity classes used throughout this paper.

\begin{definition}[$\p^C$]\label{def:oracle-classes}
    Let $C$ be a complexity class with complete problem $\Pi$.
    $\p^C = \p^\Pi$ is the class of (promise) problems that can be decided by a polynomial-time deterministic Turing machine $M$ with the ability to query an oracle for $\Pi$.
    If $M$ asks an \emph{invalid} query $x\in\Pii$, the oracle may respond arbitrarily.

    We say $\Gamma\in\p^C$ if there exists an $M$ as above such that $M$ accepts/rejects for $x\in \Gammay$/$x\in\Gamman$, regardless of how invalid queries are answered.

    For a function $f$, we define $\p^{C[f]}$ in the same way, but with the restriction that $M$ may ask at most $O(f(n))$ queries on input of length $n$.

    For an integer $k$, we define $\p^{C[k]}$, where $M$ may ask at most $k$ queries on each input.

    $\p^{\Vert C}$ denotes the class where $M$ must ask all queries at the same time.
    We call these queries \emph{non-adaptive} opposed to the \emph{adaptive} queries of the above classes, because the queries do not depend on the results of other queries.

    For a function $f: \bin^*\to\bin^*$, we define $\p^f$ and the other classes analogously, except that $M$ may now query the oracle for values $f(x)$.
\end{definition}

The $\pQMAlog{}$-complete problem is $\apxsim{}$ (approximate simulation).
It essentially asks whether a given Hamiltonian has a ground state with a certain property (e.g., a ground state where the first qubit is set to $\ket1$).


\begin{definition}[$\apxsim(H,A,k,l,a,b,\delta)$~\cite{Ambainis2014}]\label{def:apx}
        Given a $k$-local Hamiltonian $H=\sum_i H_i $ acting on $N$ qubits, an $l$-local observable $A$, and real numbers $a$, $b$, and $\delta$ such that $b-a\geq N^{-c}$ and $\delta\geq N^{-c'}$, for $c,c'>0$ constant, decide:
    \begin{itemize}
        \item[YES.] If $H$ has a ground state $\ket{\psi}$ satisfying $\bra{\psi}A\ket{\psi}\leq a$.
        \item[NO.] If for all $\ket{\psi}$ satisfying $\bra{\psi}H\ket{\psi}\leq \lmin(H)+\delta$, it holds that $\bra{\psi}A\ket{\psi}\geq b$.
    \end{itemize}
\end{definition}

Ambainis showed completeness for $k=\Theta(\log n)$ \cite{Ambainis2014}.
Gharibian and Yirka~\cite{Gharibian2019} improved this to $k=5$.
Gharibian, Piddock, and Yirka~\cite{Gharibian2019b} improved this to $k=2$ for physically motivated Hamiltonian models.

\begin{definition}[$\StoqMA{}$ \cite{Bravyi2006}]\label{def:stoqma}
    Fix polynomials $\alpha(n),\beta(n),p(n),q(n),r(n)$ with $\alpha(n)-\beta(n)\ge1/\poly(n)$.
    A promise problem $\Pi$ is in $\StoqMA$ (Stoquastic Merlin Arthur) if there exists a polynomial-time uniform quantum circuit family $\{Q_n\}$ such that the following holds:
    \begin{itemize}
        \item For all $n$, $Q_n \in \U\left(\B_A^{\otimes n} \otimes \B_B^{\otimes p(n)}\otimes \B_C^{\otimes q(n)}\otimes \B_D^{\otimes r(n)}\right)$. The register $A$ is used for the input, $B$ contains the proof, $C$ ancillae initialized to $\ket{0}$, and $D$ ancillae initialized to $\ket+$.
        $Q_n$ only uses $X$, CNOT, and Toffoli gates.
        \item For $x\in\bin^*$, $\abs x=n$, $\ket\psi\in\B^{\otimes p(n)}$, let $\ket{\psiin} := \ket x_A\ket{\psi}_B\ket{0}_C^{\otimes q(n)}\ket{+}_D^{\otimes r(n)}$.
        The acceptance probability is then given by
        \begin{equation} \Pr[Q_{n}\text{ accepts } \ket{x}\ket\psi] = \braketb{\psiin}{Q_n^\dagger \Piacc{}\, Q_n}, \end{equation}
        where $\Piacc = \ketbra++_{C_1}$ measures the first ancilla in the $\{\ket+,\ket-\}$ basis.
        \item $\forall x\in\Pi_\text{yes} \;\exists \ket\psi \in \B^{\otimes p(|x|)}: \Pr[Q_{|x|}\text{ accepts } \ket{x}\ket\psi] \ge \alpha(n)$
        \item $\forall x\in\Pi_\text{no} \;\forall \ket\psi \in \B^{\otimes p(|x|)}: \Pr[Q_{|x|}\text{ accepts } \ket{x}\ket\psi] \le \beta(n)$
    \end{itemize}
\end{definition}

Note that the only difference between $\StoqMA{}$ and $\MA$ is that $\StoqMA{}$ may perform its final measurement in the $\{\ket+,\ket-\}$ basis (i.e. setting $\Piacc:=\ketbra00_{C_1}$ would result in $\MA$) \cite{Bravyi2006}.

It further holds that the $\StoqMA{}$ verifier accepts any state with only nonnegative coordinates with probability $\ge1/2$.
Therefore, we cannot amplify the gap by majority voting as for $\MA$.
Recently, Aharonov, Grilo, and Liu~\cite{Aharonov2020} have shown that $\StoqMA$ with $\alpha(n) = 1-\negl(n)$ and $\beta(n) = 1-1/\poly(n)$ is contained in $\MA$, where $\negl(n)$ denotes a function smaller than all inverse polynomials for sufficiently large $n$.
It is therefore unlikely that such an amplification is possible.

\subsection{Graph Theory}

Let $G=(V,E)$ be a directed graph.
For a node $v\in V$, we define $\indeg(v)$ and $\outdeg(v)$ as the number of incoming and outgoing edges, respectively.
The sets $\pred(v) := \{w\in V\mid (w,v)\in E\}$ and $\suc(v) := \{w\in V\mid(v,w)\in E\}$ denote the parents and children of $v$, respectively.
The set of \emph{ancestors} (\emph{descendants}) of node $v$ is the set of all $u\in V\setminus\set{v}$, such that there is a directed path in $G$ from $u$ to $v$ ($v$ to $u$).
If $G$ contains no directed cycles, we call it a DAG (directed acyclic graph).

\begin{definition}[Tree decomposition]\label{def:tree-decomposition}
    Let $G=(V,E)$ be an undirected graph.
    A \emph{tree decomposition} $T=(V_T,E_T)$ of $G$ is a graph with $m$ nodes labelled by subsets $X_1,\dots,X_m\subseteq V$ such that:
    \begin{itemize}
        \item Each node of $G$ is contained in some node $X_i$ of $T$: $\bigcup_{i=1}^m X_i = V$.
        \item For all $(u,v)\in E$, there exists an $X_i$ such that $u,v \in X_i$.
        \item For all $v\in V$, the subtree in $T$ induced by $\set{X_i\mid v\in X_i}$ is connected.
    \end{itemize}
    The \emph{width} of a tree decomposition $T$ is defined as $\width(T) := \max_i \abs{X_i}-1$.
    The \emph{treewidth} of $G$, denoted $\tw(G)$, is defined as the minimum width among all possible tree decompositions of $G$.
\end{definition}

\noindent Bodlaender~\cite{Bodlaender1993} has shown that tree decompositions for graphs with bounded treewidth (i.e. $\tw(G)=O(1)$) can be computed in linear time.

The connection between tree decompositions and \emph{separators}, which we define next, has a long and well-studied history (e.g.~\cite{RS86,Reed1992,BGHK95,Amir10,BDDFLP13}).

\begin{definition}[Separator number \cite{Gruber2012}]\label{def:separator-number}
    Let $G=(V,E)$ be an undirected graph.
    A set $S\subseteq V$ is a \emph{separator} of $G$ if $G\setminus S$ (i.e. the graph induced by the nodes $V\setminus S$) has at least two connected components or at most one node.
    $S$ is \emph{balanced} if every connected component of $G\setminus S$ has at most $\lceil (\abs V - \abs S)/2 \rceil$ nodes. The \emph{balanced separator number} of $G$, denoted $\s(G)$, is the smallest $k$ such that for every $Q\subseteq V$, the induced subgraph $G[Q]$ has a balanced separator of size at most $k$.
\end{definition}

\begin{lemma}[Theorem 9 of \cite{Gruber2012} (see also\protect\footnote{Proposition 2.5 of~\cite{RS86} gives the slightly weaker bound $\s(G)\le \tw(G)+1$, which also suffices for our purposes.}~\cite{RS86,BGHK95})]
    \label{lem:separator-number}
    $\s(G) \le \tw(G) \le O(s(G)\cdot\log n)$.
\end{lemma}
\noindent We define tree decompositions and separator number for a directed graph $G$ on the undirected version of $G$.
It appears to be an open problem whether $\tw(G) = \Theta(s(G))$ holds.
However, resolving this question would not improve our results, since we only use the first inequality.

\subsubsection{Separator Trees}\label{ssscn:separatortree}

The separator number allows us to decompose graphs into \emph{separator trees}, which we use to evaluate query graphs more efficiently.

\begin{definition}\label{def:separator-tree}
    A \emph{(balanced) separator tree} of an undirected graph $G=(V,E)$ is a tree $T=(V_T, E_T)$, with vertices in $V_T$ labelled by subsets $\{S_1,\dots,S_m\}$ satisfying $\bigcupdot_{i=1}^m S_i = V$, and $T$ being rooted in $S_1$. $S_1$ is a (balanced) separator of $G$, and the trees rooted in the children of $S_1$ are (balanced) separator trees of $G\setminus S_1$. To distinguish vertices/edges of $G$ from vertices/edges of $T$, we refer to the latter as \emph{supervertices/superedges}. A path along superedges is called a \emph{superpath}. The unique superpath from $S_1$ to any supervertex $S$ is called a \emph{branch} of the tree.
\end{definition}
\noindent Unless noted otherwise, throughout this work we assume separators are balanced.

\begin{lemma}\label{lem:compute-separator-tree}
    Given an $n$-vertex graph $G=(V,E)$, a separator tree $T$ of $G$ with separator number $s:=\s(G)$ can be computed in time $n^{O(s)}$.
\end{lemma}
\begin{proof}
    By \Cref{def:separator-number}, every induced subgraph of $G$ has a balanced separator of size at most $s$. Thus, the brute force approach to build a separator tree is to first brute force search for a separator $S$ of $G$ in time $n^{s'}$ for $s'= s+2$ (try all ${n}\choose{s}$ subsets of vertices, for each subset check $O(n^2)$ edges), remove it, and recurse on all induced balanced subgraphs on $V\setminus S$. (Technically, since we do not know $s$ beforehand, we can try all values for separator size starting from $2$ onwards via brute force; this does not affect the overall runtime.)

    To analyze the runtime of this procedure over all recursive calls, a slight non-triviality is that for a balanced separator $S$ (\Cref{def:separator-number}), we have no control over the sizes of each connected component of $G\setminus S$, other than no one component has size more that $\abs{V}/2$. Thus, the recurrence relation one obtains scales as $T(n)= \left(\sum_{i=1}^k T(n_i)\right) + n^{s'}$, where $2\leq k\leq n$, $\sum_{i=1}^k n_i\le n$, $n_i\leq n/2$ for all $i$, and for some $s'= s+2$. (In particular, this means the standard Master Theorem~\cite{BHS1980} cannot be applied.) In fact, the values of the $n_i$ can even change between levels of the recurrence.

    The analysis, luckily, is simple. Let $L=1$ denote the base case of the recurrence, which we view as the root of a recursion tree (i.e. each node $v$ of the tree is a recursive call, whose children correspond to the recursive calls made by $v$). At any level $L\geq 1$, we claim the additive cost at a node $v$ (i.e. corresponding to the ``$+n^{s'}$'' term) is at least twice the additive cost of its children. This implies the total cost incurred at level $L+1$ is at most half the cost of level $L$, giving a total cost for the algorithm via geometric series $\sum_{L=0}^{D-1} \frac{n^{s'}}{2^L}\leq 2 n^{s'}$, for all $D$ denoting the depth of the recursion, as claimed.

    To see that the cost at any $v$ is indeed at least twice the cost of its children $w_1,\ldots, w_k$, let $n$ be the input size for $v$ and $n_1,\ldots, n_k$ the input sizes for $w_1,\ldots, w_k$, respectively. Then, the total additive cost across all children of $v$ is
    \begin{equation}
        \sum_{i=1}^k n_i^{s'}= n^{s'}\sum_{i=1}^k \frac{n_i}{n}\left(\frac{n_i}{n}\right)^{s'-1}\leq n^{s'} \max_{i}\left(\frac{n_i}{n}\right)^{s'-1}\leq \frac{1}{2}n^{s'},
    \end{equation}
    where the first inequality follows since the coefficients $n_i/n$ yield a convex combination, and the second inequality since $n_i\leq n/2$ for all $i$ and $s'=s+2\geq 1$.
\end{proof}
\begin{remark}
Note that the separator tree computed by \Cref{lem:compute-separator-tree} may contain separators of different sizes $1\leq s'\leq s$. However, in this work it is convenient to assume without loss of generality that all separators have size exactly $s$. This can trivially be achieved by ``padding'' each separator $S$ of size $1\leq s'<s$ by adding dummy vertices to $S$ (and hence to $G$; all dummy vertices are isolated). The number of dummy vertices added is trivially at most $sn$ (there can never be more than $n$ separators); thus, the size of $G$ increases by at most $sn$ vertices, which does not affect any of our results.
\end{remark}

Additionally, although by definition of balanced separator, a balanced separator tree has $O(\log n)$ depth, at times we may wish to leverage a shorter depth tree if one should exist. For convenience, we hence state the following lemma.

\begin{lemma}\label{lem:compute-separator-tree-depth}
    Given an $n$-vertex graph $G$, depth $D$ and separator size $s$ ($D$ and $s$ are specified in unary), a separator tree of $G$ of depth $D$ with separators of size $s$ can be computed in time $n^{O(Ds)}$, if it exists.
\end{lemma}
\begin{proof}
    A brute-force approach similar to \Cref{lem:compute-separator-tree} is used, except there is a catch: At any level $L$ of the recursion, for each subset of vertices $S$ we consider, even if $S$ is a separator, it may not lead to a \emph{depth $D$} separator tree, even if such a tree exists. Thus, it does not suffice at level $L$ to simply find a size $s$ separator $S$, but rather in the worst case we may need to consider all such $O(n^s)$ such separators $S$. Thus, the recurrence relation now scales as $T(n)= n^s\left(\sum_{i=1}^k T(n_i)\right) + n^{s'}$. Running the same argument as \Cref{lem:compute-separator-tree} now yields total cost
    $
        \sum_{L=0}^{D-1} n^{s'}\left(\frac{n^s}{2}\right)^L\in n^{O(Ds)},
    $
    where recall $s'=s+2$.
\end{proof}

Finally, we remark that only the size $s$ of the separators in the balanced (or low-depth) separator tree is relevant for our algorithms.
The separator number $\s(G)$ is only used to compute separator trees more efficiently.
For a balanced separator tree, we may have $\s(G) \ge \Omega (s\cdot \log(n))$.\footnote{Proof: Let $G$ be a complete binary tree on $n$-nodes with additional edges from each node to its descendants. Then $\s(G)=\Theta(\log n)$, but $G$ has a separator tree with separators of size $1$.}

\section{Query graphs and \CDAG}

The main object of study in this work is the concept of a \emph{query graph}, which we now formally define in the context of a decision problem, $\CDAG$.

\begin{figure}[t]
    \begin{center}
        \includegraphics[height=20mm]{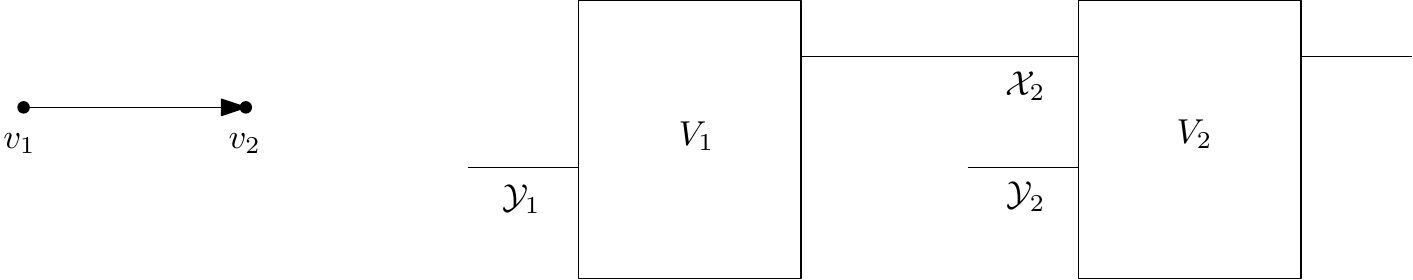}
        \vspace{2mm}
       \caption{Left: A simple example of an NP-DAG for with two nodes, with $v_2$ the output node. Right: The circuit view represented by the NP-DAG. Each $V_i$ is an NP verifier taking in input in register $\X_i$ and proof in register $\Y_i$. Note $v_1$ has in-degree $0$, hence $V_1$ has trivial input register $\X_1$. The output wire of $V_2$ carries the output of the NP-DAG.}
       \label{fig:CDAG}
    \end{center}
\end{figure}

\begin{definition}[\CDAG~(\Cref{fig:CDAG})]\label{def:c-dag}
    Fix any complexity class $C\in\QVp$.
    A \CDAG{} instance is defined by an $n$-node DAG $G=(V=\{v_1,\dots,v_n\},E)$, with structure as follows:
    \begin{itemize}
        \item  Vertex $v_n\in V$ is the unique vertex with $\outdeg(v_n) =0$, denoted the \emph{result node}.
        \item
        Each $v_i\in V$ is associated with a promise problem $\Pi^i\in C$ that determines the output of $v_i$.
        Formally, $\Pi^i$ is specified via a $\poly(n)$-sized description%
        \footnote{This description may be implicit to describe exponentially large circuits (e.g., for $\NEXP$).}
        of a verification circuit $Q_i$
        with designated input and proof registers $\X_i$ and $\Y_i$.%
        \footnote{For example, if $C=\NP$, then $\Piy^i$ is the set of all strings $x$ on $\X_i$, for which there exists a proof $y$ on $\Y_i$, such that NP verifier $Q_i$ accepts $(x,y)$.}
        The input register $\X_i$ consists of precisely $\indeg(v_i)$ bits/qubits, set to the string on $v_i$'s incoming edges/wires.
        In order to allow non-trivial $Q_i$ for bounded in-degree, we allow an implicit padding of $\X_i$ to $\poly(n)$ bits.
        $v_i$ has a single output wire, denoted $\outputc[v_i]$, corresponding to the output of the verifier $Q_i$.
\end{itemize}
    Finally, we say $G\in\CDAG_\yes$ (respectively, $G\in\CDAG_\no$) if the evaluation procedure EVALUATE (\Cref{alg:c-dag-eval}) outputs $1$ (respectively, $0$) \emph{deterministically} (i.e. regardless of how any invalid queries are answered).
\begin{algorithm}[t]
    \caption{Evaluation procedure for \CDAG}\label{alg:c-dag-eval}
    \begin{algorithmic}[1]
        \Function{Evaluate}{$G=(V,E)$}
            \State Sort the nodes of $V$ topologically into $v_1,\dots,v_n$.
            \State The variable $x_i\in\bin$ will denote the result of $v_i$'s query.
            \For{$i=1,\dots,n$}
                \State $z_i \gets \bigcirc_{v_j\in \pred(v_i)}\, x_j$\label{alg:c-dag-eval:z_i}
                \Comment $\bigcirc$ denotes concatenation (concatenation order is specified by $Q_i$).
                \State $x_i \gets\begin{cases}
                    1, &\text{if }z_i\in \Piy^i\\
                    0, &\text{if }z_i\in \Pin^i\\
                    \text{$0$ or $1$ (nondeterministically)}, &\text{if }z_i\in \Pii^i
                \end{cases}$\label{alg:c-dag-eval:x_i}
            \EndFor
            \State \Return $x_n$
            \Comment Recall $v_n$ is the result node.
        \EndFunction
    \end{algorithmic}
\end{algorithm}
\end{definition}
\begin{remark}
Observe that if $C$ is a promise class, then $\CDAG$ is a \emph{promise} problem (as opposed to a decision problem) --- this is because then $\Pii^i$ is not necessarily empty, and so we must be \emph{promised} that \Cref{alg:c-dag-eval} outputs either $0$ or $1$ deterministically, regardless of how invalid queries are answered.
\end{remark}
\begin{definition}[Correct query string]\label{def:correct}
    Any string $x\in\bin^n$ that can be produced via Line~\ref{alg:c-dag-eval:x_i} of \Cref{alg:c-dag-eval} is called a \emph{correct query string}.
\end{definition}
\begin{remark}
Intuitively, in \Cref{def:correct} the bits of $x$ encode a sequence of correct query answers corresponding to the nodes of $G$. Note the correct query string need not be unique if $C$ is a promise class (i.e. invalid queries are allowed). Also, we may view any query string as a function $V\to\bin$.
\end{remark}

\begin{remark}
Our notion of C-DAG is similar to the DAGS(NP) formalization of Gottlob (Definition 3.2 of~\cite{Gottlob1995}), except the latter has node queries encoded by propositional formulae. In contrast, here we use verification circuits at the nodes to make it easier to abstractly address a variety of verification classes $C$. (Alternatively, one might also consider ``quantizing'' the NP-dags of~\cite{Gottlob1995} by replacing propositional formulae with local Hamiltonians.)
\end{remark}

Just as Gottlob shows DAGS(NP) (more accurately, DAGS(SAT)) is $\p^{\NP}$-complete~\cite{Gottlob1995}, here we have the more general statement:
\begin{lemma}\label{l:pC}
    For any $C\in\QVp$, \CDAG{} is $\p^C$-complete.
\end{lemma}
\begin{proof}
     First, $\CDAG\in\p^C$ holds because a $\p^C$ machine can straightforwardly compute a correct query string by simulating \Cref{alg:c-dag-eval} on a \CDAG-instance $G$.
    By definition, if $G\in\CDAG_\yes$, then $x_n=1$, and if $G\in\CDAG_\no$, then $x_n=0$.

    Second, to show $\p^C$-hardness of \CDAG{}, we sketch a poly-time many-one reduction from $\p^C$ to \CDAG{}.
    Let $M$ be a $\p^C$ machine receiving input $x\in\set{0,1}^\ast$.
    Without loss of generality, we may assume that $M$ always performs $m \le \poly(|x|)$ queries, so let $G$ be a DAG with $m$ nodes.
    Node $v_i$ represents the $i$th query of $M$ and has incoming edges from $v_1,\dots,v_{i-1}$ (i.e. query $i$ depends on all previous queries).
    Then, $Q_i$ is defined as the circuit that, conditioned on the answers of queries $1$ through $i-1$, first computes the $C$-query $\phi$ (e.g. $\phi$ could be a SAT formula or a local Hamiltonian) which $M$ would send to the $C$-oracle for query $i$, and simulates the corresponding $C$-verification circuit on $\phi$, outputting the result of said verification. (For clarity, note the $C$-verifier is not actually ``run'' here; we are simply defining the action of $Q_i$ as part of the query graph for the reduction.) By construction and how YES/NO instances of $\CDAG$ are defined (\Cref{alg:c-dag-eval}), $G \in \CDAG$ if and only if $M$ accepts $x$.
\end{proof}

\begin{remark}
    When $C$ is a promise class, $\p^C$ is also a promise class (despite having $\p$ as a base). This is because, as with the definition of $\CDAG$ (\Cref{def:c-dag}), a valid $\p^C$ machine is promised to deterministically output the same answer regardless of how invalid queries are answered.
\end{remark}

Thus, \Cref{l:pC} says that on general query graphs $G$, \CDAG\ captures all of $\p^C$.
The primary aim of this paper is hence to consider graphs $G$ with \emph{bounded separator number} (which, by \Cref{lem:separator-number}, includes the case of bounded treewidth). For this, we introduce the following definition for convenience.

\begin{definition}[$\CDAGs$]\label{def:c-dagf}
    Let $s:\N\to\N$ be an efficiently computable function. Then, $\CDAGs$ is defined as $\CDAG$, except that $G$ has separator number $s(G)\in O(s(n))$, for $n$ the number of nodes used to specify the $\CDAG$ instance. For brevity, we use $\CDAGo$ to denote the case of $s\in O(1)$.
\end{definition}
\noindent Thus, the union of $\CDAGs$ over all polynomials $s:\N\mapsto\N$ equals $\CDAG$.

\section{Query Graphs with Bounded Separator Number}\label{scn:BSN}

We first state the main technical theorem of this section, \Cref{thm:separator-tree-time}, followed by the results we obtain from it as corollaries. The remainder of \Cref{scn:BSN} then proves \Cref{thm:separator-tree-time}. For clarity, throughout this work, we assume that the full specification of any \CDAG{} instance $G$ (i.e. the DAG itself, the verification circuits $Q_i$, etc) scales polynomially with its number of nodes, $n$.

\begin{theorem}\label{thm:separator-tree-time}
    Fix $C\in\QV$. As input, we are given (1) a \CDAG{} instance $G$ on $n$ nodes, and (2) a separator tree for $G$ of depth $D$ and separator size $s$.
    Then, $G$ can be decided in deterministic time $2^{O(sD+\log n)}$ with $O(sD + \log n)$ queries to a $C$-oracle.
\end{theorem}
\begin{remark}
 The class $\StoqMA$ is not included in \Cref{thm:separator-tree-time}; this is because the proof of the theorem requires $C$ with a constant promise gap, which $\StoqMA$ is not known to have. (See \Cref{sscn:stoqma} for the weaker result we are able to show for $\StoqMA$.)
\end{remark}

With \Cref{thm:separator-tree-time} in hand, we obtain the following results.

\theoremBsnS* 
\begin{proof}
    A separator tree of depth $D = O(\log(n))$ and separators of size $s=\s(G)$ is computed using \Cref{lem:compute-separator-tree} in time $n^{O(s)}$.
    Applying \Cref{thm:separator-tree-time} completes the proof.
\end{proof}

\noindent In words, this says that $\p^C$, with the restriction that the query graph used by the $\p$ machine has separator number $f(n)$, is contained in the class on the right side of \Cref{eqn:upper}. When $f\in O(1)$, this upper bound is tight:


\theoremBsn*
\begin{proof}
    $\CDAGo \in \p^{C[\log]}$ is immediate from \Cref{thm:bsn-s}. As for $\p^{C[\log]}$-hardness, we use the well-known fact that $ \p^{C[\log]} \subseteq \p^{\Vert C}$ for general~\cite{B91}\footnote{Reference \cite{B91} actually studies only NP, but the containment proof technique straightforwardly generalizes to other classes: Namely, instead of making logarithmically many adaptive queries to $C$, precompute the polynomially many potential queries the P machine could make, and send these in one parallel round to the $C$-oracle.} $C$, and observe $\p^{\Vert C}$-hardness of $\CDAGo$. Namely, the DAG $G$ for any input to a problem from $\p^{\Vert C}$ is a star, with all edges directed towards the center of the star, which is the output node. Thus, $G$ has separator size $1$ (i.e. remove the center of the star to isolate all remaining vertices), i.e. it encodes an instance of $\CDAGo$.
\end{proof}
\noindent
More generally, we obtain the following general scaling corollary when the separator number is polylogarithmic.
\corollaryQP*

\paragraph{Organization of remainder of section.} \Cref{sscn:weighting} introduces the notion of weighting functions. \Cref{sscn:graph-transformation} gives the main graph transformation which ``compresses'' a $\CDAG$ instance appropriately, and sets up a corresponding weighting function. This can be roughly thought of as a ``hardness proof'', i.e. that the compressed DAG output by this graph transformation captures the original $\CDAG$ instance. \Cref{sscn:solve-c-dag} gives the matching upper bound --- that the compressed DAG, coupled with an appropriate choice of weighting function, can now be resolved with fewer queries to a $C$-oracle. \Cref{sscn:stoqma} and \Cref{sscn:bounded} discuss the special cases of StoqMA and bounded depth (beyond the naive log bound)) $\CDAG$ instances, respectively.

\paragraph{Notation.} The remainder of this section introduces a fair amount of notation. For ease of reference, we collect notation here. $\Gamma(v) := \{w\mid (v,w)\in E\}$ is the neighbor set of vertex $v$. The descendents of vertex $v$ are denoted $\desc(v)$, i.e. the set of nodes reachable from vertex $v$ via a directed path, excluding $v$ itself. Analogously, the ancestors of vertex $v$ are denoted $\anc(v)$, i.e. the set of nodes from which there is a directed path to $v$, excluding $v$ itself.

\subsection{Weighting Functions}\label{sscn:weighting}

We now introduce the concept of \emph{weighting functions}, which assign a weight to each node in a DAG $G$.
Weighting functions were first used by Gottlob~\cite{Gottlob1995} to prove $\Trees(\NP) = \pNPlog$, and later implicitly by Ambainis~\cite{Ambainis2014} to show $\pQMAlog$-hardness of the Approximate Simulation (APX-SIM) problem.
We use a modified definition.

\begin{definition}[Weighting function]\label{def:weighting-function}
    Let $G=(V,E)$ be a DAG.
    An efficiently computable function $f:V\to\R$ is called a \emph{weighting function}.
    We say $f$ is \emph{$c$-admissible} for constant $c\in\R$ if for all $v\in V$,
    \begin{equation}\label{eq:admissible}
        f(v) \ge 1+c\sum_{w\in\Gamma(v)} f(w),
    \end{equation}
    where $\Gamma(v) := \{w\mid (v,w)\in E\}$ is the (out-going) neighbor set of $v$.
    The \emph{total weight} $W_f(G)$ of $G$ under weighting function $f$ is
    \begin{equation} W_f(G) = \sum_{v\in V}f(v).\end{equation}
\end{definition}

\begin{remark}
Our \Cref{def:weighting-function} is slightly weaker than Gottlob's~\cite{Gottlob1995}, which sums over all nodes in $\desc(v)$ (i.e. nodes reachable from $v$ via a directed path, excluding $v$ itself) instead of $\Gamma(v)$ in \eqref{eq:admissible}.
However, these definitions are equivalent up to a constant factor in $c$.
\end{remark}

\begin{definition}[Levels of a DAG]\label{def:level}
    Let $G=(V,E)$ be a DAG.
    We divide $G$ recursively into levels.
    Level $0$ is made up by the nodes without incoming edges.
    Level $i+1$ contains nodes $v$ that have only inputs $w$ (i.e. $(w,v)\in E$) with $\level(w) \le i$ and at least one input $w$ with $\level(w)=i$.
    We denote the level of a node $v\in V$ by $\level(v)$.
    Nodes on the last level are called \emph{terminal nodes}.
    The \emph{depth} of $G$, denoted $\depth(G)$ is the maximum level number.
\end{definition}

In the next lemma, we extend Gottlob's~\cite{Gottlob1995} admissible weighting functions to our definition of $c$-admissability (\Cref{def:weighting-function}).
For $c=1$,  the definitions are the same.

\begin{lemma}\label{lem:weighting-function}
    For any DAG $G=(V,E)$ and $c\ge2$, the weighting functions $\rho$ and $\omega$ below are $c$-admissible:
    \begin{align}
        \rho(v) &= (c\abs V)^{\depth(G)-\level(v)}\\
        \omega(v) &= (c+1)^{\abs{\desc(v)}}
    \end{align}
\end{lemma}
\begin{proof}
    The proof for $\rho$ is the same as in \cite{Gottlob1995}, whereas our proof for $\omega$ is significantly simplified.
    To argue $c$-admissability of $\rho$, let $v\in V$.
    By \Cref{def:level}, it holds that $\level(w) > \level(v)$ for all $w\in \Gamma(v)$.
    Therefore,
    \begin{align}
        1+c\sum_{w\in \Gamma(v)} \rho(w) \le c\abs V(c|V|)^{\depth(G)-\level(v)-1}
        = (c|V|)^{\depth(G)-\level(v)} = \rho(v).
    \end{align}
    For $\omega$, let $\{u_1,\dots,u_k\} = \Gamma(v)$ be topologically ordered (with respect to $G$).
    Then $\abs{\desc(u_i)} \le \abs{\desc(v)}-i$.
    Thus,
    \begin{align}
        1+c\sum_{i=1}^k \omega(u_i) &\le 1+c\sum_{i=1}^k (c+1)^{\abs{\desc(v)}-i}\\
        &\le 1+c\sum_{i=0}^{\abs{\desc(v)}-1}(c+1)^i\\
        &= 1+c\frac{(c+1)^{\abs{\desc(v)}}-1}{c}\\
        &= (c+1)^{\abs{\desc(v)}} \\
        &= \omega(v).
    \end{align}
\end{proof}

In Sections~\ref{sscn:graph-transformation} and \ref{sscn:solve-c-dag}, we assume $C\in\QVp$, unless stated otherwise.

\subsection{Graph transformation: The Compression Lemma}\label{sscn:graph-transformation}

Ideally, our aim for a given $\CDAG$ instance $G$ is to define a $c$-admissible weighting function $f$ with $W_f(G)$ as small as possible.
This is because in \Cref{sscn:solve-c-dag}, we show how to solve arbitrary \CDAG{}-instances using $O(\log W_f(G))$ $C$-queries.
Unfortunately, for an arbitrary \CDAG{}-instance $G$ there does not necessarily exist a $c$-admissible weighting function $f$ such that $W_f(G)$ is ``small'', e.g. subexponential.
Thus, in this section, we show:
\begin{restatable}{lemma}{lemmaCompress}\label{l:compress}
    As input, we are given a \CDAG{} instance $G$, and a separator tree for $G$ of depth $D$ and separator size $s$.
    Fix any constant $c\geq 2$. Then, a query graph $G^*=(V^*, E^*)$ with $\abs{V^*}\le 2^{O(sD)}n$, together with a $c$-admissible weighting function $f^*$ and $W_{f^*}(G^*)\le (c+1)^{O(sD)}n$, can be constructed in time $2^{O(sD+\log n)}$ such that $\Call{Evaluate}{G} = \Call{Evaluate}{G^*}$ (irrespective of nondeterministic choices in \Cref{alg:c-dag-eval}). As required by the definition of C-DAG (\Cref{def:c-dag}), each node of $G^*$ corresponds to a verification circuit of size $\poly(\abs{V^*})$.
\end{restatable}
\noindent Combining this with \Cref{sscn:solve-c-dag}, we will hence be able to decide $G^*$ with $O(sD)$ queries.

\paragraph{Brief outline.} The transformation from $G$ to $G^*$ proceeds in multiple steps.
First, we construct a graph $G'$ (\Cref{sscn:basic-construction}) where each node $v\in V'$ has $\abs{\desc(v)}\le O(sD)$, where recall $\desc(v)$ is the set of descendents of $v$. Roughly, this is achieved by exploiting the structure of separator trees to ``hardcode'' dependencies.
This leaves two issues.
First, for technical reasons $G'$ is lacking an output node, which we add in $G''$.
Second, we have redundant copies of nodes, which simplify the construction, but are problematic in the presence of invalid queries, as two copies of the same node with the same inputs may produce different outputs.
We merge these redundant nodes to obtain graph $G^*$, and define a suitable $c$-admissible weighting function in the process (\Cref{sscn:merging-nodes}).
\Cref{sscn:correctness} shows correctness.

\subsubsection{Basic Construction (\texorpdfstring{$G'$}{G'})}\label{sscn:basic-construction}

In this section, we construct $G'=(V', E')$.
We begin by formally stating the construction, followed by giving the intuition, and an illustration via \Cref{fig:transform2} and accompanying discussion.\\

\vspace{-1mm}
\noindent \emph{The graph transformation.} Let $T=(V_T, E_{T})$ be a separator tree (\Cref{def:separator-tree}) of $G$ of depth $D$ and separator size $s$.
A running example is given in \Cref{fig:transform1}.
Let $S\in V_T$ be an arbitrary supervertex and $S_1,\dots,S_{d}$ be the unique path along superedges from the root supervertex $S_1$ to $S_{d}:=S$ (define $d:=d_{S}\le D$ as the distance from the root plus one).
Recall $S$ is labelled by some subset of $s$ vertices, $S=(u_{S,1},\dots,u_{S,s})$, where we assume the sequence in which the $u_{S,i}$ are listed is consistent with some fixed topological order on all of $G$.
Define sets
\begin{equation}\label{eqn:VS}
    V_S := \Bigl\{ v_{S,i}^{z_1,\dots,z_d} \Bigm\vert i\in[s],\, z_1,\dots,z_d \in \bin^{s}\Bigr\}
\end{equation}
and set $V' = \bigcup_{S\in V_T} V_S$.
As depicted in \Cref{fig:transform2}, it will be helpful to continue to view $V_S$ as a set, even though $V_S$ is not a supervertex (i.e. $G'$ itself will not be a separator tree).
Intuitively, $v_{S,i}^{z_1,\dots,z_d}$ in $V'$ represents node $u_{S,i}$ in $V$, but conditioned on ``outcome strings'' $z_1,\dots,z_d\in\set{0,1}^s$ in the separators $S_1,\dots,S_d$.
For ease of reference, we define a surjective function $\inv:V'\mapsto V$ to formalize this relationship:
\begin{equation}\label{eqn:inv}
    \forall S,i,z_1,\ldots, z_d,\quad \inv(v_{S,i}^{z_1,\dots,z_d}) = u_{S,i}.
\end{equation}
Finally, since $T$ has at most $n$ supernodes, we have $\abs{V'}\le 2^{O(sD)}n$.

Next, define edges
\begin{equation}\label{eqn:newE}
    E_S = \Bigl\{ \bigl(v_{S,i}^{z_1,\dots,z_d}, v_{S_j,k}^{z_1,\dots,z_j}\bigr) \Bigm|
    i\in [s],\, j\in[d-1], \,
    u_{S_j,k} \in \desc(u_{S,i})
    \Bigr\},
\end{equation}
where recall $u_{S_j,k} \in \desc(u_{S,i})$ is the set of all descendants of $u_{S,i}$ in the {original} graph $G$. In words, each $E_S$ creates, for each copy of $u_{S,i}$, edges to all copies of {descendants} $u_{S_j}$ which are on a strictly higher level in the separator tree (due to the $ j\in[d-1]$ constraint). In the context of \Cref{fig:transform1}, this means we ``shortcut'' paths to descendents, but only via new edges pointing strictly ``upwards'' towards the root.
Set $E' = \bigcup_{S\in V_T} E_S$.
Observe that $\abs{\desc(v)}\le O(sD)$ for all $v\in V'$.\\
\vspace{-1mm}

\begin{figure}[p!]
\centering
\begin{subfigure}[b]{\textwidth}
\begin{center}
   \includegraphics{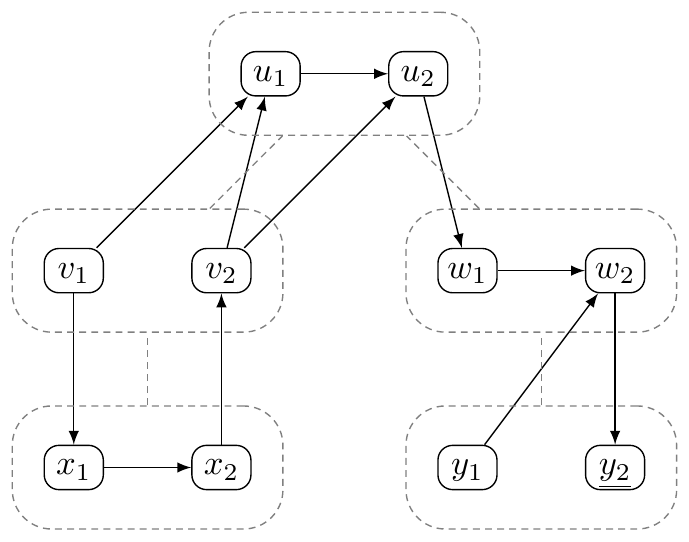}
\end{center}
   \caption{Query graph $G$ with separator tree of depth $D=3$ and separator size $s=2$ shown as an overlay via dashed lines. Recall from \Cref{def:separator-tree} that each dashed set (e.g. $\set{u_1,u_2}$) is called a supervertex, and dashed edges (e.g. between $\set{u_1,u_2}$ and $\set{v_1,v_2}$) are superedges. Vertex $y_2$ is underlined to denote it as the output node.}
   \label{fig:transform1}
\end{subfigure}

\begin{subfigure}[b]{\textwidth}
\begin{center}
   \includegraphics{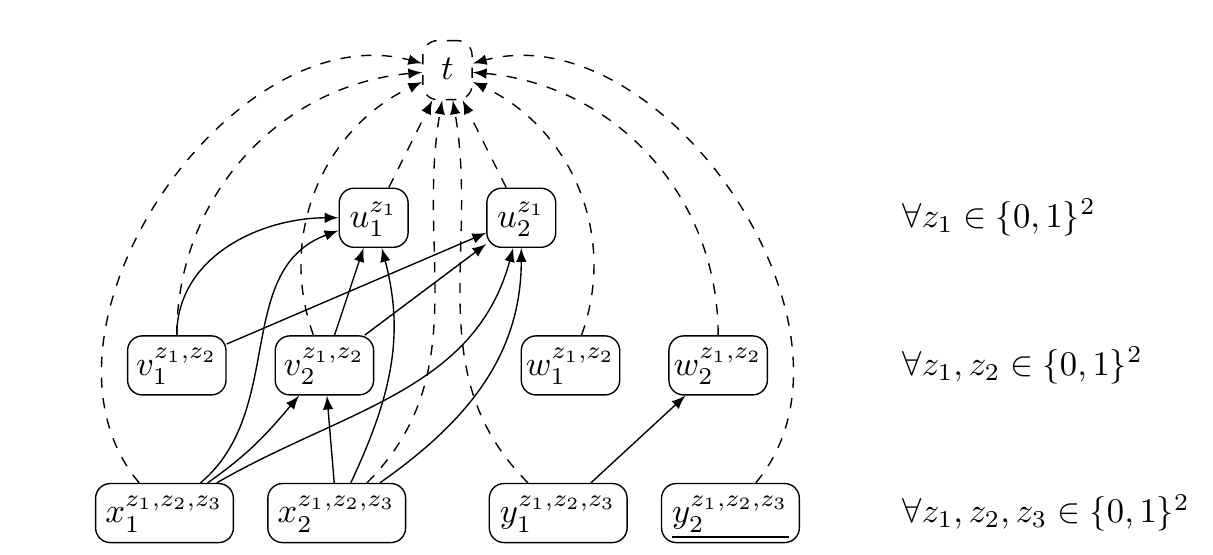}
\end{center}
   \caption{$G'$ consists of all nodes and edges drawn via solid lines. For clarity, each rectangle denotes a set of nodes $V_S$ (\Cref{eqn:VS}) corresponding to some supervertex $S$. For example, $u_1^{z_1}$ denotes a set of nodes $\set{u_1^{00},u_1^{01},u_1^{10},u_1^{11}}$, whose neighbor sets are defined via \Cref{eqn:newE}. To move from $G'$ to $G''$, we add node $t$ and all dashed edges.}
   \label{fig:transform2}
\end{subfigure}

\begin{subfigure}[b]{\textwidth}
\begin{center}
    \includegraphics{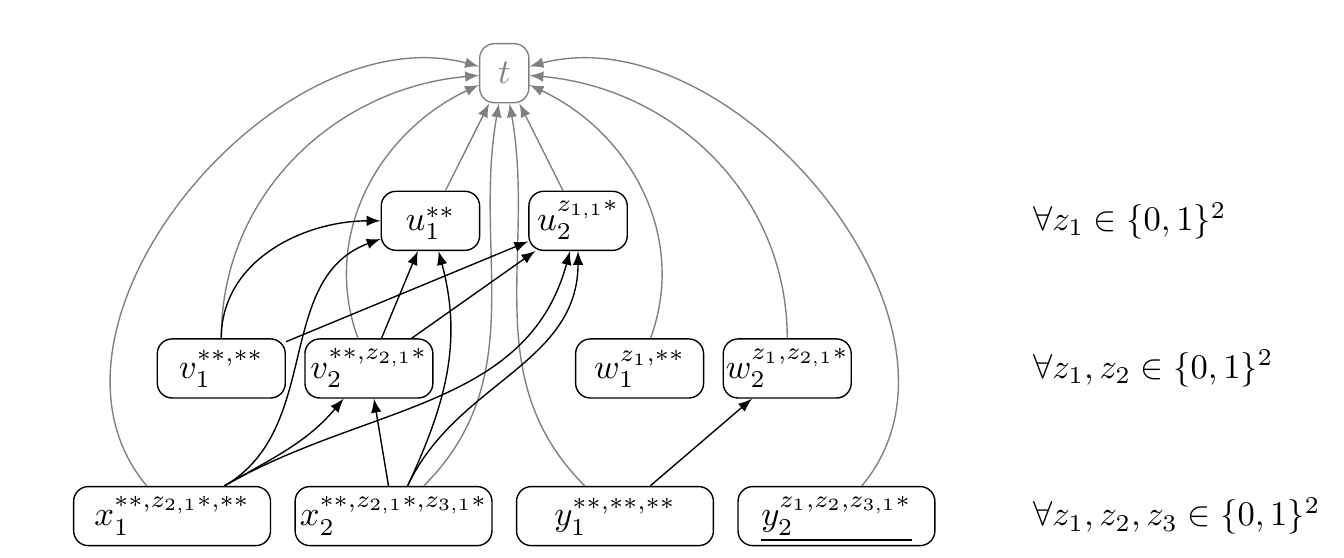}
\end{center}
    \caption{Graph $G^*$ with merged nodes indicated by asterisks in the superscript.}
    \label{fig:transform3}
\end{subfigure}
\caption{Example of the query graph transformation.}\label{fig:transform}
\end{figure}

\noindent\emph{Assigning queries to $G'$.} We have given a graph theoretic mapping $G\mapsto G'$, but not yet specified how the queries made at nodes of $G$ are mapped to queries made at nodes of $G'$.
Let us do so now.
Consider any $v_{S,i}^{z_1,\dots,z_d}\in V'$.
Roughly, the goal is for the query at $v_{S,i}^{z_1,\dots,z_d}$ to simulate the query at $\inv(v_{S,i}^{z_1,\dots,z_d}) = u_{S,i}$.
However, $v_{S,i}^{z_1,\dots,z_d}$ is ``conditioned'' on bit strings $z_1,\ldots, z_d$, so the simulation is not straightforward.
To make this formal, we use \Cref{alg:compute-inputs} as follows:
\begin{rul}\label{rule:wires}
    For each edge $(u_{T,j},u_{S,i})$ in $G$, the result of $\Call{ComputeOutput}{u_{T,j}\mid z_1,\dots,z_d}$ is used as the corresponding input to $v_{S,i}^{z_1,\dots,z_d}\in V'$.
\end{rul}
\noindent Intuitively, we may view the conditioning string $z_1,\ldots, z_d$ as specifying a ``parallel universe'', where if $(u_{T,j},u_{S,i})$ was an edge in $E$, then this parent-child relationship is simulated \emph{relative to this parallel universe} via the $\Call{ComputeOuput}$ function.

\begin{remark}
(1) By definition of a separator tree, all edges $(u_{T,j},u_{S,i})$ of $G$ must be in the same branch as $S$ (i.e. either above or below $S$ in the same branch, but not in a parallel branch of the tree).
(2) This implies there are essentially two cases to consider: When $u_{T,j}$ is closer to the root than $u_{S,i}$, or vice versa.
In the first case, Line~\ref{alg:compute-inputs:early-return} of \Cref{alg:compute-inputs} immediately returns the hardcoded bit of $z_1,\ldots, z_d$ corresponding to $u_{T,j}$.
In the second case, when $u_{S,i}$ calls \Cref{alg:compute-inputs}, Lines 6-8 will recursively compute outputs of nodes below $v_{S,i}^{z_1,\dots,z_d}$ in the same branch.
For this, $v_{S,i}^{z_1,\dots,z_d}$ will need access to the $\outputc$ functions (\Cref{def:c-dag}) of certain nodes below it; this is afforded to $v_{S,i}^{z_1,\dots,z_d}$ via the edge set $E_S$ (demonstrated via the ``upward'' black edges in \Cref{fig:transform2}).

\end{remark}

\begin{algorithm}[t]
\caption{Compute the output of $u_{S,i}$, conditioned on results $z_1,\dots,z_m$ in the separators above.}\label{alg:compute-inputs}
\begin{algorithmic}[1]
    \Function{ComputeOutput}{$u_{S,i}\mid z_1,\dots,z_{m}$} \Comment{recall $u_{S,i}\in S$}
        \State $S_1,\dots,S_{d} \gets$ path from the root to $S$
        \If{$m\geq d$}   \Comment{base case of recursion; recursion has computed $z_d$}
            \State \Return $z_{d,i}$ \Comment{recall $z_d\in \set{0,1}^s$; $z_{d,i}$ encodes answer to $u_{S,i}$} \label{alg:compute-inputs:early-return}
        \EndIf
        \State $z_{m+1}\gets 0^s$   \Comment{initialize answer bits to all zeroes to start} \label{alg:compute-inputs:z-init}
        \For{$j=1,\dots,s$}         \Comment{in topological order, set answer bits at current level of recursion} \label{alg:compute-inputs:loop}
            \State $z_{m+1,j} \gets \outputc\left[u_{S_{m+1},j}^{z_1,\dots,z_{m+1}}\right]$ \Comment{set bit $j$ of $z_{m+1}$ using query answers on incoming edges} \label{alg:compute-inputs:set-output}
        \EndFor
        \State\Return\Call{ComputeOutput}{$u_{S,i}\mid z_1,\dots,z_{m+1}$}\label{alg:compute-inputs:recursion}
    \EndFunction
\end{algorithmic}
\end{algorithm}

We have now specified the local input/output behavior of any node $v\in V'$.
Two problems remain: First, we require a designated output node in $G'$, which implicitly orchestrates the new logic in $G'$.
Second, observe in \Cref{fig:transform2} that the original output node of $G$, $y_2$, has been mapped to a new set of nodes labelled $y_2^{z_1,z_2,z_3}$, all of which are \emph{disconnected} from the rest of $G'$.
Thus, we require a mechanism to stitch together these components of $G'$.
To solve both problems simultaneously, we define $G''$ by adding a new output node, $t$, such that: (1) $t$ has incoming edges from all nodes in $V'$, and (2) the output of $G''$ is computed by having $t$ call $\Call{ComputeOutput}{v\mid \epsilon}$ and return its answer, where $\epsilon$ denotes the empty string and $v$ is the original output node of $G$.
Both $t$ and these new edges are depicted in \Cref{fig:transform2} via dashed lines.\\
\vspace{-1mm}

\noindent\emph{Intuition.} The construction of $G''$, and why it does what we need, is subtle; so let us illustrate via \Cref{fig:transform2}.
For this, recall in \Cref{fig:transform1} that $y_2\in V$ is the original output node of $G$.
For concreteness, assume in this discussion that $C=\NP$ and $s\in O(1)$, so that \Cref{thm:bsn-s} states $\CDAG_1 \subseteq \p^{\NP[\log]}$. The intuition is as follows:
\begin{enumerate}
    \item To apply admissible weighting functions and prove \Cref{thm:separator-tree-time}, a property we require\footnote{While necessary, this property itself is \emph{not} sufficient; we use it here to ease the discussion. More accurately, we require that for any $v$, $\abs{\desc(v)}\le O(\log n)$. This latter property is trickier to attain, and does not follow from $D\in O(\log n)$.} is that the length of any directed path in $G$ be at most $D\in O(\log n)$.
        However, in the separator tree decomposition of \Cref{fig:transform1}, the longest path can in principle have $O(n)$ edges.

    \item To address this, \Cref{fig:transform2} removes all ``downward edges'' with respect to the separator decomposition (i.e. edges $(v,w)\in E$ such that $\textup{level}(v)<\textup{level}(w)$ in the tree).
        Thus, the longest path now goes from a leaf to the root $t$, with each edge followed monotonically decreasing the current level\footnote{More accurately, the level is non-increasing with each edge followed. This is because the construction allows edges between pairs of vertices in the same supervertex (\Cref{eqn:newE}). This can incur an overhead in path length scaling with $s$, the separator size, which we ignore for this intuitive discussion.}. Since our separators are \emph{balanced}, any such path has length $O(\log n)$, as desired.

    \item Of course, this breaks the logic of the C-DAG $G$ itself.
    To correct this, we apply three ideas.
    \begin{enumerate}
        \item \emph{Create node copies.} Each node of $G$ (say, $x_1$) is split into multiple copies, each of which is hardcoded with a distinct possible output value of all its ``ancestors'' in the separator tree.
    In this example, $x_1$ is split into $64$ copies of form $x_1^{z_1,z_2,z_3}$, over all $z_1,z_2,z_3\in\set{0,1}^2$.
    Here, $z_1$ is intended to capture the outputs of $u_1$ and $u_2$, $z_2$ the outputs of $v_1$ and $v_2$, and $z_3$ the outputs of $x_1$ and $x_2$.
    Of these, $x_1$ only depends directly on the first bit of $z_2$ according to $G$; all other bits in $z_1,z_2,z_3$ are irrelevant for $x_1$, and are included only to make the construction systematic.
    (They will be removed shortly when moving to $G^*$ in \Cref{sscn:merging-nodes}.)

        \item \emph{Add upward shortcuts.} Add new ``upward'' edges via \Cref{eqn:newE}. In words, this roughly means that if $u$ is an ancestor of $v$ in $G$, but $v$ occurs closer to the root than $u$, then we add upward shortcut edge $(u,v)$ to $E'$.
        In our example, we connect $x_1^{z_1,z_2,z_3}$ to all (copies of) descendants of $\inv(x_1^{z_1,z_2,z_3})$ which are \emph{in the same unique superpath} from $x_1^{z_1,z_2,z_3}$ up to $t$ (such as $v_2^{z_1,z_2}$).
        The careful reader may notice that $x_1^{z_1,z_2,z_3}$ does \emph{not} have an edge to $x_2^{z_1,z_2,z_3}$.
        This is why $x_2$ has superscript $z_3$; this enumerates over all possible outputs it may have received from $x_1$.

        \item \emph{Orchestrate the madness.} All copies of all nodes send their output to $t$ via the dashed edges in \Cref{fig:transform2}.
        Roughly, $t$ now selects, out of all the possible computation paths created via node copies, which is the ``right'' path.
        In our example, the ``right'' path can start with any copy of $v_1$ (e.g. $v_1^{00,00}$ or $v_1^{11,10}$, etc), since $v_1$ has in-degree $0$ in $G$. Thus, all copies of $v_1$ encode the same NP query.
        Suppose this NP query outputs $b\in\set{0,1}$.
        Then, the ``right'' path next utilizes any copy of $x_1$ of form $x_1^{z_1,bz_{2,2},z_3}$ (first bit of $z_2$ is $b$). And so forth.
        This ``selection'' of the ``right path'' is executed when $t$ calls $\Call{ComputeOutput}{y_2\mid \epsilon}$.
    \end{enumerate}
\end{enumerate}

\vspace{3mm}
\noindent \emph{An explicit run-through.} For concreteness, we now trace through $t$'s call to $\Call{ComputeOutput}{y_2\mid \epsilon}$ for \Cref{fig:transform2}:

\begin{center}
\begin{minipage}[c]{0.6\textwidth}
\begin{algorithmic}[1]
    \Blk{$\Call{ComputeOutput}{y_2\mid \epsilon}$:}
        \State $z_{1,1}\gets \outputc[u_1^{00}]$
        \State $z_{1,2}\gets \outputc[u_2^{z_{1,1}0}]$
        \Blk{\Return $\Call{ComputeOutput}{y_2\mid z_1}$:}
            \State $z_{2,1}\gets \outputc[w_1^{z_1,00}]$
            \State $z_{2,2}\gets \outputc[w_2^{z_1,z_{2,1}0}]$
            \Blk{\Return $\Call{ComputeOutput}{y_2\mid z_1,z_2}$:}
                \State $z_{3,1}\gets \outputc[y_1^{z_1,z_2,00}]$
                \State $z_{3,2}\gets \outputc[y_2^{z_1,z_2,z_{3,1}0}]$
                \Blk{\Return $\Call{ComputeOutput}{y_2\mid z_1,z_2,z_3}$:}
                    \State \Return $z_{3,2}$
                \EndBlk
            \EndBlk
        \EndBlk
    \EndBlk
\end{algorithmic}
\end{minipage}
\end{center}

\begin{remark}[Promise gaps]\label{rem:promisegaps}
    While the \emph{size} of the verification circuit at any node $u_{S,i}\in V$ grows under the mapping to $v_{S,i}^{z_1,\dots,z_d}\in V'$ (since the latter takes in more wires), the underlying verification procedures at each $u_{S,i}$ and $v_{S,i}^{z_1,\dots,z_d}$ are identical, up to the latter's use of \Cref{rule:wires} to decide on-the-fly which input wires to use based on $(z_1,\dots,z_d)$. Thus, the \emph{promise gaps} at each $u_{S,i}$ and $v_{S,i}^{z_1,\dots,z_d}\in V'$ are also identical. When $C$ is a promise class allowing error reduction, this is not of consequence; however, for StoqMA, which is not known to have error reduction, this observation allows us to keep the exponents in \Cref{thm:bsn-s-stoqma} at $O(s(n)\log^2n)$ (versus $O(s^2(n)\log^2n)$, since the promise gaps at each $v_{S,i}^{z_1,\dots,z_d}$ node still scale as $1/\poly(n)$ due to this observation, not $1/\poly(\abs{V''})$.
\end{remark}

\begin{remark}[For StoqMA]\label{rem:stoqma}
      $\Call{ComputeOutput}{v\mid \epsilon}$ (and thus $t$) runs in $\DTIME(\poly(V''))$ to stitch together the answers of all other nodes of $V$. Thus, when $C=\StoqMA$, the action of $t$ can also be viewed as a special case of a StoqMA computation (i.e. the StoqMA ``verifier'' for $t$ would ignore its proof, use its classical gates to simulate $\Call{ComputeOutput}{v\mid \epsilon}$, and then output either $\ket{+}$ or $\ket{-}$ depending on whether it wishes to accept or reject, respectively.) Thus, in this case, all nodes of $G''$ are valid $C=\StoqMA$ nodes, so $G''$ is a valid $\StoqMA$-DAG.
\end{remark}

\subsubsection{Merging Nodes (\texorpdfstring{$G^*$}{G^*})}\label{sscn:merging-nodes}

Next, we address the issue that $G$ and $G''$ are not necessarily equivalent for promise problems, since copies of the same invalid query could have different outputs.
For example, in \Cref{fig:transform2}, if $z_{1,1}=1$, then $u_1^{00}$ and $u_2^{z_{1,1}0}$ depend on different copies of $v_2$, which could lead to inconsistencies if $v_2$ encodes an invalid query. In addressing this, we will also remove redundant copies of nodes (e.g. $v_1$, which has in-degree $0$ in \Cref{fig:transform1}, encodes the same query in \Cref{fig:transform2}, regardless of how $z_1$ and $z_2$ are set).

To proceed, we construct graph $G^*$ by merging node copies which have the same hard-coded inputs.
Consider any $u:=u_{S,i}\in V$, where as in $\Cref{alg:compute-inputs}$, we let $d$ denote the depth of $S$ on the unique superpath $P_S:=(S_1,\ldots, S_d=S)$ from the root $S_1$ to $S$ in the separator tree.
Recalling that $\anc(u)$ denotes the ancestors of $u$ in $G$, i.e. the set of queries $u$ depends on, define
\be\label{eqn:Du}
    D_u := \anc(u)\cap \bigcup_{j=1}^d S_j,
\ee
in words, the ancestors of $u$ in the superpath $P_u$.
For any $v:=v_{S,i}^{z_1,\dots,z_d}\in V''$, define $h_v: D_u\to\bin$ with action $h_v(u_{S_j,k}) := z_{j,k}$, i.e. $h_v$ selects out the hard-coded bit $z_{j,k}$ corresponding to any $u_{S_j,k}\in D_u$ (i.e. $z_{j,k}$ is the $k$th bit of $z_j$ in the definition of $v$).
Now, whenever two copies $v_1,v_2\in V''$ of the same node (i.e. $\inv(v_1) = \inv(v_2)$) satisfy $h_{v_1} = h_{v_2}$ (i.e. $h_{v_1}$ and $h_{v_2}$ have the same truth table; note $D_u$ is in the original graph $G$ in definition $h_v:D_u\to\bin$), we will merge them.
Formally, the merge is accomplished by \Cref{alg:merge}, which simultaneously computes an admissible weighting function.
Henceforth, denote $(G^*,f^*) := \Call{Merge}{G''}$.

\begin{algorithm}[htb]
    \caption{Merge nodes in $G''$ to compute $G^*$.}\label{alg:merge}
    \begin{algorithmic}[1]
        \Function{Merge}{$G''=(V'',E'')$}
            \State $G_1 \gets G'',\; V_1 \gets V'',\; E_1\gets E'',\; f_1 \gets \omega$ \Comment{for weighting function $\omega$ from
            \Cref{lem:weighting-function}}
            \State $i\gets 1$
            \While{$\exists v_1,v_2\in V_{i}$ such that $\inv(v_1) = \inv(v_2)$ and $h_{v_1} = h_{v_2}$}
                    \State Choose any such $v_1,v_2$ such that $v:=\inv(w_1)$ is furthest from root in separator tree $T$.\label{alg:merge:choose}
                    \State Create copy $v^*$ of $v$ with $h_{v^*}:=h_{v_1}=h_{v_2}$.
                    \State $V_{i+1} \gets V_{i}\setminus\{v_1,v_2\}\cup\{v^*\}$
                    \State Replace $\outputc[v_1]$ and $\outputc[v_2]$ in the logic of nodes in $V_{i+1}$ with $\outputc[v^*]$.
                    \State $E_{i+1} \gets \bigl\{ (r(x),r(y)) \bigm\vert (x,y)\in E_i\bigr\}$, where $r(x) := \begin{cases}
                        v^*, & \text{if } x \in \{v_1,v_2\}\\
                        x, & \text{else}
                    \end{cases}$
                    \State Update $f_{i+1}:V_{i+1}\to\R$ such that $f_{i+1}(x) := \begin{cases}
                        f_{i}(v_1) + f_{i}(v_2), & \text{if } x = v^*\\
                        f_{i}(x), & \text{else}
                    \end{cases}$\label{alg:merge:update f}
            \EndWhile

            \State \Return $(G_i,f_i)$
        \EndFunction
    \end{algorithmic}
\end{algorithm}

\begin{remark}\label{rem:alwaysmerge}
    When $u\in V$ has in-degree $0$, then $D_u=\emptyset$. In this case, for any two copies $v_1,v_2\in V''$ of $u$, it is vacuously true that $h_{v_1}=h_{v_2}$. Thus, all copies of $u$ in $V''$ are merged by \Cref{alg:merge}. An example of this is depicted by $v_1$ in  \Cref{fig:transform1} being mapped to $v_1^{**,**}$ in \Cref{fig:transform3}. Intuitively, this captures the fact that since $v_1$ is in-degree $0$ in \Cref{fig:transform1}, the query at all copies $v_1^{z_1,z_2}$ of \Cref{fig:transform2} is identical, regardless of how $z_1,z_2$ are set.
\end{remark}

\begin{lemma}\label{lem:merge-admissible}
    For any constant $c\geq 2$, the weighting function $f^*$ produced by \Cref{alg:merge} is $c$-admissible for $G^*$, and satisfies $W_{f^*}(G^*)=W_\omega(G'')$ (for $\omega$ from \Cref{lem:weighting-function}).
\end{lemma}
\begin{proof}
    We prove the lemma inductively on the iteration number, $i$.
    By \Cref{lem:weighting-function}, $f_1=\omega$ is $c$-admissible, and trivially $W_{f_1}(G_1)=W_\omega(G'')$.
    In the induction step, Line~\ref{alg:merge:update f} of \Cref{alg:merge} straightforwardly yields $W_{f_i}(G_i)=W_{f_{i+1}}(G_{i+1})$.
    We show $c$-admissibility for $f_{i+1}$, i.e. that $f(v)\geq 1+ c \sum_{w\in\Gamma(v)}f(w)$ (\Cref{eq:admissible}) holds for all $v\in V_{i+1}$, where recall $\Gamma(v)$ is the set of children of $v$.
    Since the admissibility condition only depends on $\Gamma(v)$, it suffices to consider two cases: $v=v^*$ and $v$ is a parent of $v^*$.
    First, if $v=v^*$, we have $\Gamma(v^*)= \Gamma(v_1)  \cup \Gamma(v_2) $, and thus
    \begin{equation}\label{eqn:stillcadmiss}
        f_{i+1}(v^*) = f_i(v_1) + f_i(v_2) \ge 2+c\sum_{u\in \Gamma(v^*)} f_{i}(u) = 2+c\sum_{u\in \Gamma(v^*)} f_{i+1}(u)
    \end{equation}
    since $f_i$ was $c$-admissible and since only $v^*$ is altered in round $i$. 
    Second, if $v$ is a parent of $v^*$, then $v$ is a parent of at least one of $v_1$ or $v_2$ by the construction of \Cref{alg:merge}.
    But by definition of the edge set of $V''$ (\Cref{eqn:newE}), $v$ is a parent of $v_1$ if and only if $v$ is a parent of $v_2$.
    Thus, $v$ was a parent of both $v_1$ and $v_2$ in round $i$.
    The claim now follows since we set $f_{i+1}(v^*)=f_i(v_1)+f_i(v_2)$.
    \end{proof}

\subsubsection{Correctness}\label{sscn:correctness}

We now prove correctness, in the process establishing the Compression Lemma (\Cref{l:compress}).
For this, we first require the following lemma, which shows how to efficiently map any given correct query string for $G^*$ to a correct query string for $G$.
Below, \Call{ComputeOutput}{} on $G^*$ takes into account merged notes, i.e. it uses $v^*$ instead of $v$ after merging $v_1$ and $v_2$ in \Cref{alg:merge}.

\begin{lemma}\label{lem:correctness}
    Let $x^*:V^*\to \bin$ be a correct query string for $G^*$. Define \COs{} to be $\Call{ComputeOutput}{}$, except with each call to $\outputc$ on Line~\ref{alg:compute-inputs:set-output} replaced by looking up the corresponding bit of $x^*$.
    Define string $x:V\to\bin$ such that bit $x(v):=\Call{ComputeOutput*}{v\mid\epsilon}$.
    Then, $x$ is a correct query string for $G$.
\end{lemma}
\begin{proof}
    Recall $\abs{V}=n$, and that by \Cref{def:correct}, a correct query string for $\CDAG$ is defined as any string producible by Line~\ref{alg:c-dag-eval:x_i} of \Cref{alg:c-dag-eval}.
    Throughout, the bits of $x$ are ordered according to the topological order $(v_1,\ldots, v_n)$ on $V$ fixed by \Cref{alg:c-dag-eval}.
    We prove the claim inductively for $t\in (1,\ldots,n)$.

    By the topological order, the base case $v_1\in V$ has in-degree $0$, i.e. takes no inputs.
    Thus, by \Cref{rem:alwaysmerge}, there is only a single node in $G^*$ corresponding to $v_1$, which by construction computes the same query as $v_1$.
    Hence, the corresponding bit of $x^*$ trivially encodes the correct answer for $v_1$ in $G$.
    This bit will then be returned for $v_1$ once Line~\ref{alg:compute-inputs:early-return} executes, as desired.

    For the inductive case, let $t\geq 2$ and assume $x_1,\dots,x_{t-1}$ satisfy the induction hypothesis, i.e. they could\footnote{We say ``could'' because \Call{Evaluate}{} is nondeterministic (due to potential invalid queries when $C$ is a promise class).} be produced by the first $t-1$ iterations of \Call{Evaluate}{}.
    We now need to argue that a correct execution of \Call{Evaluate}{} could set $x_t = \COs{v_t\mid\epsilon}$.
    By design, $\COs{v_t\mid\epsilon}$ computes $z_1,\dots,z_d$ and then returns\footnote{Technically, the algorithm actually returns $x^*(v_t^{z_1,\dots,z_{d-1},z'_d})$, where $z_d'\in\set{0,1}^s$ is an ``intermediate string'' defined as follows: If node $v_t$ was the $k$th node in the topological order for supervertex $S$, then the first $k-1$ bits of $z_d'$ have been assigned by Line~\ref{alg:compute-inputs:set-output} of \COs{}, and the remaining $s-k+1$ bits of $z_d'$ are still set to the dummy value of $0$ from Line~\ref{alg:compute-inputs:z-init}. However, this does not affect our analysis. In particular, in $G^*$ we merged $v_t^{z_1,\dots,z'_d}$ and $v_t^{z_1,\dots,z_{d-1},z''_d}$ for any such pair $z'_d$ and $z''_d$ of ``intermediate strings'', since by definition of the topological order, bits $k$ through $s$ of any such $z'_d$ cannot correspond to any vertices in $D_{v_t}$. Thus, these indices effectively disappear for all copies of $v_t$ in $G^*$.} $x^*(v_t^{z_1,\dots,z_d})$.

    Recall now that $D_{v_t}$ denotes the ancestors of $v_t$ in $G$, which are also along the superpath from the root of the separator tree down to the supervertex $S$ containing $v_t$.
    (Formally, $D_u := \anc(u)\cap \bigcup_{j=1}^d S_j$ in \Cref{eqn:Du}.)
    We claim that $z:=z_1,\dots,z_d$ matches $x_1\dotsm x_{t-1}$ on $D_{v_t}$.
    (For clarity, the bits of $z\in\set{0,1}^{sd}$ are ordered according to the recursion of \COs, and may also contain bits corresponding to vertices not in $D_{v_t}$.)
    To see this, fix any $u\in D_{v_t} =\anc(v_t)\cap \bigcup_{j=1}^d S_j$, and let $i$ be the supervertex index such that $u\in S_i$.
    For brevity, define notation $x(u)$ and $z(u)$ to mean the bit of $x$ and $z$ corresponding to $u$, respectively\footnote{Since $x$ is defined on $G$, $x(u)$ is clearly defined uniquely. On the other hand, $z$ is defined on $G^*$, which contains potentially multiple copies of $u$; thus, it is slightly more subtle that $z(u)$ is uniquely defined. Indeed, uniqueness holds since the recursive path followed by \COs{} through $G^*$ visits \emph{precisely one} copy of $u$; \emph{which} copy is visited depends on the prefix of $z$ fixed in the recursion before $u$ is encountered.}.
    Now, recall $\COs{v_t\mid\epsilon}$ recursively traverses the path $S_1,\ldots, S_d$, where $u\in S_i$ and $v_t\in S_d$ for $1\leq i\leq d$.
    Thus, the operations performed by $\COs{v_t\mid \epsilon}$ and $\COs{u\mid \epsilon}$ are identical in the first $i$ recursions.
    Once $m=i$ (i.e. conditioning strings $z_1, \ldots, z_i$ have been set), $z(u)$ is returned by $\COs{u\mid \epsilon}$ on Line~\ref{alg:compute-inputs:early-return}, whereas $\COs{v_t\mid \epsilon}$ returns $z(v_t)$ if $i=d$ and continues recursively otherwise.

    We thus conclude $\COs{v_t\mid\epsilon}$ returns $x^*(v_t^{z_1,\dots,z_d})$ with $z(u) = x(u)$ for all $u\in D_{v_t}$.
    Recall now by \Cref{rule:wires} that, in order for $v_t^{z_1,\dots,z_d}\in V^*$ to simulate the query at $v_t\in V$, it computes the input on any wire $u$ into $v_t$ via $\COs{u\mid z_1,\dots,z_d}$.
    Since the construction is based on a separator tree, this incoming wire/edge $(u,v_t)$ lies along the same branch of the tree as $v_t$. Thus, we have only two cases to consider --- $u\in\anc(v_t)$ is above or below $v_t$ in said branch.
    (In \Cref{fig:transform1}, for example, if $v_t=w_2$, the ancestor $u_2$ is above $w_2$ in the tree, whereas ancestor $y_1$ is below $w_2$.)
    So, if $u\in D_{v_t}$ (i.e. $u$ is above $D_{v_t}$), then by the argument in the previous paragraph, $z(u)=x(u)$ is used as input to $v_t$.
    Moreover, since $u\in\anc(v_t)$, $u$ comes before $v_t$ in any topological order, and thus the induction hypothesis says $x(u)$ is correct.
    Otherwise, if $u\not\in D_{v_t}$ (i.e. ancestor $u$ is below $D_{v_t}$), then it again holds that $\COs{v_t\mid \epsilon}$ and $\COs{u\mid \epsilon}$ perform exactly the same operations in the first $d$ recursions.
    Therefore, both executions compute the same values $z_1,\dots,z_d$.
    Subsequently, $\COs{u\mid \epsilon}=x(u)$ calls $\COs{u\mid z_1,\dots,z_d}$ on Line~\ref{alg:compute-inputs:recursion} during the $d$th recursion and returns its value.
    In other words, $x(u)=\COs{u\mid z_1,\dots,z_d}$.
    But by \Cref{rule:wires}, in order to simulate its input on the incoming wire corresponding to $u$, $v_t^{z_1,\dots,z_d}$ uses $\COs{u\mid z_1,\dots,z_d}=x(u)$.
    Then, since $u\in \anc(v_t)$, $u$ comes before $v_t$ in any topological order, and thus by the induction hypothesis, $x(u)$ is correct.
    We hence conclude all input wires to $v_t^{z_1,\dots,z_d}$ must be set correctly, and thus $x(v_t)$ is also correct.
\end{proof}

We finally restate and prove the main lemma of this section.

\lemmaCompress*
\begin{proof}[Proof of \Cref{l:compress}]
    $G^*$ is constructed as in \Cref{sscn:basic-construction} and \Cref{sscn:merging-nodes}.
    We have $\abs{V^*} \le \abs{V''} \le 2^{O(sD)}n$ (recall $n=\abs{V}$), since there are at $2^{O(sD)}$ choices for conditioning strings $z_1\cdots z_D$.
    Since we assume the separator tree is given as input, the time to construct $G^*$ is clearly polynomial in $\abs{V''}$, i.e. $2^{O(sD+\log n)}$.
    For weighting function $\omega$ from \Cref{lem:weighting-function}, we have $W_{f^*}(G^*)=W_\omega(G'')  \le (c+1)^{O(sD)}n$, where the equality is from \Cref{lem:merge-admissible}, and the inequality since every node in $G''$ has at most $O(sD)$ descendants.
    Correctness follows from \Cref{lem:correctness} and the fact that the output of $G^*$, by definition of node $t$, is $\Call{ComputeOutput}{v\mid \epsilon}$.
    Finally, each verification circuit corresponding to a node in $V^*$ has size $\poly(V^*)$; the largest such verification circuit corresponds to node $t$, which takes in wires from \emph{all} other vertices in $G^*$, and calls $\Call{ComputeOutput}{v\mid \epsilon}$ (which takes time $\poly(\abs{V^*})$).
\end{proof}

\subsection{Solving \texorpdfstring{\CDAG}{C-DAG} via oracle queries}\label{sscn:solve-c-dag}


We now show how to decide an $\nn$-node \CDAG-instance $G$ with a $c$-admissible weighting function $f$ using $O(\log W_f(G))$ oracle queries.
(We intentionally use $\nn$ to denote the size of $G$, to avoid confusion with the parameter $n$ from \Cref{sscn:graph-transformation}.)
Recall that at a high level, our aim is to convert the problem of deciding $G$ into the problem of maximizing a carefully chosen real-valued function $t$.
A binary search via the oracle $C$ is then conducted to compute the optimal value to $t$, from which a correct query string from $G$ can be extracted.
This high-level strategy was also used by Gottlob~\cite{Gottlob1995}; what is different here is how we define $t$ and how we implement the details of the binary search.


\subsubsection{Step 1: Defining the total solution weight function $t$}\label{sscn:totalsolnweight}

Let $G = (V=\{v_1,\dots, v_\nn\}, E)$ be a \CDAG-instance with $c$-admissible weighting function $f$.
Recall by \Cref{def:c-dag} that each circuit $Q_i$ has a proof register $\Y_i$.
Without loss of generality, we assume $Q_i$ receives a proof $\ket{\psi_i}\in\B^{\otimes m}$ and has completeness $\alpha$ and soundness $\beta$.
Then, define $t:\bin^\nn\times\left(\B^{\otimes m}\right)^{\times \nn}\mapsto \R$ such that
\begin{align}\label{eqn:t}
  t(x,\psi_1,\dots,\psi_\nn) := \sum_{i=1}^\nn f(v_i) \bigl(\underbrace{\strut x_i\Pr[Q_i(z_i(x),\psi_i)=1] + (1-x_i)\gamma}_{\textstyle g(x_i, z_i(x), \psi_i)}\bigr),
\end{align}
where $\gamma := (\alpha+\beta)/2$, and where $z_i(x)$ is defined similar to Line~\ref{alg:c-dag-eval:z_i} of \Cref{alg:c-dag-eval}, i.e. $z_i(x) \gets \bigcirc_{v_j\in \pred(v_i)}\, x_j$, for $x$ the input string to $t$.
Two comments are important here:
First, defining $z(x)$ in this manner may break the logic of \Cref{alg:c-dag-eval} when a prover is dishonest, in that the relationship between $x_i$ and $z_i$ of Line~\ref{alg:c-dag-eval:x_i} may not hold.
Nevertheless, in \Cref{sec:correct-query-string}, we prove that $t$ is {maximized} only when a prover acts \emph{honestly}.
Second, we intentionally define $t$ as taking in a \emph{cross product} over spaces $\B^{\otimes m}$, as opposed to a {tensor product}.
This simplifies the proofs of this section.
Finally, define
\begin{equation}
T := \max_{\substack{x\in\bin^\nn\\\ket{\psi_1},\dots,\ket{\psi_\nn}\in\B^{\otimes m}}} t(x,\psi_1,\dots,\psi_\nn).\label{eqn:T}
\end{equation}
In \Cref{sec:approximating-t} we show how to approximate $T$ using $O(\log W_f(G))$ $C$-queries and in \Cref{sec:correct-query-string} we prove that if $t(x,\psi_1,\dots,\psi_\nn) \approx T$, then $x$ is a correct query string.

\subsubsection{Step 2: Approximating \texorpdfstring{$T$}{T}}\label{sec:approximating-t}

In order to apply binary search to approximate $T$ (see proof of \Cref{thm:separator-tree-time}), we now show that the \emph{decision} version of approximating $T$ is in $C$. Namely, define promise problem $\Pi_\epsilon=(\Piy,\Pin)$ such that
\begin{eqnarray}
    \Piy &=& \{(t,s)\mid t:\bin^\nn\times\left(\B^{\otimes m}\right)^{\times \nn}\mapsto \R \text{ and }T \ge s\}\\
    \Pin &=& \{(t,s)\mid t:\bin^\nn\times\left(\B^{\otimes m}\right)^{\times \nn}\mapsto \R \text{ and }T \le s-\epsilon\},
\end{eqnarray}
for $T$ as in \Cref{eqn:T}, and $\epsilon:\ZZ\mapsto\R^{\geq 0}$ a fixed function of $N$ (i.e. by $\epsilon$ we mean $\epsilon(\nn)$).

\begin{lemma}\label{lem:approximating-t}
  Let $C\in\QVp$. Define $W := \sum_{i=1}^\nn f(v_i)$ for weighting function $f$ from \Cref{eqn:t}, and assume $W\leq \poly(\nn)$.
  Then, for any $\epsilon\ge 1/\poly(\nn)$, $\Pi_\epsilon\in C$.
\end{lemma}
\begin{proof}
  In the case $C\in \{\NP,\NEXP\}$, $\Pi_\epsilon$ can easily be solved in $C$ by just computing $t(x,\psi_1,\dots,\psi_N)$ directly (note that $t:\bin^\nn\times(\set{0,1}^{\otimes m})^{\times \nn}\mapsto \R$ in this case).

  For the remaining $C$, we begin by defining probabilities $p_i := f(v_i)/W$ and let
  \begin{equation}
    t'(x,\psi_1,\dots,\psi_\nn) := \frac{1}{W}t(x,\psi_1,\dots,\psi_\nn) = \sum_{i=1}^\nn p_i \cdot g(x_i, z_i, \psi_i),\label{eqn:t'}
  \end{equation}
  whose maximum over all inputs we denote as $T'$.
  We prove the claim by constructing a $C$-verifier $V$ such that
  \begin{equation}\label{eqn:T'}
    \max_{\text{proofs } \ket\psi} \Pr[V \text{ outputs }1\mid\ket\psi] = T'.
  \end{equation}
  Thus, when $(t,s)\in\Piy$ (resp., $(t,s)\in\Pin$), $V$ accepts with probability at least $s/W$ (resp., at most $(s-\epsilon)/W$), where $\epsilon/W\geq 1/\poly(\nn)$ since $W\leq \poly(\nn)$ by assumption.\footnote{For $C=\QMA(2)$, we implicitly use the nontrivial error reduction of $\QMA(2)$ due to Harrow and Montanaro~\cite{Harrow2013}.}

  $V$ has proof space $\X\otimes \Y_1\otimes\dotsm\otimes \Y_\nn$ with $\X=\B^{\otimes \nn}$ and $\Y_i=\B^{\otimes m}$.
  A subtle point here is that function $t$ takes as part of its input a sequence $(\ket{\psi_1},\ldots,\ket{\psi_\nn})$, whereas in \Cref{eqn:T'}, $V$ takes in a joint (potentially entangled) proof $\ket{\psi}$ across proof registers $\Y_1\otimes\cdots\otimes \Y_\nn$.
  However, due to the construction of $V$ below, we shall see that without loss of generality, \Cref{eqn:T'} is attained for tensor product states $\ket{\psi}=\ket{\psi_1}\otimes\cdots\otimes\ket{\psi_\nn}$, which is equivalent to sequence $(\ket{\psi_1},\ldots,\ket{\psi_\nn})$, as desired.

  Given proof $\ket{\psi}\in \X\otimes \Y_1\otimes\dotsm\otimes \Y_\nn$, $V$ acts as follows:
  \begin{algorithmic}[1]
    \State Measure $\X$ in standard basis to obtain string $x$.
    \State Select random $i$ according to distribution $p_i$.%
    \footnote{This requires being able to sample from the distribution $p_i$, which in general cannot be done efficiently.
    However, we can approximate $p_i \approx k_i/2^{\poly(\nn)}$, allowing efficient sampling without changing the distribution significantly.}
    \If{$x_i=1$}
      \State Run $Q_i$ with input $z_i(x)$ and proof register $\Y_i$.
    \Else
      \State Output $1$ with probability $\gamma$.
    \EndIf
  \end{algorithmic}
  Since (the POVM corresponding to) $V$ is block diagonal with respect to $\X$, $\Pr[V \text{ outputs }1\mid\ket\psi]$ is maximized by some $\ket\psi = \ket{x}_\X\ket{\psi'}_{\Y_{1,\ldots, \nn}}$.
  Then, since we only measure a single local verifier $Q_i$ (at random, Step 4), we have
  \begin{equation}
    \Pr[V \text{ outputs }1\mid\ket\psi] = t'(x,\sigma_1,\dots,\sigma_\nn)\quad\text{where}\quad  \sigma_i :=  \Tr_{\bigotimes_{j\ne i}\Y_j}(\ketbra{\psi'}{\psi'}).\label{eqn:sigma}
  \end{equation}
  But for any fixed $x$, this is maximized by choosing pure states $\sigma_i = \ketbra{\psi_i}{\psi_i}$. Thus, $t'$ is optimized by a tensor product $\ket{\psi}=\ket{\psi_1}\otimes\cdots\otimes\ket{\psi_\nn}$, and so \Cref{eqn:T'} holds. This completes the proof.

  Two final remarks are needed for specific choices of $C$: (1) For $C=QMA(2)$, separable $\sigma_i$ are obtained by interpreting the first proof register as $\X\otimes \Y^1_1\otimes\dotsm\otimes\Y^1_N$ and the second as $\Y^2_1\otimes \cdots\otimes\Y^2_N$ (so that the joint proof is unentangled across this cut by assumption), where proof $\Y_i$ has the registers $Y_i^1$ and $Y_i^2$. (2) For $C=\StoqMA$, $V$ is indeed a stoquastic\footnote{Briefly, a stoquastic verifier~\cite{Bravyi_Bessen_Terhal_2006} takes in a poly-size quantum proof and poly many ancillae set to $\ket{0}$ and $\ket{+}$ states, runs poly many classical gates (i.e. Pauli $X$, CNOT, and Toffoli gates), and finally applies a single Hadamard gate to its output qubit just before measuring it in the standard basis.} verifier since:
   \begin{itemize}
    \item Via its $\ket{+}$ ancillae states and ability to simulate measurement in the standard basis via the principle of deferred measurement, a stoquastic verifier can execute Steps 1 and 2 in the description of $V$.
    \item Each $Q_i$ is by definition a stoquastic verifier (see \Cref{rem:stoqma}), and Step 4 of the $V$ simply returns the output of some $Q_i$ \emph{without} postprocessing, i.e. the output qubit of $Q_i$ is simply swapped into the output qubit of $V$.
    \item By definition of StoqMA, $1/2\leq \gamma\leq 1$ with $\gamma$ requiring (without loss of generality) at most logarithmic bits of precision to specify. Thus, Step 6 of $V$ can be simulated by a stoquastic verifier which, with appropriate conditioning, swaps into its output qubit either an ancillae qubit set to $\ket{0}$ (accepted with probability $1/2$, due to the final $H$ gate on the measurement qubit of a stoquastic verifier) or an ancillae qubit set to $\ket{+}$ (accepted with probability $1$).
   \end{itemize}
\end{proof}

Note that \Cref{lem:approximating-t} says nothing about \emph{why} we want to approximate $T$, i.e. what the optimal argument $x$ buys us. This is the purpose of \Cref{sec:correct-query-string}.

\subsubsection{Step 3: Correct Query String}\label{sec:correct-query-string}

We next show that only \emph{correct} query strings $x$ can attain $T$ (even \emph{approximately}).

\begin{lemma}\label{lem:correct-query-string}
  Define $\eta := (\alpha-\beta)/2$, and let $f$ be $\eta^{-1}$-admissible.
  If $t(x,\psi_1,\dots,\psi_\nn) > T-\eta$, then $x$ is a correct query string.
\end{lemma}
\begin{proof}
  Assume there exists a $v_i\in V$ such that $x_i$ is incorrect.
  We show that there exist $x',\ket{\psi'_1},\dots,\ket{\psi'_\nn}\in\B^{\otimes \nn}$ such that $t(x',\psi_1',\dots,\psi_\nn') \ge t(x,\psi_1,\dots,\psi_\nn)+\eta$, obtaining a contradiction.
  A subtle but useful fact we exploit is that $t$ takes in a \emph{sequence} $(\psi_1,\ldots, \psi_\nn)$; this allows us to locally update each $\psi_i$ to some $\psi_i'$ as follows. Define $x'$,$\ket{\psi'_1}$,$\dots$,$\ket{\psi'_\nn}$ such that $x'_i = \overline{x_i}$ (i.e. the complement of $x_i$), $x'_j = x_j$ and $\ket{\psi'_j}=\ket{\psi_j}$ for $j\neq i$, and $\ket{\psi'_i}$ maximizing $\Pr[Q_i(z_i(x'),\psi'_i)=1]$.

  Now, if $x_i=0$, then
  \begin{equation}
    g(x_i, z_i(x), \psi_i) = x_i\Pr[Q_i(z_i(x),\psi_i)=1] + (1-x_i)\gamma = \gamma.
  \end{equation}
  Since we assumed $x_i$ was incorrect, $z_i(x)\in\Piy^i$. Thus, there exists a $\ket{\psi_i'}$ such that
  \begin{equation}
    g(x'_i, z_i(x'), \psi'_i) = x'_i\Pr[Q_i(z_i(x'),\psi'_i)=1] + (1-x'_i)\gamma = \Pr[Q_i(z_i(x'),\psi'_i)=1] \ge \alpha.
  \end{equation}
  Conversely, if $x_i=1$, we have $z_i(x)\in\Pin^i$, and thus
  \begin{equation}
    g(x_i, z_i(x), \psi_i) = x_i\Pr[Q_i(z_i(x),\psi_i)=1] + (1-x_i)\gamma \le \beta,
  \end{equation}
  whereas for any $\ket{\psi_i'}$,
  \begin{equation}
    g(x'_i, z_i(x'), \psi'_i) = x_i'\Pr[Q_i(z_i(x'),\psi_i)=1] + (1-x'_i)\gamma = \gamma.
  \end{equation}
Thus, flipping $x_i$ to $\overline {x_i}$ increases the $i$th term in the sum comprising $t$ by at least $\eta\cdot f(v_i)$ (recall $\gamma=(\alpha+\beta)/2$).
  Therefore,
  \begin{align}
    &t(x',\psi'_1,\dots,\psi'_\nn) - t(x,\psi_1,\dots,\psi_\nn) \\
    &\quad = \sum_{j=1}^\nn f(v_j) g(x'_j, z_j(x'), \psi'_j) - \sum_{j=1}^\nn f(v_j) g(x_j, z_j(x), \psi_j)\\
    &\quad = f(v_i)\bigl(\underbrace{g(x'_i, z_i(x), \psi'_i) - g(x_i, z_i(x), \psi_i)}_{\textstyle \ge \eta}\bigr) + \sum_{(v_i,v_j)\in E} f(v_j)\bigl(\underbrace{g(x'_j,z_j(x'),\psi_j) - g(x_j,z_j(x),\psi_j)}_{\textstyle \ge -1}\bigr)\\
    &\quad \ge \eta \cdot f(v_i) - \sum_{(v_i,v_j)\in E} f(v_j)\\
    &\quad = \eta\biggl( f(v_i) - \eta^{-1}\sum_{(v_i,v_j)\in E} f(v_j)\biggr)\\
    &\quad \ge \eta, 
  \end{align}
  where the second statement holds since $g(\cdot)\in[0,1]$ and since flipping $x_i$ to $x_i'$ only affects the immediate children of $v_i$ (since each $z_j$ function only depends on the direct inputs to node $v_j$), and the last statement since $f$ is $\eta^{-1}$-admissible.
\end{proof}

\subsubsection{Step 4: Completing the Proof}

We now combine everything to show the main technical result of \Cref{scn:BSN}, \Cref{thm:separator-tree-time}.
\begin{proof}[Proof of \Cref{thm:separator-tree-time}]
    First, apply the Compression Lemma (\Cref{l:compress}) to transform $G$ into an equivalent $G^*$ with $\abs{V^*}\le 2^{O(sD)}n$ and $W_{f^*}(G^*)\le (c+1)^{O(sD)}n$.
    This takes $2^{O(sD+\log n)}$ time.
    Second, define the total solution weight function $t$ as in \Cref{eqn:t}, whose maximum value we denoted $T$ (\Cref{eqn:T}). By \Cref{lem:correct-query-string}, we know that any query string $x$ satisfying $t(x,\psi_1,\ldots, \psi_n)>T-\eta$ (for $\eta=(\alpha-\beta)/2$, $\alpha$ and $\beta$ the completeness/soundness parameters for each $C$-verifier $Q_i$ in the $C$-DAG, and for $f=f^*$ a $\eta^{-1}$-admissible weighting function) is a correct query string. So, assume without loss of generality (since $C\in\{\NP,\MA,\QCMA,\QMA,\QMA(2)\}$, where for $\QMA(2)$ we use~\cite{Harrow2013}) that $\alpha=2/3$ and $\beta=1/3$, so that $\eta^{-1}=6$.
    By \Cref{lem:merge-admissible}, $f^*$ is $c$-admissible for any $c\geq 2$, and hence $\eta^{-1}$-admissible.
    Third, use \Cref{lem:approximating-t} in conjunction with binary search to approximate $T$ for $G^*$. Here, we must be slightly careful. Set $N=\abs{V^*}\le 2^{O(sD)}n$. Since the precision parameter $\eta\in \Theta(1)$, it suffices to use $\log (W_{f^*}(G^*)) \in O(\log\abs{V^*}) \in O(sD + \log n)$ $C$-queries to resolve $T$ within additive error $\eta$.
    Let $\widetilde{T}$ denote this estimate of $T$.
    Fourth, make a final $C$-query via \Cref{lem:approximating-t} to decide whether there exists a correct query string $x$ and proofs $\ket{\psi_1},\ldots, \ket{\psi_N}$, such that $t(x,\psi_1,\ldots, \psi_N)\geq \widetilde{T}$ and for which $x_N=1$, and return its answer. (Recall that $x_N$, by definition, encodes the output of the $C$-DAG.)
\end{proof}

\subsection{The case of $\StoqMA$}\label{sscn:stoqma}

We are only able to show a weaker version of \Cref{thm:separator-tree-time} for $\StoqMA$, due to the fact that error reduction for $\StoqMA$ is not known (i.e. one cannot assume completeness/soundness $2/3$ and $1/3$).
Specifically, \Cref{l:compress,lem:approximating-t,lem:correct-query-string} still hold for $C=\StoqMA$.
However, as amplification of $\StoqMA$'s promise gap is not known (see, e.g.~\cite{Aharonov2020}), we cannot assume $\eta = \Omega(1)$ in the proof of \Cref{thm:separator-tree-time}.
If we instead use $\eta = 1/\poly(n)$ (note the use of $n$ versus $N$ here; see \Cref{rem:promisegaps}), \Cref{lem:approximating-t,lem:correct-query-string} require a $\poly(n)$-admissible weighting function.
However, for any $c$-admissible weighting function $f$, $W_f(G) \ge c^{\depth(G)}$, which is superpolynomial when $\depth(G)=\omega(1)$ and $c=\poly(n)$.
Thus, for StoqMA we can only prove the following weaker analogue of \Cref{thm:bsn-s}:

\theoremBsnSStoqma*
\begin{proof}
    This follows analogously to \Cref{thm:bsn-s}, but with $c\in\poly(n)$ (see \Cref{rem:promisegaps}) instead of $c\in O(1)$, which incurs an additional $\log$ factor in the exponent.
\end{proof}

Akin to \Cref{cor:qp}, it follows that:

\begin{corollary}\label{cor:qp-stoqma}
    For $C=\StoqMA$, $\CDAG_{\log^k}\in \QP^{C[\log^{k+2}(n)]}$ for all constants $k\in\N$.
\end{corollary}

\subsection{Query Graphs of Bounded Depth}\label{sscn:bounded}

One can ask whether there are other kinds of graphs for which we can apply the techniques developed in this section.
Using the weighting function $\rho$ from \Cref{lem:weighting-function}, we obtain the following results for query graphs of bounded depth.

\begin{definition}[$\CDAGd$]\label{def:bsn_depth}
    Let $d:\N\to\N$ be an efficiently computable function. Then, $\CDAGd$ is defined as $\CDAG$, except that $G$ has depth scaling as $O(d(n))$, for $n$ the number of nodes used to specify the $\CDAG$ instance.
\end{definition}
\noindent We caution that this notation is very close to that of $\CDAGs$ --- the $d$ in $\CDAGd$ distinguishes that here we are considering bounded depth (as opposed to bounded separator number with $\CDAGs$). The union of $\CDAGd$ over all polynomials $d:\N\mapsto\N$ equals $\CDAG$.

The next theorem was shown by Gottlob \cite{Gottlob1995} for $C=\NP$ and $d(n)=\log^in, i\in\N$.
We strengthen it for NP and simultaneously extend it to the quantum setting.

\theoremConstantDepth*
\begin{proof}
    Follows analogously to \Cref{thm:bsn-s}, except we do not need to apply the graph transformation.
    It suffices to use \Cref{lem:approximating-t,lem:correct-query-string} directly with the weighting function $\rho$ from \Cref{lem:weighting-function}.
    For $C\in\{\NP, \NEXP, \QMAexp\}$, we do not need to increase the runtime of the base class to the total weight $W_\rho(G)$ because the queries for \Cref{lem:approximating-t} can be performed exactly or with exponential precision in the case of $\QMAexp$.
\end{proof}

\begin{corollary}\label{cor:constant-depth}
    For $C\in\QVp$ and $d\in O(1)$, $\CDAGd$ is $\p^{C[\log]}$-complete.
\end{corollary}
\begin{proof}
    Containment in $\p^{C[\log]}$ is given by \Cref{thm:constant-depth}. That $\CDAGd$ for $d\in O(1)$ is $\p^{C[\log]}$-hard follows analogously to \Cref{thm:bsn} (in particular, since $ \p^{C[\log]} \subseteq \p^{\Vert C}$ for general $C$~\cite{B91}, and since parallel queries correspond to a constant depth query graph).
\end{proof}

\section{Hardness for APX-SIM via a unified framework}\label{scn:APXSIM}

We now show how the construction of \Cref{scn:BSN} can be embedded directly into the flag-qubit Hamiltonian construction of \cite{WBG20}, thus directly yielding hardness results for the APX-SIM problem (\Cref{def:apx}).
We first give required definitions in \Cref{sscn:APXdefs}.
\Cref{sscn:apxmainresult} and \Cref{sscn:helpfullemmas} state and prove the main result of this section, the Generalized Lifting Lemma (\Cref{lem:lift}).
Finally, \Cref{sscn:apply} shows how to apply the Generalized Lifting Lemma to obtain hardness results for APX-SIM (see \Cref{def:apx}).

\subsection{Definitions}\label{sscn:APXdefs}



The following definitions were introduced in \cite{WBG20} to allow one to abstractly speak about large classes of circuit-to-Hamiltonian mappings. This allows the Lifting Lemma of \cite{WBG20}, as well as its generalized version shown in \Cref{sscn:apxmainresult} (\Cref{lem:lift}), to be used in a black-box fashion (i.e. agnostic to the particular choice circuit-to-Hamiltonian construction used). As a result, both Lifting Lemmas automatically preserve desirable properties of the actual circuit-to-Hamiltonian mappings employed, such as being 1D or translation invariant.

\begin{definition}[Conformity~\cite{WBG20}]\label{def:conforming}
    Let $H$ be a Hamiltonian with some well-defined structure $S$ (such as $k$-local interactions, all constraints drawn from a fixed finite family, with a fixed geometry such as 1D, translational invariance, etc). We say a Hermitian operator $P$ \emph{conforms} to $H$ if $H+P$ also has structure $S$.
\end{definition}

\noindent For example, if $H$ is a $2$-local Hamiltonian on a 2D square lattice, then $P$ conforms to $H$ if $H+P$ is also a $2$-local Hamiltonian on a 2D square lattice. Next, define $\unitary{\X}$ as the set of unitary operators acting on space $\X$.

\begin{definition}[Local Circuit-to-Hamiltonian Mapping~\cite{WBG20}]\label{def:local-mapping} Let $\spa{X}=(\C^2)^{\otimes p}$ and $\spa{Y}=(\C^2)^{\otimes q}$. A map $\Hw:\unitary{\spa{X}}\mapsto\herm(\spa{Y})$ is a \emph{local circuit-to-Hamiltonian mapping} if, for any $L>0$ and any sequence of $2$-qubit unitary gates $U=U_LU_{L-1}\cdots U_1$, the following hold:
\begin{enumerate}
    \item (Overall structure) $\Hw(U)\succeq0$ has a non-trivial null space, i.e.\ $\Null(\Hw(U))\neq 0$.
    This null space is spanned by (some appropriate notion of) ``correctly initialized computation history states'', i.e.\ with ancillae qubits set ``correctly'' and gates in $U$ ``applied'' sequentially.
    \item (Local penalization and measurement) Let $q_1$ and $q_2$ be the first two output wires of $U$ (each a single qubit), respectively. Let $\spre\subseteq \spa{X}$ and $\spost\subseteq\spa{Y}$ denote the sets of input states to $U$ satisfying the structure enforced by $\Hw(U)$ (e.g.\ ancillae initialized to zeroes), and null states of $\Hw(U)$, respectively. Then, there exist projectors $M_1$ and $P_L$, projector $M_2$ conforming to $\Hw(U)$, and a bijection $\f:\spre\mapsto\spost$, such that for all $i\in\set{1,2}$ and $\ket{\phi}\in\spre$, the state $\ket{\psi}=\f(\ket{\phi})$ satisfies
        \begin{equation}\label{eqn:sim}
            \Tr\big(\ketbra{0}{0}_i(U_LU_{L-1}\ldots U_1)\ketbra{\phi}{\phi}(U_LU_{L-1}\ldots U_1)^\dagger\big) = \Tr\big(\ketbra{\psi_L}{\psi_L} M_i\big),
        \end{equation}
        where $\ket{\psi_L} = P_L \ket\psi / \enorm{P_L \ket\psi}$ is $\ket\psi$ postselected on measurement outcome $P_L$ (we require $P_L \ket\psi\neq 0$). Moreover, there exists a function $\g:\N\times\N\mapsto\R$ such that
        \begin{align}
            \enorm{P_L\ket{\psi}}^2&=\g(p,L)\text{ for all }\ket{\psi}\in \Null(\Hw(U)),\label{eqn:sim3}\\
            M_i&= P_L M_i P_L\label{eqn:sim2}.
        \end{align}
\end{enumerate}
The map $\Hw$, and all operators/functions above ($M_1$,$M_2$,$P_L$,$\f$,$\g$) are computable given $U$.
\end{definition}
\noindent To gain intuition about \Cref{def:local-mapping}, consider the simplest case of Kitaev's $5$-local construction applied to a QMA verification circuit $U=U_L\cdots U_1$~\cite{Kitaev2002}. Then\footnote{In~\cite{Kitaev2002}, $\hin$ checks that all ancillae are set to $\ket{0}$ before the verification, $\hprop$ checks that each step $i$ in the verification follows from step $i-1$, $\hstab$ ensures the clock register is correctly encoded, and $\hout$ checks that the verifier accepts the given proof.}, $\Hw(U)=\hin+\hprop+\hstab$, since recall the null space of $\Hw(U)$ is precisely the set of all correctly initialized history states. (Notably, the term $\hout$ is omitted.) The sets $\spre$ and $\spost$ correspond to the correctly initialized inputs to $U$ (i.e. of form $\ket{\psi}_A\ket{0\cdots 0}_B$ for some proof $\ket{\psi}_A$ and the ancilla register $B$ set to all zeroes) and all correctly initialized history states\footnote{Kitaev's history state~\cite{Kitaev2002} encodes the history of the verification in \emph{superposition}, i.e. as $\ket{\psihist}\propto \sum_{t=0}^{L}U_t\cdots U_1\ket{\psi}_A\ket{0\cdots 0}_B\ket{t}_C$, where $C$ is a clock register. This in contrast to the Cook-Levin theorem~\cite{Cook1971,Levin1973}, which encodes the history in a tableau.} $\ket{\psihist}$, respectively. The projector $P_L=\ketbra{L}{L}_C$ projects onto timestep $L$ in clock register $C$, with $g(p,L)=1/(L+1)$. Finally, $M_i=\ketbra{0}{0}_{A_i}$ (for $A_i$ the $i$th qubit of register $A$, and where the projection onto time step $L$ has already happened due to the use of $\ket{\psi_L}$ in \Cref{eqn:sim}).

%
%

\subsection{The Generalized Lifting Lemma}\label{sscn:apxmainresult}

\begin{figure}[t]
\[
\Qcircuit @C=1.5em @R=0.5em {
     \lstick{\ket{x}} & \multigate{5}{~~V~~} &   \rstick{x_N}\qw\\
     \lstick{\ket{\psi_1}} & \ghost{~~V~~} &   \rstick{\qflag}\qw\\
     \lstick{\ket{\psi_2}} & \ghost{~~V~~} &  \\
     \push{\vdots\rule{0em}{1em}\hspace{2mm}}&    &\\
     \lstick{\ket{\psi_N}} & \ghost{~~V~~} &  \\
     &  &
}
\]
\caption{A depiction of the circuit $V$ constructed in \Cref{lem:approximating-t}, with two minor modifications for our purposes here. First, the second wire above denotes the output wire of $V$, and is relabbeled $\qflag$ here. Second, we assume without loss of generality that $V$ outputs the $N$th bit of $x\in\set{0,1}^N$ on the first wire above, labelled $x_N$. For simplicity, we depict the proofs $\ket{\psi_i}$ above in tensor product, but we make no such \emph{a priori} assumption in any of our proofs.}
\label{fig:verifier}
\end{figure}
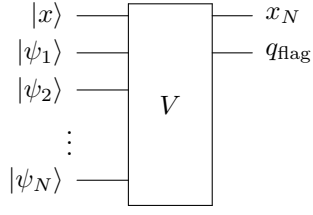

\begin{lemma}[Generalized Lifting Lemma for APX-SIM]\label{lem:lift}
Fix $C\in\QVp$. As input, we are given a \CDAG{} instance $G^*$ on $N$ nodes, and $c$-admissible weighting function $f^*$. Let $V$, as depicted in \Cref{fig:verifier}, be the verification circuit constructed in \Cref{lem:approximating-t}, given $(G^*,f^*)$.
Define shorthand $\Delta$ for $\Delta(\Hw(V))$.
Fix a local circuit-to-Hamiltonian mapping $\Hw$, and assume the notation in \Cref{def:local-mapping}. Fix any function $\alpha:\N\mapsto\N$ such that
 \begin{equation}
    \alpha > \max\left(\frac{4\norm{M_2}}{\Delta},\frac{\Delta}{3\norm{M_2}^2},1\right).
 \end{equation}
Then, the Hamiltonian $H\coloneqq \alpha\Hw(V)+ M_2$ satisfies:
    \begin{itemize}
		\item If $G$ is a YES instance, then for all $\ket{\psi}$ with $\bra{\psi}H\ket{\psi}\leq\lmin(H)+\frac{1}{\alpha^2}$,
		\begin{equation}
		  \bra{\psi}M_1\ket{\psi}\leq \frac{1}{\alpha}\left[\frac{W}{\eta}\left(\frac{1}{\alpha}+\frac{12\norm{M_2}^2}{\Delta}\right)+\frac{12\norm{M_2}^2}{\Delta}\right].
		\end{equation}
		\item If $G$ is a NO instance, then for all $\ket{\psi}$ with $\bra{\psi}H\ket{\psi}\leq\lmin(H)+\frac{1}{\alpha^2}$,
		\begin{equation}
		  \bra{\psi}M_1\ket{\psi}\geq g(p,L)-\frac{1}{\alpha}\left[\frac{W}{\eta}\left(\frac{1}{\alpha}+\frac{12\norm{M_2}^2}{\Delta}\right)-\frac{12\norm{M_2}^2}{\Delta}\right],
		\end{equation}
    \end{itemize}
    for $W$ and $\eta$ defined in \Cref{lem:approximating-t} and \Cref{lem:correct-query-string}, respectively, and $\g(p,L)$ defined in \Cref{def:local-mapping}.
\end{lemma}
\begin{proof}
    The claim follows immediately by defining $\delta:=1/\alpha^2$, and then combining \Cref{l:Trace_Distance}, \Cref{l:low-flag}, and \Cref{l:finalhardness} (all given subsequently in \Cref{sscn:helpfullemmas}). Roughly, \Cref{l:Trace_Distance} first shows that any low-energy state of $H$ must be ``close'' to a history state (formally, a null state of $\Hw(V)$, as per \Cref{def:local-mapping}). \Cref{l:low-flag} shows that, in turn, any low-energy \emph{history} state of $H$ must have most of its weight on correct query strings. Finally, \Cref{l:finalhardness} combines the previous two lemmas, along with \Cref{def:local-mapping}, to obtain the claim, i.e. the ground state of $H$ must encode the full computation represented by $G$, and thus a local measurement suffices to decide $G$.
\end{proof}

\subsection{Lemmas required for proof of Lifting Lemma}\label{sscn:helpfullemmas}

We now give the three lemmas required for the proof of \Cref{lem:lift}, all of which assume the notation for the latter. The first of these can be stated and proven identically to Lemma 22 of  \cite{WBG20}, since it does not leverage any properties of $V$ itself, but only the abstract definition of $\Hw(U)$. While the proof is simple, it uses the Extended Projection Lemma~\cite{Kempe2006,Gharibian2019}; for brevity we omit both here.

\begin{lemma}[\cite{WBG20}]\label{l:Trace_Distance}
	 Fix any function $\alpha:\N\mapsto\N$ such that
	\begin{equation}\label{eqn:6}
	\alpha > \max\left(\frac{4\norm{M_2}}{\Delta},\frac{\Delta}{3\norm{M_2}^2},1\right),
	\end{equation}
	and any $\delta\leq 1/\alpha^2$. Then, for any $\ket{\psi}$ such that $\bra{\psi}H\ket{\psi}\leq \lmin(H)+\delta$, there exists a uniform history state $\ket\phi \in \Null(\Hw(V))$ such that
	\begin{equation}\label{eqn:trbound2}
	   \trnorm{\ketbra{\psi}{\psi} - \ketbra{\phi}{\phi}} \leq\frac{12\norm{M_2}}{\alpha\Delta}
	\end{equation}
	and where $\ket{\phi}$ has energy
	\begin{equation}\label{eqn:5}
	\bra{\phi}H\ket{\phi}\leq \lmin(H)+\delta+\frac{12\norm{M_2}^2}{\alpha\Delta}.
	\end{equation}
\end{lemma}

For the second lemma, \Cref{l:low-flag}, recall $V$ has proof space $\X\otimes \Y_1\otimes\dotsm\otimes \Y_\nn$ with $\X=\B^{\otimes \nn}$ and $\Y_i=\B^{\otimes m}$. Henceforth, we denote an arbitrary (potentially entangled) proof in this space as $\ket{w_{\X\Y}}$. We remark \Cref{l:low-flag} is our version of Lemma 23 of \cite{WBG20}; however, our proof is significantly simplified, despite our lifting lemma allowing arbitrary C-DAGs, due to the specific design of our verifier $V$ from \Cref{lem:approximating-t}. (In particular, Lemma 23 of \cite{WBG20} requires a somewhat involved argument using conditional probabilities to obtain soundness against entanglement across proofs.)

\begin{lemma}\label{l:low-flag}
	Suppose history state $\ket{\phi} \in \Null(\Hw(V))$ has preimage $\ket{\psiin}=\f^{-1}(\ket{\phi})$ (for bijection $\f$ from \Cref{def:local-mapping}), where $\ket{\psiin}$ has proof $\ket{w_{\X\Y}}$ with total amplitude $\pbad$ on incorrect query strings in $\X$. Then,
	\begin{equation}\label{eqn:eta}
	   \bra{\phi}H\ket{\phi}> \lmin(H)+g(p,L)\frac{\pbad\cdot \eta}{W}.
	\end{equation}	
\end{lemma}
\begin{proof}
	Let $\ket{\psiout}=V\ket{\psiin}$.
Letting $X_+$ and $X_-$ denote the sets of correct and incorrect query strings, respectively, we may write
\begin{equation}\label{eqn:wxy}
    \ket{w_{\X\Y}}=\sum_{x\in X_-}\alpha_x\ket{x}_{\X}\ket{\psi_x}_{\Y}+\sum_{x\in X_+}\alpha_x\ket{x}_{\X}\ket{\psi_x}_{\Y},
\end{equation}
for $\sum_{x\in X_+\cup X_-}\abs{\alpha_x}^2=1$, arbitrary unit vectors $\set{\ket{\psi_x}}_x$, and $\pbad:=\sum_{x\in X_-}\abs{\alpha_x}^2$.
Recall from \Cref{def:local-mapping} that $M_2$ simulates the projector $\ketbra{0}{0}_{\qflag}$ via
	\begin{equation}\label{eqn:here}
	\Tr\left(\ketbra{0}{0}_2(U_LU_{T-1}\ldots U_1)\ketbra{\psiin}{\psiin}(U_LU_{T-1}\ldots U_1)^\dagger\right) = \Tr\big(\ketbra{\phi_L}{\phi_L} M_2\big),
	\end{equation}
	(since we assumed in \Cref{fig:verifier} that the second output qubit of $V$ is the flag qubit), where $\ket{\phi_L}$ is the history state $\ket{\phi}$ projected down onto time step $T$. We thus have
	\begin{eqnarray}
	\bra{\phi}H\ket{\phi}&=&\bra{\phi}M_2\ket{\phi}\\
    &=& g(p,L)\bra{\phi_L}M_2\ket{\phi_L}\\
    &=& g(p,L)\Tr\left(\ketbra{\psiout}{\psiout}\cdot\ketbra{0}{0}_{\qflag}\right)\\
    &=&g(p,L)\Pr[V \text{ rejects}\mid\ket{w_{\X\Y}}],
\end{eqnarray}
where the second statement follows from \Cref{eqn:sim3} and \Cref{eqn:sim2}, and the third from \Cref{eqn:here}.
By \Cref{eqn:T} and \Cref{eqn:t'}, there exists a proof $\ket{w'_{\X\Y}}=\ket{x}\ket{\psi_1}\cdots\ket{\psi_N}$ accepted by $V$ with probability precisely $T/W$.
Let $\ket{\psiin'}$ be an input state containing this optimal proof $\ket{w'_{\X\Y}}$.
\Cref{lem:correct-query-string} now yields\footnote{We are implicitly using the fact that, as observed in the proof of \Cref{lem:approximating-t}, for any fixed query string $x$, the acceptance probability of $V$ is maximized by choosing a product state proof $\ket{\psi_1}\cdots\ket{\psi_N}$ on $\Y$.}
\begin{eqnarray}
    \bra{\phi}H\ket{\phi}&>&g(p,L)\left(\Pr[V \text{ rejects}\mid\ket{w'_{\X\Y}}]+\left(\sum_{x\in X_-}\abs{\alpha_x}^2\right)\frac{\eta}{W}\right)\label{eq:low-flag:1}\\
    &=&\bra{\phi'}M_2\ket{\phi'}+g(p,L)\frac{\pbad\cdot \eta}{W}\label{eq:low-flag:2}\\
    &\geq&\lmin(H)+g(p,L)\frac{\pbad\cdot \eta}{W},\label{eq:low-flag:3}
\end{eqnarray}
where the first inequality \eqref{eq:low-flag:1} uses the fact that
\begin{equation}
    \Pr[V \text{ accepts} \mid \ket{w_{\X\Y}}] \le \pgood\cdot\frac TW + \pbad \left(\frac TW - \frac\eta W\right)= \frac TW - \frac{\pbad\cdot\eta}W.
\end{equation}
The second statement uses \Cref{eqn:sim3,eqn:sim2}, with $\ket{\phi'}:=\f(\ket{\psiin'})$,
and the last statement \eqref{eq:low-flag:3} uses $\ket{\phi'}\in\Null(\Hw(V))$ by the definition of $\f$ in \Cref{def:local-mapping}.

As a final aside, the proof above is written with the context of \emph{quantum} verification classes such as $C=\QMA$ in mind. However, the same proof can be applied directly to (say) $C=\NP$ by embedding an NP verifier in the usual manner into a QMA verifier (i.e. the QMA verifier begins by measuring its proof in the standard basis via the principle of deferred measurement). Of course, even when $C=\NP$, the construction of this section still yields a genuinely quantum Hamiltonian $H$ (as opposed to a Hamiltonian $H$ diagonal in the standard basis), due to our use of circuit-to-Hamiltonian mappings $\Hw$.
\end{proof}

Finally, the third lemma, \Cref{l:finalhardness}, is our analog of Lemma 25 of \cite{WBG20}. We follow the same high-level approach as the latter, but again, our proof here is simplified. This is because \Cref{lem:correct-query-string} can be directly leveraged to obtain that any history state close enough to the ground space of $H$ must simply output the correct answer to the input \CDAG{} on wire $x_N$ in \Cref{fig:verifier}. (In contrast,~\cite{WBG20} needed the Commutative Quantum Union Bound to argue that all proofs are simultaneously correct.)

\begin{lemma} \label{l:finalhardness}
	Consider any $\ket{\psi}$ satisfying $\bra{\psi}H\ket{\psi}\leq\lmin(H)+\delta$. If $\delta\leq 1/\alpha^2$, then
	\begin{itemize}
		\item if $G^*$ is a YES instance, then
		\begin{equation}
		  \bra{\psi}M_1\ket{\psi}\leq \frac{W}{\eta}\left(\delta+\frac{12\norm{M_2}^2}{\alpha\Delta}\right)+\frac{12\norm{M_2}^2}{\alpha\Delta}
		\end{equation}
		\item if $G^*$ is a NO instance, then
		\begin{equation}
		  \bra{\psi}M_1\ket{\psi}\geq g(p,L)-\frac{W}{\eta}\left(\delta+\frac{12\norm{M_2}^2}{\alpha\Delta}\right)-\frac{12\norm{M_2}^2}{\alpha\Delta}
		\end{equation}
	\end{itemize}
\end{lemma}
\begin{proof}
	We first use \Cref{l:Trace_Distance} to map, assuming $\delta\leq 1/\alpha^2$, $\ket{\psi}$ to a history state $\ket{\phi}\in \Null(\Hw(V))$ such that
	\begin{equation}\label{eqn:step1}
	\trnorm{\ketbra{\psi}{\psi} - \ketbra{\phi}{\phi}} \leq\frac{12\norm{M_2}}{\alpha\Delta}\quad\text{and}\quad
	\bra{\phi}H\ket{\phi}\leq \lmin(H)+\delta+\frac{12\norm{M_2}^2}{\alpha\Delta}.
	\end{equation}
	We next use \Cref{l:low-flag} to obtain that preimage $\ket{\phiin}=f^{-1}(\ket{\phi})$ contains proof $\ket{\wm}$ (see \Cref{eqn:wxy}) with
\begin{equation}
    \pbad < \frac{W}{g(p,L)\cdot\eta}\left(\delta+\frac{12\norm{M_2}^2}{\alpha\Delta}\right),\label{eqn:pbad}
\end{equation}	
	i.e. the total amplitude $\pbad$ of $\ket{w_{\X\Y}}$ on incorrect query strings in $\X$ is bounded. But by \Cref{def:c-dag}, if the query string $x_1\cdots x_N$ in $\X$ is correct, then $x_N$ encodes the correct output of C-DAG $G^*$.
Moreover, by design, $V$ in \Cref{fig:verifier} always outputs $x_N$ on its first wire.
Thus,
\begin{eqnarray*}
    \text{if $G^*$ is a YES instance}\quad&\Rightarrow&\quad\Tr\big(\ketbra{0}{0}_1V\ketbra{\phiin}{\phiin}V^\dagger\big)\leq \pbad,\label{eqn:yes}\\
    \text{if $G^*$ is a NO instance}\quad&\Rightarrow&\quad\Tr\big(\ketbra{0}{0}_1V\ketbra{\phiin}{\phiin}V^\dagger\big)\geq 1-\pbad\label{eqn:no}
\end{eqnarray*}
Since for $\ket{\phi_L} = P_L \ket\phi / \enorm{P_L \ket\phi}$,
\begin{equation}
    \bra{\phi}M_1\ket{\phi} = g(p,L)\bra{\phi_L}M_1\ket{\phi_L}=g(p,L)\Tr\big(\ketbra{0}{0}_1V\ketbra{\phiin}{\phiin}V^\dagger\big)
\end{equation}
(by \Cref{eqn:sim2}, \Cref{eqn:sim3}, and \Cref{eqn:sim}), we thus have that if $G^*$ is a YES instance, $\bra{\phi}M_1\ket{\phi}\leq g(p,L)\pbad$, and if $G^*$ is a NO instance, $\bra{\phi}M_1\ket{\phi}\geq g(p,L)(1-\pbad)$. Combining this with \Cref{eqn:pbad} and \Cref{eqn:step1} via H\"{o}lder's inequality yields the claim.
\end{proof}

\subsection{Applying the Lifting Lemma}\label{sscn:apply}

We now give two examples of how to use \Cref{lem:lift} to obtain hardness results for APX-SIM, for the cases of $C=\QMA$ and $C=\StoqMA$.

\paragraph{Example 1: $C=\QMA$.} The theorem below sets $N:=\min(2^{O(s(n)\log n)},2^{O(d(n)\log n)})$ --- the two values in $\min(\cdot,\cdot)$ correspond to the use of the bounded separator framework (\Cref{thm:bsn-s}) or bounded depth framework (\Cref{thm:constant-depth}), respectively, in conjunction with \Cref{lem:lift}.

\begin{theorem}[Hardness of APX-SIM for $C=\QMA$ via \Cref{lem:lift}] \label{thm:liftedhardness}
    Fix $C=\QMA$, and let $G$ be any $\CDAG$ instance on $n$ nodes with separator number and depth scaling as $s(n)$ and $d(n)$, in the sense of $\CDAGs$ and $\CDAGd$, respectively. Set $N:=\min(2^{O(s(n)\log n)},2^{O(d(n)\log n)})$.  Then, there exists a $\poly(N)$-time many-one reduction from $G$ to an instance $(H,a,b,\delta)$ of APX-SIM, which satisfies: (1) $H$ has size $\poly(N)$ (i.e. acts on $\poly(N)$ qubits/qudits, and has $\poly(N)$ local terms), (2) $H$ is either $5$-local acting on qubits or $2$-local on a 1D chain of $8$-dimensional qudits (depending on which circuit-to-Hamiltonian mapping is employed), (3) $b-a\geq 1/\poly(N)$ and $\delta\geq 1/\poly(N)$.
\end{theorem}
\begin{proof}
    If $N=2^{O(d(n)\log n)}$, set $(G^*,f^*)=(G,\rho)$ for $\rho$ from \Cref{lem:weighting-function} and proceed to the next paragraph. Otherewise, as in \Cref{thm:bsn-s}, apply \Cref{lem:compute-separator-tree} to $G$ to compute a separator tree of depth $D = O(\log(n))$ with separators of size $s=\s(G)$ in time $n^{O(s)}$. This is then fed into \Cref{l:compress} to obtain an equivalent C-DAG $G^*$ with $N=2^{O(s(n)\log n)}$ nodes, each of which corresponds to a QMA verifier of size $\poly(N)$ (i.e. with constant promise gap, taking in proof of size $\poly(N)$, and running a verification circuit of size $\poly(N)$), along with weighting function $f^*$.

    Next, invoke \Cref{lem:lift} on $(G^*,f^*)$.
    Depending on whether we desire $H$ to be $5$-local on qubits or a 1D chain on qudits, set $\Hw$ to be Kitaev's $5$-local construction~\cite{Kitaev2002} or Hallgren, Nagaj, and Narayanaswami's 1D construction~\cite{HNN13}, respectively (except in both cases, we omit the $\hout$ term which penalizes rejected proofs).
    We then plug $\Hw(V)$ for $V$ from \Cref{fig:verifier} into \Cref{lem:lift} with parameters as follows. (Note that by \Cref{l:compress} and \Cref{lem:approximating-t}, $V$ has size $\poly(N)$, and thus $\Hw(V)$ has size $\poly(N)$, by the constructions of \cite{Kitaev2002,HNN13}.)

    First, $M_1$ and $M_2$ are appropriately encoded $1$-local rank-$1$ projectors onto $\ketbra{0}{0}$ at the last verification time step on the output and flag qubits, respectively; thus, $\norm{M_2}=1$. The spectral gap $\Delta(\Hw(V))$ scales as $1/\poly(N)$~\cite{GK12,Gharibian2019b}, and $g(p,L)=1/(1+L)=1/\poly(N)$ in both cases. If $N=2^{O(d(n)\log n)}$, then the weighting function $W=W_{f^*}$ satisfies $W_{f^*}(G^*)\le n(cn^{d(n)})\in\poly(N)$ for any $c\in \poly(n)$. Else, by \Cref{l:compress}, the weighting function $W=W_{f^*}$ satisfies $W_{f^*}(G^*)\le (c+1)^{O(sD)}n\in \poly(N)$, and $\eta\in O(1)$ (defined in \Cref{lem:correct-query-string}, and since in time $\poly(N)$, each QMA verifier at a node of $G^*$ can be amplified to have constant promise gap). In both cases, we conclude that by setting $\alpha$ to be a large enough fixed polynomial in $N$, we obtain a $1/\poly(N)$ promise gap in \cref{lem:lift}, thus satisfying all claims regarding $a,b,\delta$. All functions involved (e.g. $g(m,T)$, $\Delta$), including the reduction itself, run in time $\poly(N)$.
\end{proof}

\noindent As noted in \Cref{sscn:results}, combining \Cref{thm:bsn-s} with \Cref{thm:liftedhardness}, we have that $\CDAG_1$ can directly be embedded into an instance of APX-SIM.

\paragraph{Example 2: $C=\StoqMA$.}  In \Cref{lem:lift}, when $N=2^{O(s(n)\log n)}$ (i.e. bounded separator number framework) the promise gap of $C$ directly influences $\eta$, which in turn affects $W$, $\alpha$, and $\Delta(\Hw(V))$. Thus, we can apply it to obtain hardness for APX-SIM on stoquastic Hamiltonians. The tradeoff is that due to the extra log factor in \Cref{thm:bsn-s-stoqma} (versus \Cref{thm:bsn}), the size of the stoquastic APX-SIM instance obtained still unfortunately grows quasi-polynomially, even for $s\in O(1)$. (Recall this extra log factor is itself due to the lack of error reduction!) However, when $N=2^{O(d(n)\log n)}$ (i.e. bounded depth framework), no such hit is incurred. As in \Cref{lem:lift}, both frameworks are considered below.

\begin{theorem}[Hardness of APX-SIM for $C=\StoqMA$ via \Cref{lem:lift}] \label{thm:liftedhardnessstoqma}
    Fix $C=\StoqMA$ and any efficiently computable function $s:\N\to\N$, and define $N:=\min(2^{O(s(n)\log^2 n)},2^{O(d(n)\log n)})$.  Then, there exists a $\poly(N)$-time many-one reduction from any instance of $\CDAG$ to an instance $(H,a,b,\delta)$ of APX-SIM for stoquastic $H$, which satisfies: (1) $H$ has size $\poly(N)$ (i.e. acts on $\poly(N)$ qubits, and has $\poly(N)$ local terms), (2) $H$ is $2$-local, (3) $b-a\geq 1/\poly(N)$ and $\delta\geq 1/\poly(N)$.
\end{theorem}
\begin{proof}
    The proof is almost identical to that of \Cref{thm:liftedhardness}, except for two differences: (1) Set $\Hw$ as the stoquastic circuit-to-Hamiltonian construction of Bravyi, Bessen, and Terhal \Cite{Bravyi_Bessen_Terhal_2006}, so that the output Hamiltonian $H$ is indeed stoquastic. (Recall by that $V$ in \Cref{fig:verifier} is indeed stoquastic by \Cref{rem:stoqma}.) (2) When $C=\StoqMA$ and $N=2^{O(s(n)\log^2 n)}$, $\eta = 1/\poly(n)$ (versus $\eta= O(1)$), and so $c\in \poly(n)$. This means that although the size of $G^*$ produced by \Cref{l:compress} remains unchanged, the weighting function $f^*$ now satisfies $W_f^{*}\leq 2^{O(s\log^2n)}$ (versus $W\leq 2^{O(s\log n)}$). This, in turn, means that $V$ (\Cref{lem:approximating-t} and \Cref{fig:verifier}) grows polynomially in size as $2^{O(s\log^2n)}$, implying $\Delta(\Hw(V))$ (see Lemma 5 of \cite{Bravyi_Bessen_Terhal_2006}), and thus $\alpha$, also scale with $2^{O(s\log^2n)}$. (The analysis for the $N=2^{O(d(n)\log n)}$ case remains unchanged.)
\end{proof}
\noindent Thus, in the $N=2^{O(d(n)\log n)}$ case (i.e. bounded depth framework), we recover that APX-SIM on stoquastic Hamiltonians is $\pStoqMAlog$-hard~\cite{Gharibian2019b}. For clarity, this follows because $\pStoqMAlog=\p^{\Vert \StoqMA}$~\cite{Gharibian2019b}, and $\p^{\Vert \StoqMA}$ corresponds to a depth-$1$ $\StoqMA$-DAG.

\section{No-go statement for ``weak compression'' of polynomials}\label{scn:polycompress}

We now make a simple observation that the weighting function approach applied to NP queries (introduced in~\cite{Gottlob1995} and used here as well) can be turned upside-down to obtain a no-go statement about a purely mathematical question: \emph{Can arbitrary multi-linear polynomials be ``weakly compressed''}? Throughout this section, we consider weighting functions applied to NP-DAGs.

\paragraph{Definitions.} To define ``weak compression'', recall first the definition of an \emph{arithmetic circuit}, which is a standard succinct encoding for polynomials.

\begin{definition}[Arithmetic circuit]\label{def:arithmetic}
    An \emph{arithmetic circuit} $C$ over field $F$ is given via a DAG as follows. Each vertex of in-degree $0$ is labelled by either a variable $x_i$ or a constant from $F$. Each vertex of in-degree at least $2$ is labelled by either the ``$+$'' or ``$\times$'' operation. Vertices of in-degree $1$ are not allowed. There is a single node of out-degree $0$, the \emph{output node}. The polynomial $p_C$ computed by $C$ is obtained by evaluating the circuit with order of operations dictated by any topological order on $C$, where the output node is fixed as the last node in the order.
\end{definition}

\begin{figure}[t]
\begin{center}
\includegraphics[width=3.5cm]{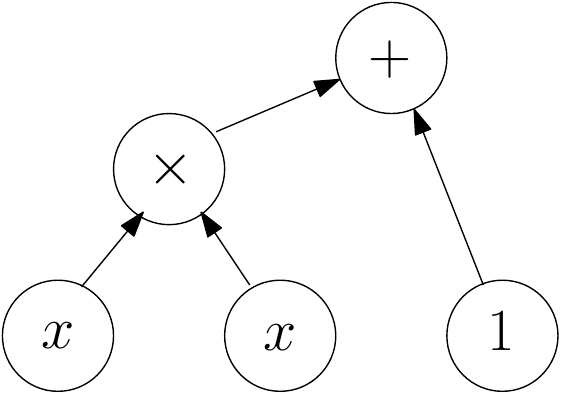}
\end{center}
\caption{An example of an arithmetic circuit which computes polynomial $p(x)=x^2+1$.}
\label{fig:arithmetic}
\end{figure}


We now define our notion of weak compression; intuition given subsequently.

\begin{definition}[Weak compression of polynomials]\label{def:compress}
    Let $f:[0,1]^m\to\R^+$ be a multi-variate polynomial with rational coefficients, specified via an arithmetic circuit of size $M$. Assume there exists $x^*\in[0,1]^m$ maximizing $f$ such that $f(x)$ can be specified exactly\footnote{For clarity, we are assuming a naive binary expansion of $f(x^*)$.} via $B$ bits, for some (finite) $B$. We say $f$ is \emph{weakly compressible to $B'$ bits} if there exists an efficiently computable mapping taking $f$ to another function $g:[0,1]^{m'}\to \R^+$ such that:
\begin{enumerate}
    \item For any $y\in[0,1]^{m'}$, $g(y)$ is computable in $\poly(m)$ time.
    \item (Optimality preserved) For any optimal $y^*$ maximizing $g(y^*)$ over $[0,1]^{m'}$, there exists a $\poly(m)$-time map taking $y^*$ to an optimal $x^*\in[0,1]^m$ maximizing $f(x^*)$.
    \item (Compression) There exists an optimal $y^*$ requiring at most $B'$ bits to specify exactly.
\end{enumerate}
\end{definition}
\noindent Very roughly, \Cref{def:compress} says we may efficiently reduce the number of bits required to represent the optimal value $f(x^*)$. More formally, we can efficiently map polynomial $f$ to a new function $g$ such that: (1) $g$ may deviate from $f$ arbitrarily, except on at least one optimal point $x^*$ for $f$, which $g$ must ``preserve'' via some $g$-optimal point $y^*$. (2) $f(y^*)$ must require fewer (i.e. $B'$) bits than $f(x^*)$ (i.e. $B)$ to represent. Note that $g$ is not required to be a polynomial, nor do we require that $g(y^*)=f(x^*)$.\\

\noindent \emph{Sanity checks regarding \Cref{def:compress}.} When $B'\geq B$, $f$ is trivially weakly compressible to $B'$ bits (simply set $g$ to $f$). More interestly, one might ask: Given $f$, why can one not simply divide $f$ by $f(x^*)$, i.e. set $f'(x)=f(x)/f(x^*)$? This would allow $B'=1$. The problem here is that $x^*$ is \emph{not known a priori}, and crucially, $f$ is specified via an arithmetic circuit. Thus, it is not at all clear how one might efficiently compute $f(x^*)$, given just this circuit description.\\

We now observe a no-go statement regarding weak compressibility of polynomials (expressed as arithmetic circuits).

\lemmaPolycompress*
\noindent The proof, while simple, requires a few ingredients, and is thus given in \Cref{sscn:proofpolycompress}. It leads to the following concrete no-go statements.

\corollaryPolycompressA*
\begin{proof}
    Immediate from \Cref{l:polycompress} and the fact that in its proof, the admissible weighting function $\omega$ can have at most exponential total weight on an arbitrary NP-DAG, which requires $B$ to scale as a polynomial in the worst case.
\end{proof}

\corollaryPolycompressB*
\begin{proof}
    The proof is similar to \Cref{cor:polycompress1}, except when we start with a $\pNPtwo$ computation (i.e. making $2$ NP queries). The weighting function $\omega$ now has at most $O(1)$ total weight, justifying the choice $B\in O(1)$ in the claim. The claim now follows since if $\pNPtwo=\pNPone$, then $\PH=\textup{P}^{\Sigma_2^p}$~\cite{Hartmanis1993}.
\end{proof}

\subsection{Proof of \Cref{l:polycompress}}\label{sscn:proofpolycompress}

We first require the following lemma for encoding correct NP query strings into polynomial optimization. For concreteness, consider the admissible weighting setup of \Cref{sscn:solve-c-dag}, specialized to the case of NP queries. Recall from \Cref{sec:approximating-t} that the admissible weighting function framework allows us to reduce the task of identifying a correct \NP\ query string $x^*\in\set{0,1}^m$ to optimizing a real-valued function $t:\set{0,1}^{\poly(m)}\to\R$ of form (c.f. \Cref{eqn:t}, which also allowed QMA queries)
\begin{equation}\label{eqn:f}
    t(x,y_1,\ldots, y_m)=\sum_{i=1}^mw_i \left(x_i V_i(x,y_i) + \frac{(1-x_i)}{2}\right),
\end{equation}
where $w_i$ are the admissible weights (assumed to be rational), the bit $x_i$ encodes the claimed answer to \NP\ verifier $V_i$ in the NP-DAG, and $y_i$ is the \NP\ proof to verifier $V_i$. (Remark: In \Cref{eqn:t}, $V_i$ takes in $z_i(x)$ rather than all of $x$, where recall $z_i(x)$ selects the substring of $x$ corresponding to the input wires of node $V_i$. For simplicity, here we assume without loss of generality that the function $z_i$ is embedded into the definition of $V_i$ itself, so that we can omit writing $z_i$.) Then, by \Cref{lem:correct-query-string}, $x^*$ is simply read off the optimal $(x^*,y^*_1,\ldots, y^*_m)$ which maximizes $t$.

We now use standard tricks to encode this setup into optimization of a multi-linear polynomial.
\begin{lemma}\label{l:poly}
    Let $t$ be as in \Cref{eqn:f}, specified using $n$ bits of precision (used to describe weights $w_i$ and verifiers $V_i$). There exists a polynomial-time Turing machine which, given $t$, produces an arithmetic circuit encoding multi-linear polynomial $\pout:[0,1]^{\poly(n)}\to\R^+$ with rational coefficients such that
\begin{equation}
    \max_{x,y_1,\ldots, y_m\in\set{0,1}^{\poly(m)}} t(x,y_1,\ldots, y_m)=\max_{s\in[0,1]^{\poly(m)}}\pout(s).
\end{equation}
(Both $f$ and $\pout$ have range $[0,\sum_i\abs{w_i}]$ over their respective domains.) Moreover, given an optimal $s^*$ maximizing $\pout$, one can efficiently compute a correct NP query string for the NP-DAG underlying $t$.
\end{lemma}
\begin{proof}
The construction applies standard tricks (used, e.g., in the proof of $\IP=\PSPACE$~\cite{S92}). Fix a topological ordering $R:=(V_1,\ldots, V_m)$ on the vertices of the NP-DAG, and let $L$ denote the maximum level (\Cref{def:level}) of any node in $R$. Throughout, we abuse notation and interchangably refer to $V_i$ as both nodes in the DAG and \NP\ verifiers $V_i$. The construction of $\pout$ is accomplished by the following iterative algorithm:
\begin{enumerate}
    \item Set $i=0$.
    \item While $i\leq L$ do:
    \begin{enumerate}
        \item Let $S_i$ denote the set of nodes at level $i$ (with respect to $R$).
        \item For all $V\in S_i$:
        \begin{enumerate}
                \item (Map circuits to $3$-SAT formulae) Map $V$ to a $3$-SAT formula $\phi_V$ via the Cook-Levin theorem~\cite{C72,L73} with a minor modification: Since the input to $V$ is \emph{a priori} unknown (it depends on the outputs of the predecessors of $V$), omit the constraints in the Cook-Levin construction which force the input to a fixed string. Note:
                    \begin{itemize}
                        \item $\phi_V(x,y_V,z_V)$ takes in three strings: $x$ (query answers to predecessor queries), $y_V$ (verifier $V$'s proof), $z_V$ (auxilliary variables introduced by Cook-Levin construction).
                        \item Without loss of generality, all $\phi_V$ throughout this construction are assumed to have the same number $N$ of variables and $M$ of clauses (via trivial padding arguments).
                    \end{itemize}
    \item (Arithmetize each clause of $\phi_V$) Let $c_{V,j}$ denote the $j$th clause of $\phi_V$, where $j\in[M]$. Arithmetize each $c_{V_j}$ via rules $\overline{x}\mapsto 1-x$ and $x\vee y\vee z\mapsto 1-xyz$ (with this order of precedence). For example,
    \begin{equation}
        (z_1\vee \overline{z_2}\vee z_3)\mapsto 1-(1-z_1)(z_2)(1-z_3).
    \end{equation}
    View the right hand side as a multi-linear polynomial $r_{V_j}:[0,1]^3\to [0,1]$.
    \item (Combine clauses of $\phi_V$) For each $\phi_V$, define polynomial $q_V:[0,1]^N\to [0,1]$ as
        $
            q_V:=\Pi_{j=1}^Mr_{V_j}.
        $
        Note $q_V$ has range $[0,1]$, but is no longer multi-linear. Also, $\phi_V(x,y_V,z_V)$ and $q_V(x,y_V,z_V)$ take in the same arguments (although for $q_V$, each coordinate of $x,y_V,z_V$ lies in $[0,1]$).
        \end{enumerate}
        \item Set $i=i+1$.
    \end{enumerate}
        \item (Combine polynomials to simulate weighting function) Substituting into \Cref{eqn:f}, define polynomial
        \begin{equation}\label{eqn:p1}
            p(x, y_{V_1},\ldots, y_{V_m},z_{V_1},\ldots, z_{V_m}):=\sum_{i=1}^mw_i \left(x_i q_{V_i}(x,y_{V_i},z_{V_i}) + \frac{(1-x_i)}{2}\right).
        \end{equation}
        Note we define the domain of $p$ as $[0,1]^{m(2N+1)}$; for brevity, let $s_i$ for $i\in[m(2N+1)]$ now denote the $i$th real parameter of $p$'s input. Observe $p$ has range  $[0,\sum_i \abs{w_i}]$, since each $q_{V_i}$ has range $[0,1]$.

         \item (Linearize the polynomial) Round $p$ to a multi-linear polynomial $\pout$ via the following iterative process, for which we define $p^{(k)}$ as the polynomial $p$ after round $k\in\set{0,\ldots, m(2N+1)}$:
        \begin{enumerate}
            \item Set $k=1$.
            \item While $k\leq m(2N+1)$, set
            \begin{equation}\label{eqn:linearize}
                p^{(k)}(s_1,\ldots, s_{m(2N+1)}) = (1-s_k)\cdot \fix(p^{(k-1)},k,0) +s_k\cdot\fix(p^{(k-1)},k,1),
            \end{equation}
            where $\fix(p^{k-1},k,b)$ is obtained by fixing $s_k$ of $p^{k-1}$ to $b\in\set{0,1}$.
            \item Set $k=k+1$.
        \end{enumerate}
        Observe that for all $k\in\set{1,\ldots, m(2N+1)}$,
        \begin{equation}
            p^{(k-1)}(s_1,\ldots, s_{k-1},b,s_{k+1},\ldots s_{m(2N+1)})=p^{(k)}(s_1,\ldots, s_{k-1},b,s_{k+1},\ldots s_{m(2N+1)})
        \end{equation}
        for any $b\in\set{0,1}$. Thus, $\pout:=p^{(m(2N+1))}$ satisfies $\pout(s)=p(s)$ for all $s\in\set{0,1}^{m(2N+1)}$. Moreover, by construction $\pout$ is multi-linear and has range $[0,\sum_i \abs{w_i}]$ (since each iteration of line 4b introduces a convex combination over local assignments).
\end{enumerate}
Finally, we assume all arithmetic operations above are represented implicitly via gates of an arithmetic circuit (required due to Step 2biii, as expanding $q_V$ explicitly in a monomial basis can result in exponentially many terms). The resulting arithmetic circuit clearly has size $\poly(n)$.

\paragraph{Correctness.} Since $\pout$ is multi-linear, it obtains\footnote{Here is a simple proof via exchange argument, for completeness: Let $x=(x_1,...,x_n)$ be a point maximizing multi-linear $f:[a,b]^n\to\R$ for arbitrary $a,b\in \R$. Assume without loss of generality $x_1 \not \in \set{a,b}$. Then, fixing $x_2,\ldots ,x_n$, the resulting function $f(x_1)$ is linear in $x_1$ by definition, and so $\max(f(a),f(b))\geq f(x_1)$. Exchanging $\argmax(f(a),f(b))$ for $x_1$ completes the claim.} its maximum on an extreme point of the compact set $[0,1]^{m(2N+1)}$, i.e.
        \begin{equation}
            \max_{s\in[0,1]^{m(2N+1)}}\pout(s)=\max_{s\in\set{0,1}^{m(2N+1)}}\pout(s).
        \end{equation}
Thus, we may restrict attention\footnote{Note the linearization of Step 4 is \emph{necessary} to obtain this statement. For example, consider an unsatisfiable $2$-SAT formula $\phi(x_1,x_2)=(x_1\vee x_2)\wedge(\overline{x_1}\vee x_2)\wedge(x_1\vee \overline{x_2})\wedge(\overline{x_1}\vee \overline{x_2})$. Let $q$ be the multi-variate polynomial produced by arithmetizing $\phi$ as in steps $2(ii)$ and $2(iii)$. Then, the maximum value of $q$ over strings is $0$, but setting each variable to $1/2$ yields value $(3/4)^4>0$. With this said, note that for $3$-SAT, since a $7/8$-approximation ratio is optimal via the PCP theorem~\cite{Ha97}, one can show via AM-GM inequality that for any unsatisfiable $\phi$, optimizing all variables of $q$ over $[0,1]$ yields value at most $(7/8)^m$ for $m$ clauses. Thus, up to inverse exponential corrections, one \emph{could} avoid the linearization step, but the tradeoff is added clutter and the need to assume $\p\neq \NP$.} to $s\in\set{0,1}^{m(2N+1)}$. But on this set, $\phi_V$ (Cook-Levin output) and $q_V$ (arithmetization of $\phi_V$) coincide. We conclude
\begin{equation}
    \max_{s\in[0,1]^{m(2N+1)}}\pout(s)= \max_{x,y_1,\ldots, y_m\in\set{0,1}^{\poly(m)}} t(x,y_1,\ldots, y_m)
\end{equation}
for $t$ from \Cref{eqn:f}. Moreover, recalling that $s=xy_{V_1}\cdots y_{V_m}z_{V_1}\cdots z_{V_m}$ (viewed as a concatenation of strings), it follows that given the optimal $s^*$, we may recover the correct \NP\ query string simply by reading off $x$.
\end{proof}

With \Cref{l:poly} in hand, the proof of \Cref{l:polycompress} now follows straightforwardly.
\begin{proof}[Proof of \Cref{l:polycompress}]
    The proof is similar to that of \Cref{thm:bsn-s} (and thus \Cref{thm:separator-tree-time}), except we need not apply the Compression Lemma (\Cref{l:compress}) in that the content of \Cref{sscn:solve-c-dag} suffices, and now we use \Cref{l:poly} to make the connection to polynomials. Specifically, let $\Pi$ be any instance of a $\pNP$ problem, and $M$ a $\pNP$ machine deciding $\Pi$. Map the NP-DAG representing $M$'s action directly (i.e. without utilizing \Cref{l:compress}) to the function $t$ in \Cref{eqn:f}. For this, the weights $w_i$ can be any $c$-admissible weighting function which satisfies the preconditions of \Cref{lem:correct-query-string}; for concreteness, choose the $2$-admissible function $\omega(v)=3^{\abs{\desc(v)}}$ from \Cref{lem:weighting-function}. Apply \Cref{l:poly} to map $t$ to polynomial $\pout(x,y_1,\ldots, y_m, z_1,\ldots, z_m)$. Since $\pout$ is efficiently evaluated on any given input, and has an optimal value $\pout(x^*,y^*_1,\ldots, y^*_m, z^*_1,\ldots, z^*_m)$ expressible using $B$ bits of precision\footnote{For an arbitrary NP-DAG, $B$ can be polynomial in the NP-DAG's size.}, a binary search using $B$ queries to an NP-oracle suffices to identify an optimal input $(x^*,y^*_1,\ldots, y^*_m, z^*_1,\ldots, z^*_m)$. By \Cref{l:poly}, one can now efficiently extract the answers to all NP queries made by $M$ (specifically, this is the string $x^*$), and thus efficiently simulate $M$ itself to decide $\Pi$.
\end{proof}

\subsection*{Acknowledgments}

We thank Eric Allender, Johannes Bausch, Stephen Piddock, James Watson, and Justin Yirka for helpful discussions. SG acknowledges funding from DFG grants 432788384 and 450041824.

\printbibliography

\end{document}